\documentclass[11pt,english]{article}
\usepackage[T1]{fontenc}
\usepackage[utf8]{inputenc}
\usepackage{babel}
\usepackage{array}
\usepackage{verbatim}
\usepackage{cprotect}
\usepackage{float}
\usepackage{booktabs}
\usepackage{bm}
\usepackage{multirow}
\usepackage{amsmath}
\usepackage{amsthm}
\usepackage{amssymb}
\usepackage{graphicx}
\usepackage{geometry}
\usepackage{setspace}
\usepackage[authoryear]{natbib}
\PassOptionsToPackage{normalem}{ulem}
\usepackage{ulem}
\usepackage[]
 {hyperref}

\makeatletter

\providecommand{\tabularnewline}{\\}

\numberwithin{equation}{section}
\theoremstyle{plain}
\newtheorem{thm}{\protect\theoremname}[section]
\theoremstyle{plain}
\newtheorem{lem}{\protect\lemmaname}[section]
\theoremstyle{plain}
\newtheorem*{cor*}{\protect\corollaryname}
\theoremstyle{plain}
\newtheorem{prop}{\protect\propositionname}[section]
\theoremstyle{plain}
\newtheorem{assumption}{\protect\assumptionname}

\usepackage{hyperref}
\usepackage{natbib} 
\usepackage{lmodern}
\usepackage[T1]{fontenc}
\usepackage{placeins}
\usepackage{setspace}
\usepackage{makecell}

\ifdefined\showcaptionsetup
 \PassOptionsToPackage{caption=false}{subfig}
\fi
\usepackage{subfig}
\makeatother

\providecommand{\assumptionname}{Assumption}
\providecommand{\corollaryname}{Corollary}
\providecommand{\lemmaname}{Lemma}
\providecommand{\propositionname}{Proposition}
\providecommand{\theoremname}{Theorem}

\begin{document}
\title{Selecting the Best Arm in One-Shot Multi-Arm RCTs:\\
The Asymptotic Minimax-Regret Decision Framework \\
for the Best-Population Selection Problem}
\author{Joonhwi Joo\thanks{Naveen Jindal School of Management, The University of Texas at Dallas.
\texttt{\protect\href{mailto:joonhwi.joo@utdallas.edu}{joonhwi.joo@utdallas.edu}}~\protect \\
I thank Jean-Pierre Dub\'{e}, Xinyao Kong, Charles Manski, Ram C.
Rao, Adam N. Smith, and seminar participants at UTD for helpful comments.}}
\date{{\small First version: August 29, 2025}\\
{\small This version: \today}\vspace{-0bp}
}
\maketitle
\begin{abstract}
We develop a frequentist decision-theoretic framework for selecting
the best arm in one-shot, multi-arm randomized controlled trials (RCTs).
Our approach characterizes the minimax-regret (MMR) optimal decision
rule for any multivariate location family reward distribution with
full support. We show that the MMR rule is deterministic, unique,
and computationally tractable. We then specialize to the case of multivariate
normal (MVN) rewards with an arbitrary covariance matrix, and establish
the local asymptotic minimaxity of a plug-in version of the rule when
only estimated means and covariances are available. This asymptotic
MMR (AMMR) procedure maps a covariance-matrix estimate directly into
decision boundaries, allowing straightforward implementation in practice.
Our analysis highlights a sharp contrast between two-arm and multi-arm
designs. With two arms, the “pick-the-winner” empirical success rule
remains MMR-optimal, regardless of the arm-specific variances. By
contrast, with three or more arms and heterogeneous variances, the
empirical success rule is no longer optimal: the MMR decision boundaries
become nonlinear and systematically penalize high-variance arms, requiring
stronger evidence to select them. Our multi-arm AMMR framework offers
a rigorous foundation that leads to practical criteria for comparing
multiple policies simultaneously.
\end{abstract}
\newpage{}

\section{Introduction}

Randomized controlled trials (RCTs) are now indispensable in data-driven
decision-making across public policy, business, and clinical trials.
Most one-shot RCTs still use two-arm “A/B test'' designs, for which
simple threshold rules are known to be decision-theoretically optimal
under usual conditions \citep[e.g.,][]{Karlin1956}. When comparisons
span more than two arms, decision-makers (DMs) typically rely on ad
hoc rules, such as pairwise tests designed for two-arm settings, multiple-hypothesis
testing procedures, or the “pick-the-winner” empirical success rule
that ignores arm-specific variance. None of these is decision-theoretically
optimal when the goal is to select the one best policy among multiple
candidates in a one-shot, multi-arm RCT. This lack of an optimal rule
for multi-arm comparison has left the two-arm ``A/B test'' designs
pervasive for one-shot RCTs, leading to inefficiencies in RCT designs
and suboptimality in the resulting decisions.

We address this inefficiency by characterizing the decision-theoretically
optimal rule for the \emph{best-population selection} problem in one-shot,
multi-arm RCTs. Our analysis is cast in the frequentist minimax-regret
(MMR) framework, where regret is the opportunity cost from not choosing
the best arm ex post, and the MMR criterion minimizes the worst-case
expected regret over the parameter space. We first characterize the
MMR-optimal rule for any multivariate location family reward distribution
with fully-supported density that shifts its location-parameter vector
but keeps the same shape and is positive everywhere. We then specialize
to multivariate normal (MVN) rewards with an arbitrary covariance
matrix. Lastly, when the mean-reward estimator is only (locally uniformly)
asymptotically normal and covariances are estimated, we show that
a plug-in version of the MVN-MMR rule is locally asymptotically minimax.
The resulting Asymptotic MMR (AMMR) procedure maps a consistent covariance-matrix
estimate directly into decision boundaries for the mean-reward estimates,
making implementation straightforward in practice.

Our results reveal a sharp and practically important distinction between
two-arm and multi-arm designs. With two arms, the empirical success
rule remains MMR-optimal, and heterogeneity in variances across arms
does not alter the decision threshold. With three or more arms and
heterogeneous variances, however, the empirical success rule is no
longer optimal. The MMR decision boundaries become nonlinear and penalize
high-variance arms, requiring stronger evidence to select them. Intuitively,
for the two-arm case, arm-specific variance is irrelevant because
the problem reduces to comparing only the realized reward difference
between the two arms. However, arm-specific variance becomes critical
once multiple options are compared simultaneously, as the winning
arm must be superior over all other arms with different arm-specific
variances.

Formally, we study the problem of selecting the best arm after observing
one noisy signal vector whose distribution belongs to a location family
with an unknown location parameter vector. We adopt the regret loss
and the frequentist minimax criterion following the MMR analyses in
two-arm setups \citep[e.g.,][]{Manski2004,Manski2016,Manski2019,Manski2021,Hirano2009,Tetenov2012,Stoye2009,Stoye2012,HuZhuBrunskillWager2024,JooForthcoming}.
Our frequentist approach to the decision problem yields a plug-in
rule that does not require specifying or estimating a subjective prior,
facilitating immediate practical applications.

Directly solving for the MMR problem beyond the two-arm setup is typically
intractable. To make progress, we recast the problem as a zero-sum
game against nature, the dual formulation of the problem following
\citet{Wald1945}. In this dual formulation, nature chooses a least-favorable
prior to which the DM responds with the Bayes rule. We prove that
nature’s least-favorable prior places exactly one support point in
each region where a particular arm is optimal. We further prove that
the corresponding Bayes rule is deterministic, unique, and therefore
unique minimax for the original MMR problem. These results drastically
simplify the numerical characterization and computation of the MMR-optimal
policy.

In applications, the exact reward distribution is rarely known. We
therefore provide a proof that the MVN-MMR rule and its regret risk
uniformly approximate their exact finite-sample counterparts when
the mean-reward estimator is locally uniformly asymptotically normal,
a mild assumption satisfied by many smooth M-estimators including
GMM and MLE, and by several shrinkage/robust estimators. Consequently,
one can plug $\sqrt{n}$-scaled mean-reward estimates into the MVN
decision boundaries of a consistent covariance estimate. This plug-in
rule leads to the AMMR rule that is optimal in the local asymptotic
sense, without requiring asymptotic efficiency. This complements the
approach of \citet{Hirano2009} and facilitates operationalization.

We solve the MVN-MMR problem numerically and plot the resulting MVN-MMR
decision boundaries. With two arms, the boundaries are linear and
coincide with the empirical success rule, regardless of the arm-specific
variances. With three or more arms and non-homoskedastic covariance
matrix, the optimal boundaries curve and shrink the region for high-variance
arms, formalizing the intuitive idea that higher arm-specific uncertainty
requires stronger evidence.

Taken together, our results demonstrate that the MMR decision theory
can directly inform regret-optimal decision-making using one-shot,
multi-arm RCTs. Therefore, our multi-arm AMMR allows researchers and
practitioners to compare multiple competing policies by designing
and implementing multi-arm RCTs within a unified decision-theoretic
framework, moving beyond the pervasive two-arm ``A/B test'' designs.

\subsection{Related Literature}

We are, to our knowledge, the first to characterize the general minimax-risk-optimal
best-population selection rule under regret loss in one-shot, multi-arm
experiments. Classic decision-theory results \citep{bahadur1950problem,bahadur1952impartial,lehmann1966theorem,eaton1967some}
show that the empirical success rule is minimax, admissible, and Bayes
under permutation invariance, a symmetry that effectively requires
homoskedastic rewards in the MVN case and certain discrete, permutation-invariant
priors. In most applications, covariance matrices are heteroskedastic,
breaking permutation invariance and eliminating the optimality properties
of the empirical success rule. Indeed, later work shows that empirical
success is not minimax under heteroskedasticity (\citealp{Dhariyal1994}
under 0-1 loss; \citealp{Masten2023} under regret loss). The minimax
rule in the general heteroskedastic case had remained unknown, with
only partial results available for Bayes rules under special priors
\citep[e.g.,][]{abughalous1995selecting}. A common workaround is
balanced design \citep[e.g.,][]{bechhofer1954single,somerville1954some},
which equalizes realized arm-specific variances via sample allocation.
However, balanced design cannot be applied to one-shot multi-arm RCTs
because it typically requires pilot experiments or sequential experimentation.

Our problem is conceptually related to best-arm identification in
pure-exploration bandits, where regret-based objectives and explore-then-commit
designs are standard \citep[e.g.,][]{Audibert2010,Jamieson2014,Garivier2016,Kaufmann2016,Grover2018,Agrawal2019,Russo2020,Komiyama2021,Howard2021,Kasy2021,AdusumilliRESforthcoming}.
However, sequential bandit experimentation is often infeasible in
RCT settings with delayed outcomes or when evaluating multiple policies
retrospectively using completed trials. Even when feasible, bandit
algorithms can be operationally demanding. Our AMMR rule for one-shot
RCTs is therefore complementary and broadly applicable.

Our multi-arm characterization nests the two-arm MMR problem of \citet{Tetenov2012}:
the optimal rule in two-arm experiments is the empirical success threshold
of zero, which our dual characterization reproduces. While directly
maximizing the regret function as in Tetenov does not scale to multiple
arms, our dual characterization of the multi-arm MMR-optimal policy
that is a deterministic Bayes rule under the least-favorable prior
restores computational tractability. Recent numerical approximation
methods such as \citet{Fernandez2025,Guggenberger2025} could further
improve computational efficiency.

Our local asymptotics and the associated plug-in approach are similar
in spirit to \citet{Hirano2009}. \citeauthor{Hirano2009} derive
the local asymptotic minimax bounds using Le Cam’s limits-of-experiments
framework and show that rules based on asymptotically efficient (“best
regular”) estimators attain those bounds, thereby justifying the plug-in
of the efficient estimators into the optimal decision rules derived
for normally distributed rewards. By contrast, our asymptotic approach
is to begin with a mean-reward estimator that is locally uniformly
asymptotically normal, and directly establish the local uniform convergence
of both the risk function and the associated MMR decision rule. As
a result, our plug-in approach can be broadly applied to smooth GMM/MLE
and many shrinkage estimators that may not be asymptotically efficient.

Other related works include \citet{Kitagawa2022}, who study least-favorable
priors when regret is nonlinearly transformed, and broader ranking/subset-selection
problems \citep[e.g.,][]{GuKoenker2023}. We also do not address covariate-based
treatment assignment or its multi-arm extensions \citep[e.g.,][]{Stoye2009,Stoye2012,Kitagawa2018,Athey2021}.

$ $

The remainder of the article is organized as follows. Section \ref{sec:Characterization-of-the}
characterizes the general DM's MMR-optimal decision rule when rewards
follow a fully-supported multivariate location family. Section \ref{sec:The-MVN-Best}
applies the general result to the MVN-distributed rewards and characterizes
the properties of MVN-MMR decision rule. Section \ref{sec:First-order-Asymptotic-Approxima}
provides the asymptotic approximation results when the rewards are
only asymptotically MVN distributed. Section \ref{sec:Applications}
provides numerical examples. Section \ref{sec:Conclusion} concludes.

\section{Characterization of the MMR-Risk-Optimal Solution for the Best-Population
Selection Problem}\label{sec:Characterization-of-the}

In this section, we characterize the DM's MMR-risk-optimal rule for
the best-population selection problem. We define the best-population
selection problem, regret loss function, and the minimax-risk solution
concept. Because directly characterizing the DM’s MMR-risk-optimal
strategy is intractable, we formulate the dual problem as a two-player
zero-sum game against nature. We show that the game has a value and
that a minimax theorem applies. Then we characterize nature's least-favorable
prior in the dual maximin problem, and solve for the DM's optimal
strategy under the least-favorable prior. The resulting DM's optimal
strategy in the dual problem coincides with the DM's minimax decision
rule in the primal problem.

\subsection{The Best-Population Selection Problem Setup}\label{subsec:Setup}

\paragraph*{Setup, Notation, and Assumptions}

Let $\mathcal{J}=\left\{ 1,2,...,J\right\} $. We use boldfaces to
denote vectors throughout. Denote the parameter set by $\Theta=\left[-B,B\right]^{J}\subset\mathbb{R}^{J}$,
a compact hypercube. $B>0$ can be arbitrarily large but finite. We
assume the true parameter vector $\bm{\theta}$ lies in the interior
of $\Theta$.

For any $\bm{\theta}\in\Theta$ let $i^{*}\left(\bm{\theta}\right)=\min\left\{ l:\theta_{l}=\max_{j\in\mathcal{J}}\theta_{j}\right\} $,
the smallest index such that the component $\theta_{l}$ attains $\max_{j\in\mathcal{J}}\left\{ \theta_{j}\right\} $.
Define $\Theta^{i}=\left\{ \bm{\theta}\in\Theta:i^{*}\left(\bm{\theta}\right)=i\right\} $.
$\Theta^{i}$ is the union of the set $\left\{ \bm{\theta}\in\Theta:-B\leq\theta_{j}<\theta_{i}\leq B\quad\forall j\neq i\right\} $
and certain boundary surfaces of ties, where $\Theta^{1}$ contains
all the tie points of the form $\max_{k}\left\{ \theta_{k}\right\} =\theta_{1}=\theta_{j}$
for $j>1$ whereas $\Theta^{J}$ does not include any ties in the
largest element. By construction, $\left\{ \Theta^{i}\right\} _{i\in\mathcal{J}}$
partitions $\Theta$, i.e., $\Theta=\bigcup_{i\in\mathcal{J}}\Theta^{i}$
and $\Theta^{i}\cap\Theta^{j}=\emptyset$ for any $i\neq j$.

We use superscript $i$ to designate that the true parameter belongs
to $\Theta^{i}$, and subscript $i$ to denote $i$-th element of
a vector. For instance, $\theta_{J}^{J}>\theta_{j}^{J}$ for all $j\neq J$
by construction because $\bm{\theta}^{J}\in\Theta^{J}$ and $\Theta^{J}$
does not include ties in the largest element. For the dual problem
introduced later, we denote by $\pi^{i}$ the prior probability of
$\Theta^{i}$, and by $\bm{\pi}=\left(\pi^{i}\right)_{i\in\mathcal{J}}$
the collection of these prior probabilities. We denote by $\pi\equiv\left\{ \pi^{i},\bm{\theta}^{i}\right\} _{i\in\mathcal{J}}$
a discrete prior distribution $\pi$ supported on $\left\{ \bm{\theta}^{j}\right\} _{j\in\mathcal{J}}$.

Let $\left\{ P_{\bm{\theta}}\right\} _{\bm{\theta}\in\Theta}$ be
a multivariate location family. We assume each $P_{\bm{\theta}}$
is a probability measure on $\mathbb{R}^{J}$ dominated by the Lebesgue
measure $\mu$, the densities $p_{\bm{\theta}}$ are continuous and
fully supported on $\mathbb{R}^{J}$, and $p_{\bm{\theta}}\left(\mathbf{x}\right)\neq p_{\bm{\theta}'}\left(\mathbf{x}\right)$
$\mu$-a.e. whenever $\bm{\theta}\neq\bm{\theta}'$. When $\bm{\theta}^{i}\in\Theta^{i}$
for each $i$, we use the shorthand $p^{i}=p_{\bm{\theta}^{i}}$.

The \emph{best-population selection} problem (the \emph{selection}
problem) is finding the largest element (arm) of the location-parameter
vector $\bm{\theta}$ after observing $\mathbf{x}\sim P_{\bm{\theta}}$.
A decision rule for the selection problem is a (measurable) mapping
$\bm{\phi}:\mathbb{R}^{J}\rightarrow\bm{\Delta}$, where the codomain
$\bm{\Delta}=\left\{ \bm{\phi}\in\mathbb{R}^{J}:\sum_{j=1}^{J}\phi_{j}=1\right\} $
is the probability simplex of $\mathbb{R}^{J}$. $\phi_{j}\left(\mathbf{x}\right)$
is the probability of selecting arm $j$ when an observation $\mathbf{x}$
is realized. Note that the decision rule $\bm{\phi}$ permits randomization
as $\bm{\phi}\in\text{int}\left(\bm{\Delta}\right)$ is possible.
Note, a nonrandomized decision rule can take values only among the
vertices $\left\{ \mathbf{e}_{i}\right\} _{i\in\mathcal{J}}$ of the
simplex $\bm{\Delta}$, where $\mathbf{e}_{i}$ is the $i$-th unit
vector.

\paragraph*{Loss, Risk, and the Solution Concept}

The regret loss with the realization of $\mathbf{x}$ under $\bm{\phi}$
is defined as:
\begin{equation}
L\left(\bm{\phi}\left(\mathbf{x}\right),\bm{\theta}\right)=\max_{k\in\mathcal{J}}\left\{ \theta_{k}\right\} -\sum_{j\in\mathcal{J}}\theta_{j}\phi_{j}\left(\mathbf{x}\right)=\sum_{j\in\mathcal{J}}\left(\max_{k\in\mathcal{J}}\left\{ \theta_{k}\right\} -\theta_{j}\right)\phi_{j}\left(\mathbf{x}\right).\label{eq:regret_loss}
\end{equation}
The risk of a decision rule $\bm{\phi}$ is the expected loss taking
$\bm{\theta}$ as given, where the expectation is taken with respect
to the measure $dP_{\bm{\theta}}\left(\mathbf{x}\right)=p_{\bm{\theta}}\left(\mathbf{x}\right)d\mu\left(\mathbf{x}\right)$:
\begin{align}
R\left(\bm{\theta},\bm{\phi}\right) & =\int L\left(\bm{\phi}\left(\mathbf{x}\right),\bm{\theta}\right)p_{\bm{\theta}}\left(\mathbf{x}\right)d\mu\left(\mathbf{x}\right)=\max_{k\in\mathcal{J}}\left\{ \theta_{k}\right\} -\sum_{j}\theta_{j}\mathbb{E}_{P_{\bm{\theta}}}\left[\phi_{j}\left(\mathbf{x}\right)\right].\label{eq:risk}
\end{align}
The regret risk is continuous in $\bm{\theta}$ for any decision
rule $\bm{\phi}$ under our setup (see Lemma \ref{lem:A2}).

The DM's objective is to find the MMR-risk-optimal decision rule $\bm{\delta}$
that minimizes the worst-case risk over the parameter space $\Theta$:
\begin{equation}
\bm{\delta}=\arg\inf_{\bm{\phi}}\sup_{\bm{\theta}\in\Theta}R\left(\bm{\theta},\bm{\phi}\right),\label{eq:minimax_risk}
\end{equation}
where $R\left(\bm{\theta},\bm{\phi}\right)$ is defined by (\ref{eq:regret_loss})
and (\ref{eq:risk}). 

Directly characterizing the minimax-risk decision rule is intractable
except for a few special cases \citep[e.g.,][]{Manski2004,Stoye2009,Stoye2012,Kitagawa2022}.
Therefore, in what follows, we turn to the “dual” maximin problem
of a two-player zero-sum game between the DM and nature to characterize
the DM’s optimal decision rule.

\subsection{Characterization of Nature's Least-Favorable Prior in the Dual Maximin
Problem}\label{subsec:Properties-of-the}

In this subsection, we first establish the existence of nature’s \emph{least-favorable
prior} $\pi$ for the selection problem under regret loss, which is
supported by at most $J$ distinct points in $\mathbb{R}^{J}$ (Theorem
\ref{thm:(Minimax-Theorem-for}). Next, we characterize the form of
the Bayes rule under this least-favorable prior, which is deterministic
up to tie-breaking (Lemma \ref{lem:(Pointwise-Bayes-Action}). Then
we establish that the support points of the least-favorable prior
are exactly $J$ distinct points, with $\bm{\theta}^{i}\in\Theta^{i}$
for each $i$ and the corresponding prior probability vector $\left(\pi^{i}\right)_{i\in\mathcal{J}}\in\text{int}\left(\bm{\Delta}\right)$
(Theorem \ref{thm:thm_2}). Taken together, the results established
in this subsection simplify the search for nature's least-favorable
prior in the dual problem drastically.
\begin{thm}[Minimax Theorem for the Best-Population Selection Problem]
\label{thm:(Minimax-Theorem-for} The following hold for the selection
problem:

(i)
\begin{align}
\inf_{\bm{\phi}}\sup_{\bm{\theta}\in\Theta}R\left(\bm{\theta},\bm{\phi}\right) & =\sup_{\tilde{\pi}}\inf_{\bm{\phi}}\mathbb{E}_{\tilde{\pi}}\left[R\left(\bm{\theta},\bm{\phi}\right)\right]\equiv V\in\left[0,2B\right],\label{eq:minimax_value}
\end{align}
where $\sup_{\tilde{\pi}}$ is over all Borel probability measures
on $\Theta$.

(ii) There exists a least-favorable prior $\pi$ supported on at most
$J$ distinct support points of $\Theta$.

(iii) Suppose $\bm{\delta}$ is a Bayes rule with respect to a least-favorable
prior $\pi$, i.e., $\bm{\delta}=\arg\inf_{\bm{\phi}}\mathbb{E}_{\pi}\left[R\left(\bm{\theta},\bm{\phi}\right)\right]$.
Then, $\bm{\delta}$ is the DM's minimax-risk decision rule, i.e.,
$\bm{\delta}=\arg\inf_{\bm{\phi}}\sup_{\bm{\theta}\in\Theta}R\left(\bm{\theta},\bm{\phi}\right)$.

(iv) Consider the least-favorable prior $\pi$ supported on at most
$J$ distinct points on $\Theta$. Let $\bm{\delta}$ be a Bayes rule
with respect to $\pi$. Then, for all $\bm{\theta}\in\text{supp}\left(\pi\right)$,
$R\left(\bm{\theta},\bm{\delta}\right)=V$.
\end{thm}
\begin{proof}
See Appendix \ref{subsec:Thm1_proof}.
\end{proof}
Parts (i) and (iii) follow from the standard minimax theorem for statistical
decision problems with finite action sets under bounded loss and dominated
experiments (e.g., \citealp[ch.2-3]{Blackwell1954};\citealp[ch.5]{Berger1985};\citealp[ch.3]{Liese2008}).

Part (ii) asserts the existence of a least-favorable prior supported
on at most $J$ distinct points. The geometric intuition behind this
is as follows. Define the $J$-dimensional risk vector $\mathbf{R}\left(\bm{\theta}\right)\equiv\left(\max_{k\in\mathcal{J}}\left\{ \theta_{k}\right\} -\theta_{j}\right)_{j\in\mathcal{J}}=\left(R\left(\bm{\theta},\mathbf{e}_{j}\right)\right)_{j\in\mathcal{J}}$
and the risk set $S\equiv\left\{ \left(\max_{k\in\mathcal{J}}\left\{ \theta_{k}\right\} -\theta_{j}\right)_{j\in\mathcal{J}}:\bm{\theta}\in\Theta\right\} $.
The convex hull $\text{co}\left(S\right)=\left\{ \int_{\Theta}\mathbf{R}\left(\bm{\theta}\right)d\pi\left(\bm{\theta}\right):\pi\text{ }\text{is Borel prob. measure over }\Theta\right\} $
collects the risk vectors induced by priors. For any mixed action
$\mathbf{a}\in\bm{\Delta}$, nature's worst-case risk is $\sup_{\mathbf{u}\in\text{co}\left(S\right)}\mathbf{a}^{\top}\mathbf{u}$,
so the value is $V_{S}\equiv\min_{\mathbf{a}\in\bm{\Delta}}\sup_{\mathbf{u}\in\text{co}\left(S\right)}\mathbf{a}^{\top}\mathbf{u}=\sup_{\mathbf{u}\in\text{co}\left(S\right)}\min_{\mathbf{a}\in\bm{\Delta}}\mathbf{a}^{\top}\mathbf{u}$.
The maximizer $\mathbf{u}^{*}$ lies on a supporting hyperplane $\left\{ \mathbf{u}\in\text{co}\left(S\right):\mathbf{a}^{*\top}\mathbf{u}=V_{S}\right\} $,
a $J-1$ dimensional affine subspace. Hence, Caratheodory's theorem
yields a representation with at most $J$ support points. While the
minimax/Bayes rule is a data-dependent Markov kernel $\bm{\phi}:\mathbb{R}^{J}\rightarrow\bm{\Delta}$,
the forgoing $S$-game geometry is used to control nature's side of
the problem, justifying the finite-support property; it does not restrict
$\bm{\phi}$ to ignore the data.

Part (iv) can either be imposed as a constraint in finding the least-favorable
prior numerically, or can be checked after a candidate least-favorable
prior is found. We come back to this in Section \ref{sec:Applications}.
We also note that Theorem \ref{thm:(Minimax-Theorem-for} (iv) under
a slightly different set of assumptions can be found in Corollary
to Theorem 2.2 of \citet{Kempthorne1987}.

We next characterize the form of the Bayes rule under a finite prior,
which is deterministic up to tie-breaking that can be chosen arbitrarily.
For the following lemma and the proof, we use the superscript $\left(k\right)$
to denote the $k$-th element of the support point set $\left\{ \bm{\theta}^{\left(k\right)}\right\} _{k=1}^{m}$
and the prior probability set $\left\{ \pi^{\left(k\right)}\right\} _{k=1}^{m}$,
where $\bm{\theta}^{\left(k\right)}$ may not necessarily belong to
$\Theta^{k}$.
\begin{lem}[Pointwise Bayes Action for Any Finite Prior]
\label{lem:(Pointwise-Bayes-Action} Let $\pi$ be supported on finitely
many support points $\left\{ \bm{\theta}^{\left(k\right)}\right\} _{k=1}^{m}\subset\Theta$
with the corresponding masses $\left\{ \pi^{\left(k\right)}\right\} _{k=1}^{m}$,
where $1\leq m\leq J$. For each $\mathbf{x}\in\mathbb{R}^{J}$ and
$i\in\mathcal{J}$, define 
\begin{equation}
h_{i}\left(\mathbf{x}\right)\equiv\sum_{k=1}^{m}\pi^{\left(k\right)}p_{\bm{\theta}^{\left(k\right)}}\left(\mathbf{x}\right)\theta_{i}^{\left(k\right)}.\label{eq:Bayes_score}
\end{equation}
Then, any Bayes rule $\bm{\delta}$ can be chosen pointwise as
\begin{equation}
\delta_{i}\left(\mathbf{x}\right)=1\left(i=\arg\max_{j\in\mathcal{J}}h_{j}\left(\mathbf{x}\right)\right),\label{eq:argmax_delta_Bayes}
\end{equation}
with any measurable tie-breaking rule.
\end{lem}
\begin{proof}
See Appendix \ref{subsec:Lem1_proof}.
\end{proof}
The following Theorem \ref{thm:thm_2} establishes that nature’s least-favorable
prior is supported by exactly $J$ distinct points $\left\{ \bm{\theta}^{i}\right\} _{i\in\mathcal{J}}$
with $\bm{\theta}^{i}\in\Theta^{i}$, and with strictly positive prior
probabilities for each arm. In turn, this implies the prior probabilities
$\bm{\pi}\in\text{int}\left(\bm{\Delta}\right)$.
\begin{thm}[Support Point Characterization of the Least-Favorable Prior]
\label{thm:thm_2}Suppose $\pi$ is a least-favorable prior supported
on at most $J$ distinct points in $\Theta$. Let $\bm{\delta}$ be
a minimax rule under $\pi$. Then, the following hold:

(i) $\mu\left(\left\{ \mathbf{x}:\delta_{i}\left(\mathbf{x}\right)>0\right\} \right)>0$
for each $i\in\mathcal{J}$, implying $\delta_{i}\left(\mathbf{x}\right)$
cannot be identically 0 or $1$ $\mu$-a.e.

(ii) Any least-favorable prior with at most $J$ distinct support
points must be supported by exactly $J$ distinct support points $\left\{ \bm{\theta}^{i}\right\} _{i\in\mathcal{J}}$
where each $\bm{\theta}^{i}\in\Theta^{i}$ with the associated prior
probability $\pi^{i}>0$.
\end{thm}
\begin{proof}
See Appendix \ref{subsec:Thm2_proof}.
\end{proof}

\subsection{The DM's Unique MMR Decision Rule}\label{subsec:Essential-Complete-Class}

In Theorem \ref{thm:(Characterization-of-the}, we characterize the
DM's Bayes decision rule $\bm{\delta}$ under the form of nature’s
least-favorable prior established in Theorems \ref{thm:(Minimax-Theorem-for}-\ref{thm:thm_2},
supported on $J$ distinct points with $\bm{\theta}^{j}\in\Theta^{j}$.
Under a mild regularity condition on the multivariate location-family
density $p_{\bm{\theta}^{k}}=p^{k}\left(\mathbf{x}\right)$, $\bm{\delta}$
is the unique Bayes rule and the unique minimax rule, implying that
$\bm{\delta}$ is admissible. The unique MMR-risk rule $\bm{\delta}$
is deterministic up to a Lebesgue-null set, characterized by simple
pointwise comparisons of weighted densities.
\begin{thm}[DM's MMR-Risk Strategy]
\label{thm:(Characterization-of-the}Suppose $\pi=\left(\pi^{j},\bm{\theta}^{j}\right)_{j\in\mathcal{J}}$
is nature's least-favorable prior characterized in Theorem \ref{thm:thm_2}.
Let $\bm{\delta}$ be a decision rule defined by: 
\begin{equation}
\delta_{i}\left(\mathbf{x}\right)=\begin{cases}
1 & \text{if }\sum_{k=1}^{J}\theta_{i}^{k}\pi^{k}p^{k}\left(\mathbf{x}\right)>\sum_{k=1}^{J}\theta_{j}^{k}\pi^{k}p^{k}\left(\mathbf{x}\right)\:\forall j\neq i\\
0 & \text{otherwise},
\end{cases}\label{eq:Bayes_general}
\end{equation}
with any measurable tie-breaking rule. Then,

(i) $\bm{\delta}$ is Bayes with respect to $\pi$ up to Lebesgue
null set, implying $\bm{\delta}$ is minimax.

\noindent{}Define the Bayes-score difference $h_{ij}\left(\mathbf{x}\right)=\sum_{k=1}^{J}\pi^{k}p^{k}\left(\mathbf{x}\right)\left(\theta_{i}^{k}-\theta_{j}^{k}\right)$.
If, in addition, $\mu\left(\left\{ \mathbf{x}:h_{ij}\left(\mathbf{x}\right)=0\right\} \right)=0$
for any $i,j$ pair, then statements (ii)-(iv) hold:

(ii) $\bm{\delta}$ is unique Bayes with respect to $\pi$ up to Lebesgue
null set.

(iii) $\bm{\delta}$ is unique minimax for the DM's problem up to
Lebesgue null set.

(iv) $\bm{\delta}$ is admissible.
\end{thm}
\begin{proof}
See Appendix \ref{subsec:Thm3_proof}.
\end{proof}
The additional condition $\mu\left(\left\{ \mathbf{x}:h_{ij}\left(\mathbf{x}\right)=0\right\} \right)=0$
imposed for (ii)-(iv) is very mild. This is imposed to rule out pathological
behavior of the density $p^{k}\left(\cdot\right)$ having a strictly
positive flat region. Sufficient conditions include $p^{k}\left(\cdot\right)$
being real analytic or being $C^{1}$ with $\nabla h_{ij}\left(\mathbf{x}\right)\neq\mathbf{0}$
whenever $h_{ij}\left(\mathbf{x}\right)=0$. Most importantly, the
condition is satisfied for the MVN density that we analyze in detail
below in Sections \ref{sec:The-MVN-Best}-\ref{sec:Applications},
because each $\pi^{k}\left(\theta_{i}^{k}-\theta_{j}^{k}\right)p^{k}\left(\mathbf{x}\right)$
is real analytic, $\left\{ p^{k}\left(\mathbf{x}\right)\right\} _{k\in\mathcal{J}}$
are linearly independent, and $\pi^{k}\left(\theta_{i}^{k}-\theta_{j}^{k}\right)$
is not identically zero for all $k$. We formalize the MVN case in
the following Corollary.
\begin{cor*}
Suppose $\mathbf{x}\sim\mathcal{N}\left(\bm{\theta},\bm{\Sigma}\right)$,
where $\bm{\theta}$ is unknown but $\bm{\Sigma}$ is known to the
DM. Then, $\bm{\delta}$, defined in (\ref{eq:Bayes_general}) with
respect to the nature's least-favorable prior $\pi$, is unique Bayes
and unique minimax up to Lebesgue null set, and admissible.
\end{cor*}

\subsection{Location Invariance of the Problem and the Decision Rule}\label{subsec:Location-Invariance-of}

$\left\{ p_{\bm{\theta}}\right\} _{\bm{\theta}\in\Theta}$ is a multivariate
location family and the resulting decision problem is location invariant.
Specifically, $\forall\mathbf{x},\bm{\theta}\in\mathbb{R}^{J}$ and
$\forall t\in\mathbb{R}$, let $\mathbf{x}'=\mathbf{x}-t\mathbf{1}_{J}$
and $\bm{\theta}'=\bm{\theta}-t\mathbf{1}_{J}$. Then, $p_{\bm{\theta}}\left(\mathbf{x}\right)=p_{\mathbf{0}}\left(\mathbf{x}-\bm{\theta}\right)=p_{\bm{\theta}'}\left(\mathbf{x}'\right)$,
$\bm{\delta}\left(\mathbf{x}'\right)=\bm{\delta}\left(\mathbf{x}\right)$,
$L\left(\bm{\delta}\left(\mathbf{x}'\right),\bm{\theta}'\right)=L\left(\bm{\delta}\left(\mathbf{x}\right),\bm{\theta}\right)$,
and $R\left(\bm{\theta}',\bm{\delta}\right)=R\left(\bm{\theta},\bm{\delta}\right)$.\cprotect\footnote{The parameter space $\Theta=\left[-B,B\right]^{J}$ is a compact set
in $\mathbb{R}^{J}$. Therefore, technically, the transformation $\bm{\theta}'=\bm{\theta}-t\mathbf{1}$
is not a location group defined on $\Theta$. At the cost of complicating
the notations, we may define the extended parameter space $\tilde{\Theta}=\mathbb{R}^{J}$
and define the value of the loss function on $\tilde{\Theta}\backslash\Theta$
as $2B$%
. However, this upper bound is never attained under the DM's MMR-risk
rule as long as $B$ is large enough.}

This implies that if $\left(\pi^{k},\bm{\theta}^{k}\right)_{k\in\mathcal{J}}$
is a least-favorable prior, then $\left(\pi^{k},\bm{\theta}^{k}-t\mathbf{1}_{J}\right)_{k\in\mathcal{J}}$
is also a least-favorable prior, i.e., the support points can slide
along the subspace $\left\{ t\mathbf{1}_{J}:t\in\mathbb{R}\right\} $.
To pin down the support points during implementation, either $\sum_{i=1}^{J}\theta_{i}^{j}=0$
or $\theta_{J}^{j}=0$ can be imposed for some fixed $j$.

\section{The MVN Best-Population-Selection Problem}\label{sec:The-MVN-Best}

In this section, we apply the results of Section \ref{sec:Characterization-of-the}
to the important special case of MVN-distributed rewards with an arbitrary
positive definite (p.d.) covariance matrix. We establish two main
results for the MVN selection problem in this section. First result,
Proposition \ref{prop:Scaling}, is the $\sqrt{n}$-scaling of the
value and the minimax decision rule. Second result, Proposition \ref{prop:Prop2},
is the continuity of the MVN-MMR decision rule and value in the covariance
matrix, respectively. These results will be used to establish the
large-sample approximation results of the exact MMR problem in the
subsequent section.

The definition of the MVN density with parameters $\left(\bm{\theta},\bm{\Sigma}\right)$
is:
\begin{equation}
p_{\bm{\theta},\bm{\Sigma}}\left(\mathbf{x}\right)=\left(2\pi\right)^{-\frac{J}{2}}\left|\bm{\Sigma}\right|^{-\frac{1}{2}}\exp\left(-\frac{1}{2}\left(\mathbf{x}-\bm{\theta}\right)^{\top}\bm{\Sigma}^{-1}\left(\mathbf{x}-\bm{\theta}\right)\right).\label{eq:MVN_density}
\end{equation}
In the MVN selection problem, the p.d. covariance matrix $\bm{\Sigma}$
is known to the DM whereas the location-parameter vector $\bm{\theta}$
is unknown. The goal of the MVN selection problem is to select the
largest element of $\bm{\theta}$ after observing $\mathbf{x}\sim\mathcal{N}\left(\bm{\theta},\bm{\Sigma}\right)$. 

In the remainder of the paper, we index the objects of interest by
$\bm{\Sigma}$ when necessary: we denote by $\pi_{\bm{\Sigma}}=\left(\pi_{\bm{\Sigma}}^{k},\bm{\theta}_{\bm{\Sigma}}^{k}\right)_{k\in\mathcal{J}}$
the least-favorable prior under $\bm{\Sigma}$, by $R_{\bm{\Sigma}}\left(\cdot,\cdot\right)$
the Gaussian risk function, by $\bm{\delta}_{\bm{\Sigma}}$ the associated
a.e.-unique MMR rule, and by $V_{\bm{\Sigma}}$ the associated value.

The following Proposition \ref{prop:Scaling} establishes the $\sqrt{n}$-scaling
of the MVN selection problem. For Proposition \ref{prop:Scaling}
and its proof, we use subscript $\left[\cdot\right]$ to denote the
dependence of the covariance matrix on the sample size. Specifically,
let $\bm{\Sigma}_{\left[1\right]}$ be a positive-definite covariance
matrix for $n=1$, and for any $n\geq1$, denote by $\bm{\Sigma}_{\left[n\right]}=\frac{1}{n}\bm{\Sigma}_{\left[1\right]}$.
Let $\Theta_{\bm{\Sigma}_{\left[1\right]}}=\Theta=\left[-B,B\right]^{J}$
be the parameter space for $n=1$.
\begin{prop}[Scaling of the MVN Selection Problem]
\label{prop:Scaling}The following hold:

(i) Risk and parameter-space scaling: $\sqrt{n}R_{\bm{\Sigma}_{\left[n\right]}}\left(\frac{1}{\sqrt{n}}\bm{\theta},\bm{\varphi}\right)=R_{\bm{\Sigma}_{\left[1\right]}}\left(\bm{\theta},\bm{\phi}\right)$,
where $\bm{\varphi}\left(\frac{1}{\sqrt{n}}\mathbf{x}\right)=\bm{\phi}\left(\mathbf{x}\right)$
on $\Theta_{\bm{\Sigma}_{\left[n\right]}}=\frac{1}{\sqrt{n}}\Theta_{\bm{\Sigma}_{\left[1\right]}}$.

(ii) Value scaling: $V_{\bm{\Sigma}_{\left[n\right]}}=\frac{1}{\sqrt{n}}V_{\bm{\Sigma}_{\left[1\right]}}$.

(iii) The least-favorable prior support scaling: If $\pi_{\bm{\Sigma}\left[1\right]}=\left(\pi^{k},\bm{\theta}_{\bm{\Sigma}_{\left[1\right]}}^{k}\right)_{k\in\mathcal{J}}$
is least-favorable under $\bm{\Sigma}_{\left[1\right]}$, $\pi_{\bm{\Sigma}_{\left[n\right]}}=\left(\pi^{k},\bm{\theta}_{\bm{\Sigma}_{\left[n\right]}}^{k}\right)_{k\in\mathcal{J}}=\left(\pi^{k},\frac{1}{\sqrt{n}}\bm{\theta}_{\bm{\Sigma}_{\left[1\right]}}^{k}\right)_{k\in\mathcal{J}}$
is least-favorable under $\bm{\Sigma}_{\left[n\right]}$.

(iv) The MMR rule scaling: $\bm{\delta}_{\bm{\Sigma}_{\left[n\right]}}\left(\frac{1}{\sqrt{n}}\mathbf{x}\right)=\bm{\delta}_{\bm{\Sigma}_{\left[1\right]}}\left(\mathbf{x}\right)$.
\end{prop}
\begin{proof}
See Appendix \ref{subsec:Prop1_proof}.
\end{proof}
Let $\mathcal{S}=\left\{ \bm{\Sigma}\in\mathbb{S}_{++}^{J}:\text{\underbar{\ensuremath{\lambda}}}\mathbf{I}_{J}\preceq\bm{\Sigma}\preceq\bar{\lambda}\mathbf{I}_{J}\right\} $
be a compact set of p.d. covariance matrix for some $0<\underbar{\ensuremath{\lambda}}<\bar{\lambda}<\infty$,
where $\mathbf{A}\preceq\mathbf{B}$ means $\mathbf{B}-\mathbf{A}$
is p.s.d. and $\mathbf{I}_{J}$ is the $J$-dimensional identity matrix.
$\underbar{\ensuremath{\lambda}}$ ($\bar{\lambda}$) can be arbitrarily
small (large) but has to be finite. For the remainder of the paper,
we assume the covariance matrix $\bm{\Sigma}\in\mathcal{S}$. The
following Proposition \ref{prop:Prop2} establishes continuity of
the values $V_{\bm{\Sigma}}$ and the MMR rule $\bm{\delta}_{\bm{\Sigma}}\left(\mathbf{x}\right)$
in $\bm{\Sigma}$.
\begin{prop}[Continuity of the MVN-MMR Decision Rule and the Value in the Covariance
Matrix]
\label{prop:Prop2}Fix a compact $\mathcal{S}$ as in the above.
For each $\bm{\Sigma}\in\mathcal{S}$, let $\bm{\delta}_{\bm{\Sigma}}$
be the a.e.-unique MVN-MMR rule. Then, for every sequence $\bm{\Sigma}_{n}\rightarrow\bm{\Sigma}$
in $\mathcal{S}$, the following hold:

(i) $\bm{\delta}_{\bm{\Sigma}_{n}}\left(\mathbf{x}\right)\rightarrow\bm{\delta}_{\bm{\Sigma}}\left(\mathbf{x}\right)$
for Lebesgue-a.e. $\mathbf{x}\in\mathbb{R}^{J}$; equivalently, the
map $\bm{\Sigma}\mapsto\bm{\delta}_{\bm{\Sigma}}\left(\mathbf{x}\right)$
is continuous on $\mathcal{S}$ for a.e.-$\mathbf{x}$.

(ii) $\sup_{\bm{\theta}\in\Theta}R_{\bm{\Sigma}_{n}}\left(\bm{\theta},\bm{\delta}_{\bm{\Sigma}_{n}}\right)\rightarrow\sup_{\bm{\theta}\in\Theta}R_{\bm{\Sigma}}\left(\bm{\theta},\bm{\delta}_{\bm{\Sigma}}\right)=V_{\bm{\Sigma}}$;
equivalently, the map $\bm{\Sigma}\mapsto V_{\bm{\Sigma}}$ is continuous
on $\mathcal{S}$.

(iii) For any compact $\tilde{\Theta}\subset\mathbb{R}^{J}$ and for
$\mathcal{S}$,
\begin{eqnarray*}
 &  & \sup_{\bm{\Lambda}\in\mathcal{S}}\sup_{\bm{\theta}\in\tilde{\Theta}}\left\Vert \bm{\delta}_{\bm{\Sigma}_{n}}-\bm{\delta}_{\bm{\Sigma}}\right\Vert _{L^{1}\left(P_{\bm{\theta},\bm{\Lambda}}\right)}\\
 & = & \sup_{\bm{\Lambda}\in\mathcal{S}}\sup_{\bm{\theta}\in\tilde{\Theta}}\int\sum_{i\in\mathcal{J}}\left|\delta_{\bm{\Sigma}_{n},i}\left(\mathbf{x}\right)-\delta_{\bm{\Sigma},i}\left(\mathbf{x}\right)\right|p_{\bm{\theta},\bm{\Lambda}}\left(\mathbf{x}\right)d\mu\left(\mathbf{x}\right)\\
 & \rightarrow & 0\qquad\text{as }n\rightarrow\infty.
\end{eqnarray*}
\end{prop}
\begin{proof}
See Appendix \ref{subsec:Proof-of-Proposition}.
\end{proof}

\section{Asymptotic Approximation of the Risk Function and the MMR Decision
Rule}\label{sec:First-order-Asymptotic-Approxima}

In this section, we establish that the finite-sample regret function
and the associated MMR-risk-optimal decision rule can be asymptotically
approximated by the MVN counterpart studied in Section \ref{sec:The-MVN-Best}.
We assume that the estimator $\hat{\bm{\theta}}_{n}$ for $\bm{\theta}$
is $\sqrt{n}$-consistent and locally uniformly asymptotically Normal,
and a consistent estimator $\hat{\bm{\Sigma}}$ for the covariance
matrix $\bm{\Sigma}$ is available. Then, Theorems \ref{thm:(Uniform-Convergence-of}-\ref{thm:(Plug-in-of-the}
establish that $\hat{\bm{\theta}}_{n}$ can be plugged into the MMR
decision rule $\bm{\delta}\left(\cdot\right)$ to select the best
arm as if $\hat{\bm{\theta}}_{n}\sim\mathcal{N}\left(\bm{\theta},\frac{1}{n}\hat{\bm{\Sigma}}\right)$
when $n$ is large.

\subsection{Local Uniform CLT of the Reward Estimator and Scaling of the Measure}

We begin this subsection by imposing a regularity assumption on the
decision rule $\bm{\phi}$ being considered.

\begingroup
\renewcommand{\theassumption}{C}
\begin{assumption}[$\mu$-a.e. Continuity of the Decision Rule]
\label{assu:For-each-,}For each $\bm{\theta}\in\Theta$, the function
$\mathbf{x}\mapsto L\left(\bm{\phi}\left(\mathbf{x}\right),\bm{\theta}\right)$
is continuous $\mu$-a.e.
\end{assumption}
\endgroup

Assumption \ref{assu:For-each-,} is a restriction imposed on the
decision rule $\bm{\phi}$ to rule out pathological behavior, such
as $\bm{\phi}$ being discontinuous on sets of positive Lebesgue measure.
This assumption does not rule out randomized decision rules a priori.
Importantly, we note that the DM's minimax decision rule characterized
in Theorem \ref{thm:(Characterization-of-the} satisfies this assumption.

Next, we introduce the local uniform CLT assumption and scaling of
the problem. We assume the true parameter $\bm{\theta}$ lies in the
interior of $\Theta$. For every $\bm{\theta}\in\text{int}\left(\Theta\right)$
and for any local parameter vector $\mathbf{h}$ such that $\left\Vert \mathbf{h}\right\Vert \leq H<\infty$,
we consider a locally shifted parameter sequence $\bm{\theta}+\frac{\mathbf{h}}{\sqrt{n}}$
throughout.

\begingroup
\renewcommand{\theassumption}{LUCLT}
\begin{assumption}[Local Uniform CLT of the Reward Estimator]
\label{assu:For-every-}$\forall\bm{\theta}\in\text{int}\left(\Theta\right)$,
$\forall H\in\left(0,\infty\right)$ and $\forall\mathbf{h}$ such
that $\left\Vert \mathbf{h}\right\Vert \leq H$, $\sqrt{n}\left(\hat{\bm{\theta}}_{n}-\left(\bm{\theta}+\frac{\mathbf{h}}{\sqrt{n}}\right)\right)\stackrel{\bm{\theta}+\frac{\mathbf{h}}{\sqrt{n}}}{\rightsquigarrow}\mathcal{N}\left(\mathbf{0},\bm{\Sigma}\right)$
uniformly in $\mathbf{h}$.
\end{assumption}
\endgroup

Assumption \ref{assu:For-every-} says the shifted and scaled reward
estimator $\sqrt{n}\left(\hat{\bm{\theta}}_{n}-\bm{\theta}\right)$
is asymptotically normal uniformly in $\bm{\theta}$ over local neighborhoods
of radii shrinking at the $\frac{1}{\sqrt{n}}$ rate. Heuristically,
the assumption can be understood as $\sqrt{n}\left(\hat{\bm{\theta}}_{n}-\bm{\theta}\right)\stackrel{\mathbf{h}}{\rightsquigarrow}\mathcal{N}\left(\mathbf{h},\bm{\Sigma}\right)$%
{} uniformly in $\mathbf{h}$ such that $\left\Vert \mathbf{h}\right\Vert \leq H<\infty$
for some $H>0$. When $\mathbf{h}=\mathbf{0}$, it reduces to the
usual pointwise weak convergence in $\bm{\theta}$: $\mathbf{u}_{\bm{\theta},n}\equiv\sqrt{n}\left(\hat{\bm{\theta}}_{n}-\bm{\theta}\right)\stackrel{}{\rightsquigarrow}\mathcal{N}\left(\mathbf{0},\bm{\Sigma}\right)\equiv\mathbf{u}_{\bm{\theta}}$.
This assumption is weaker than imposing the uniform CLT over the original
parameter space $\Theta$ and it is the minimal condition required
for the usual asymptotic power comparison to be valid \citep[e.g.,][ch. 14.1]{vanderVaart1998}.
Regular (i.e., locally asymptotically shift-equivariant) estimators
satisfy this assumption, including but not limited to the general
class of MLE and GMM estimators under suitable regularity conditions
\citep[e.g.,][]{Newey1994,vanderVaart1998}.

In the following, we localize at $\bm{\theta}+\frac{\mathbf{h}}{\sqrt{n}}$
centered at $\bm{\theta}=t\mathbf{1}_{J}$, and set $t=0$ WLOG by
invoking the location invariance discussed in Section \ref{subsec:Location-Invariance-of}.
$\bm{\theta}=t\mathbf{1}_{J}$ is where the competing arms are hardest
to distinguish; if we localize at any other point $\bm{\theta}\neq t\mathbf{1}_{j}$,
the arm that is best at $\bm{\theta}$ will be best for all the local
alternatives $\bm{\theta}+\frac{\mathbf{h}}{\sqrt{n}}$ as the sample
size grows, so the selection problem becomes trivial. The same localization
is standard for the two-arm problem (e.g., \citealp{Hirano2009};
\citealp[sec 14.1]{vanderVaart1998}; \citealp[sec. 8.9]{Liese2008}),
which we are adapting to the multi-arm best-population selection problem.

Next, we scale the estimator $\hat{\bm{\theta}}_{n}$ under the local
parameter sequence $\bm{\theta}_{n}\left(\bm{\vartheta}\right)\equiv\frac{\bm{\vartheta}}{\sqrt{n}}$.
Let 
\[
\mathbf{y}_{n}=\sqrt{n}\hat{\bm{\theta}}_{n}\quad\text{and}\quad\mathbf{u}_{n}\left(\bm{\vartheta}\right)=\mathbf{y}_{n}-\bm{\vartheta}=\sqrt{n}\left(\hat{\bm{\theta}}_{n}-\bm{\theta}_{n}\left(\bm{\vartheta}\right)\right),
\]
and we omit the argument $\left(\bm{\vartheta}\right)$ when it is
obvious from context. Denote by $P_{\bm{\vartheta},n}$ the law of
$\mathbf{y}_{n}$, by $Q_{\bm{\vartheta},n}$ the law of $\mathbf{u}_{n}$.
The decision rule $\bm{\phi}:\mathbb{R}^{J}\rightarrow\bm{\Delta}$
now takes $\mathbf{y}_{n}$ as its argument, matching the $\sqrt{n}$-scaling
in Proposition \ref{prop:Scaling}.

The following Proposition \ref{prop:(Consequences-of-Local} establishes
that the CLT holds uniformly over the original parameter space $\Theta$
when the estimator $\hat{\bm{\theta}}_{n}$ is properly centered and
``inflated'' by $\sqrt{n}$, a direct consequence of Assumption
\ref{assu:For-every-}.
\begin{prop}[Uniform CLT on the $\sqrt{n}$-Scaled Experiment]
\label{prop:(Consequences-of-Local} Denote by $d_{LP}\left(\cdot,\cdot\right)$
the Levy-Prokhorov metric. Then,
\begin{equation}
\sup_{\bm{\vartheta}\in\Theta}d_{LP}\left(Q_{\bm{\vartheta},n},\mathcal{N}\left(\mathbf{0},\bm{\Sigma}\right)\right)\rightarrow0\quad\text{as }n\rightarrow\infty,\label{eq:Q_central}
\end{equation}
or equivalently, by translation invariance of $d_{LP}$ in $\mathbb{R}^{J}$,
\begin{equation}
\sup_{\bm{\vartheta}\in\Theta}d_{LP}\left(P_{\bm{\vartheta},n},\mathcal{N}\left(\bm{\vartheta},\bm{\Sigma}\right)\right)\rightarrow0\quad\text{as }n\rightarrow\infty.\label{eq:P_central}
\end{equation}
\end{prop}
\begin{proof}
See Appendix \ref{subsec:Prop2_proof}.
\end{proof}

\subsection{Local Uniform Convergence of the Regret to the Gaussian Benchmark
and Asymptotic Minimaxity of the Plug-in MVN Decision Rule}

We establish our main Gaussian approximation results in this subsection.
As before, we index the MVN-MMR rule $\bm{\delta}$, the Gaussian
regret risk function $R\left(\bm{\vartheta},\bm{\delta}\right)$,
and the minimax value $V$ by the covariance matrix $\bm{\Sigma}$
or its consistent estimate $\hat{\bm{\Sigma}}$, respectively. Denote
by $P_{\bm{\vartheta},\bm{\Sigma}}$ the Gaussian measure $\mathcal{N}\left(\bm{\vartheta},\bm{\Sigma}\right)$,
and for any decision rule $\bm{\phi}$ and ($\sqrt{n}$-inflated local)
parameter $\bm{\vartheta}$, define:
\begin{align}
R_{\left[n\right]}\left(\bm{\vartheta},\bm{\phi}\right) & =\int L\left(\bm{\phi}\left(\mathbf{y}\right),\bm{\vartheta}\right)dP_{\bm{\vartheta},n}\left(\mathbf{y}\right)\label{eq:finite_sample_regret}\\
R_{\bm{\Sigma}}\left(\bm{\vartheta},\bm{\phi}\right) & =\int L\left(\bm{\phi}\left(\mathbf{y}\right),\bm{\vartheta}\right)dP_{\bm{\vartheta},\bm{\Sigma}}\left(\mathbf{y}\right).\label{eq:gaussian_regret}
\end{align}

Theorem \ref{thm:(Uniform-Convergence-of} below establishes, for
any fixed rule $\bm{\phi}$, $R_{\left[n\right]}\left(\bm{\vartheta},\bm{\phi}\right)$
converges to $R_{\bm{\Sigma}}\left(\bm{\vartheta},\bm{\phi}\right)$
uniformly over $\bm{\vartheta}\in\Theta$. It immediately follows
the MVN-MMR rule $\bm{\delta}_{\bm{\Sigma}}$ plugged into the Gaussian
risk function achieves the minimax value $V_{\bm{\Sigma}}$. However,
in practice, the covariance matrix $\bm{\Sigma}$ is unknown so the
rule $\bm{\delta}_{\bm{\Sigma}}$ is infeasible. A feasible version
of the decision rule is $\bm{\delta}_{\hat{\bm{\Sigma}}}$, which
is the MVN-MMR decision rule obtained by plugging-in the consistent
estimate $\hat{\bm{\Sigma}}$. Therefore, the subsequent Theorem \ref{thm:(Plug-in-of-the}
establishes that the feasible version $\bm{\delta}_{\hat{\bm{\Sigma}}}$%
{} also achieves the MVN minimax-risk value asymptotically.

Taken together, these two theorems justify plugging $\sqrt{n}\hat{\bm{\theta}}_{n}$
into the MVN decision boundaries derived using the consistent estimate
$\hat{\bm{\Sigma}}$ (equivalently, plugging $\hat{\bm{\theta}}_{n}$
into the decision boundaries for $\frac{1}{n}\hat{\bm{\Sigma}}$)
when $\hat{\bm{\theta}}_{n}$ is $\sqrt{n}$-consistent and (locally
uniformly) asymptotically MVN.
\begin{thm}[Local Uniform Convergence of Regret to the MVN Limit]
\label{thm:(Uniform-Convergence-of}Let $\bm{\phi}:\mathbb{R}^{J}\rightarrow\bm{\Delta}$
be any decision rule that satisfies Assumption \ref{assu:For-each-,}.
Under Assumption \ref{assu:For-every-}, 
\begin{equation}
\sup_{\bm{\vartheta}\in\Theta}\left|R_{\left[n\right]}\left(\bm{\vartheta},\bm{\phi}\right)-R_{\bm{\Sigma}}\left(\bm{\vartheta},\bm{\phi}\right)\right|\rightarrow0\qquad\text{as }n\rightarrow\infty.\label{eq:thm4_conclusion}
\end{equation}
\end{thm}
\begin{proof}
See Appendix \ref{subsec:Thm4_proof}.
\end{proof}
\begin{cor*}
$\sup_{\bm{\vartheta}\in\Theta}R_{\left[n\right]}\left(\bm{\vartheta},\bm{\delta}_{\bm{\Sigma}}\right)\rightarrow\sup_{\bm{\vartheta}\in\Theta}R_{\bm{\Sigma}}\left(\bm{\vartheta},\bm{\delta}_{\bm{\Sigma}}\right)=V_{\bm{\Sigma}}$
as $n\rightarrow\infty$.
\end{cor*}
We next prove a lemma that upgrades the pointwise (in $\bm{\phi}$)
convergence of the risk function to the uniform convergence when we
restrict the class of decision rules to MVN-MMR rules indexed by the
covariance matrices $\bm{\Gamma}\in\mathcal{S}$. Note, this class
of rules satisfy Assumption \ref{assu:For-each-,} by Theorem \ref{thm:(Characterization-of-the}
and its Corollary.
\begin{lem}[Uniform Convergence of the Risk Over the MVN Rules]
\label{lem:41}Let $\bm{\delta}_{\bm{\Gamma}}$ be the MVN-MMR decision
rule indexed by p.d. covariance matrix $\bm{\Gamma}\in\mathcal{S}$.
Then, under Assumption \ref{assu:For-every-},
\[
\sup_{\bm{\Gamma}\in\mathcal{S}}\sup_{\bm{\vartheta}\in\Theta}\left|R_{\left[n\right]}\left(\bm{\vartheta},\bm{\delta}_{\bm{\Gamma}}\right)-R_{\bm{\Sigma}}\left(\bm{\vartheta},\bm{\delta}_{\bm{\Gamma}}\right)\right|\rightarrow0\qquad\text{as }n\rightarrow\infty.
\]
\end{lem}
\begin{proof}
See Appendix \ref{subsec:Proof-of-Lemma41}.
\end{proof}
Finally, the following Theorem \ref{thm:(Plug-in-of-the} establishes
the plug-in MVN rule $\bm{\delta}_{\hat{\bm{\Sigma}}}$ also achieves
$V_{\bm{\Sigma}}$ asymptotically.
\begin{thm}[Plug-in of the Consistent Covariance Estimator]
\label{thm:(Plug-in-of-the}Let $\hat{\bm{\theta}}_{n}$ be a mean-reward
estimator satisfying Assumption \ref{assu:For-every-} with $\hat{\bm{\Sigma}}\rightarrow_{p}\bm{\Sigma}\in\mathcal{S}$,
where $\hat{\bm{\Sigma}}\in\mathcal{S}$ with probability approaching
one as $n\rightarrow\infty$. Then, 
\[
\sup_{\bm{\vartheta}\in\Theta}R_{\left[n\right]}\left(\bm{\vartheta},\bm{\delta}_{\hat{\bm{\Sigma}}}\right)=\sup_{\bm{\vartheta}\in\Theta}R_{\bm{\Sigma}}\left(\bm{\vartheta},\bm{\delta}_{\bm{\Sigma}}\right)+o_{p}\left(1\right)=V_{\bm{\Sigma}}+o_{p}\left(1\right)\qquad\text{as }n\rightarrow\infty.
\]
\end{thm}
\begin{proof}
See Appendix \ref{subsec:Proof-of-Theorem42}.
\end{proof}

\section{Numerical Optimization and Examples}\label{sec:Applications}

In this section, we describe our numerical optimization approach for
finding the MVN-MMR decision rule and report decision boundaries for
several examples. In Section \ref{subsec:Numerical-Optimization},
we briefly describe how we carry out the numerical optimization to
find the least-favorable prior, which hinges on the theoretical results
from Sections \ref{sec:Characterization-of-the}-\ref{sec:The-MVN-Best}.
In Section \ref{subsec:Decision-Boundaries-for-2}, we illustrate
the decision boundaries for the two-arm and three-arm MVN-MMR problems
under homoskedasticity and an arbitrary covariance matrix, respectively.

Through the numerical analysis, we reconfirm that the decision boundaries
are linear and that the “pick-the-winner” empirical success rule is
optimal either (i) in two-arm experiments or (ii) when the covariance
matrix is homoskedastic in multi-arm experiments. However, when the
covariance matrix is not homoskedastic in multi-arm experiments, the
problem is no longer permutation invariant and the decision boundaries
become nonlinear; in this case the empirical success rule is no longer
optimal. The optimal decision rule penalizes higher-variance arms,
requiring stronger evidence to select a high-variance arm.

\subsection{Numerical Optimization to Find the Least-Favorable Prior}\label{subsec:Numerical-Optimization}

We characterized the least-favorable prior $\left(\pi^{k},\bm{\theta}^{k}\right)_{k\in\mathcal{J}}$
in the dual maximin problem, which is supported by $J$ distinct points
with exactly one point in $\Theta^{j}$. The Bayes rule $\bm{\delta}$
under this least-favorable prior takes the deterministic form (Theorem
\ref{thm:(Characterization-of-the}). The form of $\bm{\delta}$ is
tractable enough that it can be directly plugged into the Bayes risk.
Therefore, we numerically solve nature's problem of maximizing the
Bayes risk under $\bm{\delta}$:
\begin{equation}
\max_{\left(\bm{\theta}^{j},\pi^{j}\right)_{j\in\mathcal{J}}}\sum_{j=1}^{J}\pi^{j}\left[\theta_{j}^{j}-\sum_{i=1}^{J}\theta_{i}^{j}\int1\left(\sum_{k=1}^{J}\theta_{i}^{k}\pi^{k}p_{\bm{\theta}^{k}}\left(\mathbf{x}\right)>\sum_{k=1}^{J}\theta_{j}^{k}\pi^{k}p_{\bm{\theta}^{k}}\left(\mathbf{x}\right)\:\forall j\neq i\right)p_{\bm{\theta}^{j}}\left(\mathbf{x}\right)d\mu\left(\mathbf{x}\right)\right].\label{eq:gaussian_integral}
\end{equation}

In the operationalization, we use $2^{19}$ Sobol quasi Monte Carlo
draws to approximate the Gaussian integral in (\ref{eq:gaussian_integral}).
We then use two algorithms to solve the maximization problem: (i)
direct maximization using the derivative-free subplex algorithm; and
(ii) a softmax smoothing of the indicator function inside the Gaussian
integral combined with a derivative-based optimizer, supplying closed-form
derivatives. The first approach works better with a small problem
dimension $\left(J\leq4\right)$, whereas the second approach works
better with modest parameter dimensions $\left(J\geq5\right)$. We
confirm that both algorithms converge to the same least-favorable
prior when they converge.

Because the objective function has many near-optimal points, we use
many randomized starting points to increase the chance that the algorithm
reaches the global maximum. The theoretical prediction from Theorem
\ref{thm:(Minimax-Theorem-for} (iv) that $R\left(\bm{\theta}^{k},\bm{\delta}\right)=V$
for all $k$ serves as a diagnostic. We relegate the details of the
numerical optimization to Appendix \ref{sec:Details-of-Numerical}. 

\subsection{Decision Boundaries for Two-Arm and Three-Arm MVN Experiments}\label{subsec:Decision-Boundaries-for-2}

\subsubsection{Decision Boundaries for Two-Arm MVN Experiments}\label{subsec:Decision-Boundaries-for}

It is known that the problem with $J=2$ can be reduced to a one-dimensional
``testing'' problem with $x_{2}-x_{1}$, and the MMR-optimal threshold
for the difference $x_{2}-x_{1}$ is $0$ \citep{Tetenov2012,JooForthcoming}.
In \citeauthor{Tetenov2012,JooForthcoming}, the optimal decision
threshold in the two-arm case is derived by directly maximizing the
Gaussian risk over the parameter-coordinate difference $\theta_{2}-\theta_{1}$.
In the following, we find that our decision boundaries, the minimax
value, and the locations of the least-favorable prior’s support points
calculated from the dual maximin problem exactly coincide with those
reported in \citeauthor{Tetenov2012,JooForthcoming}. 

For the two-arm experiments, we let 
\[
\bm{\Sigma}_{\text{sym}}=\begin{pmatrix}0.5 & 0\\
0 & 0.5
\end{pmatrix},\quad\bm{\Sigma}_{\text{asym}}=\begin{pmatrix}1 & -0.5\\
-0.5 & 5
\end{pmatrix}.
\]
For $\bm{\Sigma}_{\text{sym}}$, we have $\Sigma_{11}+\Sigma_{22}=1$
and the minimax value is known to be $0.170$ (\citeauthor{Tetenov2012,JooForthcoming}),
which we confirm in the Bayes risk column of Table \ref{tab:2-sym}.
In both Tables \ref{tab:2-sym} and \ref{tab:2-asym}, the risk values
attained at the support points of the least-favorable prior are identical
up to rounding error, reconfirming the theoretical prediction of Theorem
\ref{thm:(Minimax-Theorem-for} (iv).

In Figure \ref{fig:Two-arm-decision-boundaries}, the decision boundaries
align exactly with the $45^{\circ}$ line, corresponding to the empirical
success rule. Therefore, we also reconfirm that our dual formulation
of the problem yields the same decision boundary for both the symmetric
and asymmetric cases. The asymmetric arm-specific variance does not
affect the decision boundaries nor the prior probabilities in this
two-arm case---the location of the least-favorable prior’s support
points adjusts to equalize the risk values attained at those points. 

\begin{table}[H]
\caption{Least-favorable prior and minimax/Bayes risk values for $J=2$}\label{tab:Least-favorable-prior-and}

\begin{centering}
{\small\subfloat[\label{tab:2-sym}$\bm{\Sigma}_{\text{sym}}$]{{\small}{\small\par}
\centering{}{\small{}%
\begin{tabular}{cccccc}
\toprule 
\multirow{1}{*}{{\small$k$}} & \multirow{1}{*}{{\small$\pi^{k}$}} & \multirow{1}{*}{{\small$R\left(\bm{\theta}^{k},\bm{\delta}\right)$}} & \multirow{1}{*}{{\small Bayes Risk}} & {\small LFP $\bm{\theta}^{k}$ Coordinate} & {\small$d\left(\bm{\theta}^{k},t\mathbf{1}\right)$}\tabularnewline
\midrule 
{\small 1} & {\small 0.499} & {\small 0.171} & \multirow{2}{*}{{\small 0.170}} & {\small$\left(\text{–}0.001,\text{–}0.751\right)$} & {\small 0.530}\tabularnewline
{\small 2} & {\small 0.501} & {\small 0.169} &  & {\small$\left(\text{–}0.755,0.000\right)$} & {\small 0.534}\tabularnewline
\bottomrule
\end{tabular}}{\small\par}}}{\small\par}
\par\end{centering}
\begin{centering}
{\small\subfloat[\label{tab:2-asym}$\bm{\Sigma}_{\text{asym}}$]{{\small}{\small\par}
\centering{}{\small{}%
\begin{tabular}{cccccc}
\toprule 
\multirow{1}{*}{{\small$k$}} & \multirow{1}{*}{{\small$\pi^{k}$}} & \multirow{1}{*}{{\small$R\left(\bm{\theta}^{k},\bm{\delta}\right)$}} & \multirow{1}{*}{{\small Bayes Risk}} & {\small LFP $\bm{\theta}^{k}$ Coordinate} & {\small$d\left(\bm{\theta}^{k},t\mathbf{1}\right)$}\tabularnewline
\midrule 
{\small 1} & {\small 0.500} & {\small 0.449} & \multirow{2}{*}{{\small 0.450}} & {\small$\left(-1.138,-3.133\right)$} & {\small 1.411}\tabularnewline
{\small 2} & {\small 0.500} & {\small 0.450} &  & {\small$\left(-1.994,0.000\right)$} & {\small 1.410}\tabularnewline
\bottomrule
\end{tabular}}{\small\par}}}{\small\par}
\par\end{centering}
{\small Note. The tables report the least-favorable prior $\left(\bm{\theta}^{k},\pi^{k}\right)_{k\in\mathcal{J}}$,
pointwise-risk at the support point $\bm{\theta}^{k}$s, and the Bayes
risk.LFP stands for least-favorable prior. $d\left(\bm{\theta}^{k},t\mathbf{1}\right)$
is the Euclidean distance between the LFP support point and the line
$t\mathbf{1}$, calculated as $\left\Vert \bm{\theta}^{k}-\bar{\theta}^{k}\mathbf{1}\right\Vert _{2}$
where $\bar{\theta}^{k}=\frac{1}{J}\sum_{j=1}^{J}\theta_{j}^{k}$
is the average of the coordinates. $\theta_{2}^{2}=0$ is imposed
during the numerical optimization by invoking location invariance. }{\small\par}
\end{table}

\begin{figure}[H]
\caption{Two-arm decision boundaries and the least-favorable prior}\label{fig:Two-arm-decision-boundaries}

\begin{centering}
\subfloat[Decision boundaries for $\bm{\Sigma}_{\text{sym}}$]{
\begin{centering}
\includegraphics[width=0.4\textwidth]{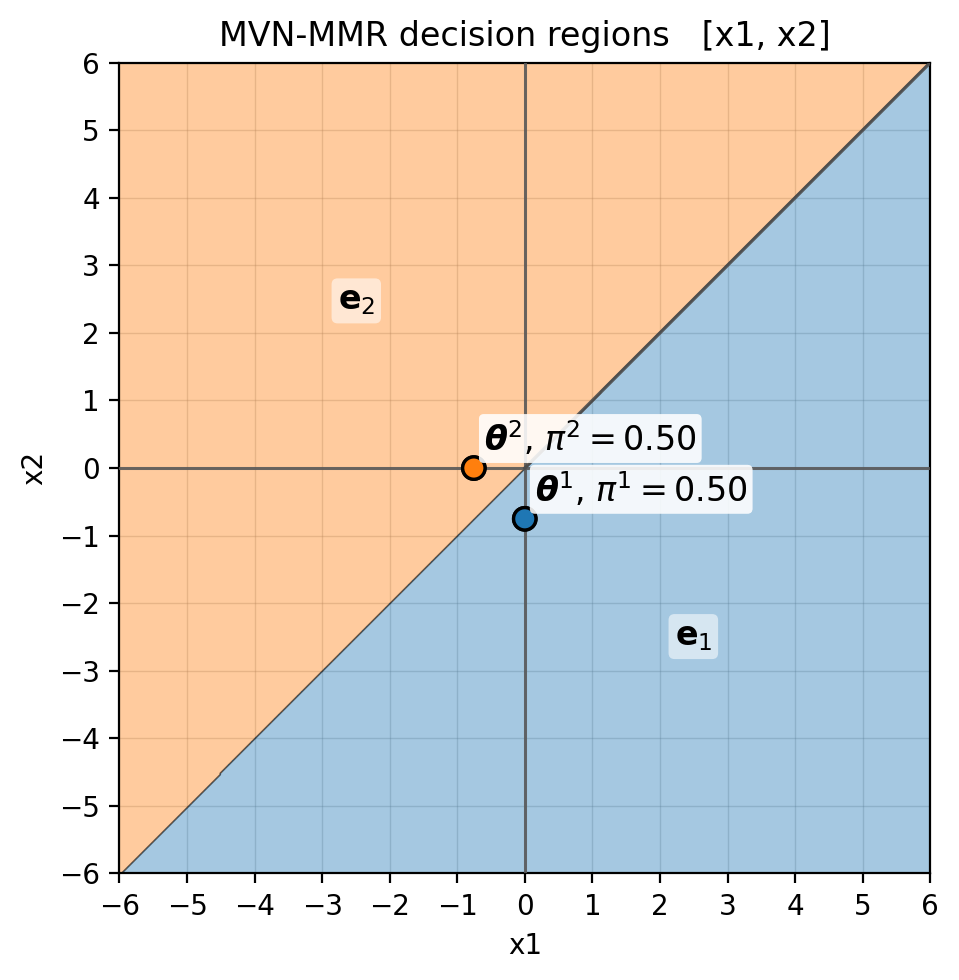}
\par\end{centering}
}\subfloat[Decision boundaries for $\bm{\Sigma}_{\text{asym}}$]{
\centering{}\includegraphics[width=0.4\textwidth]{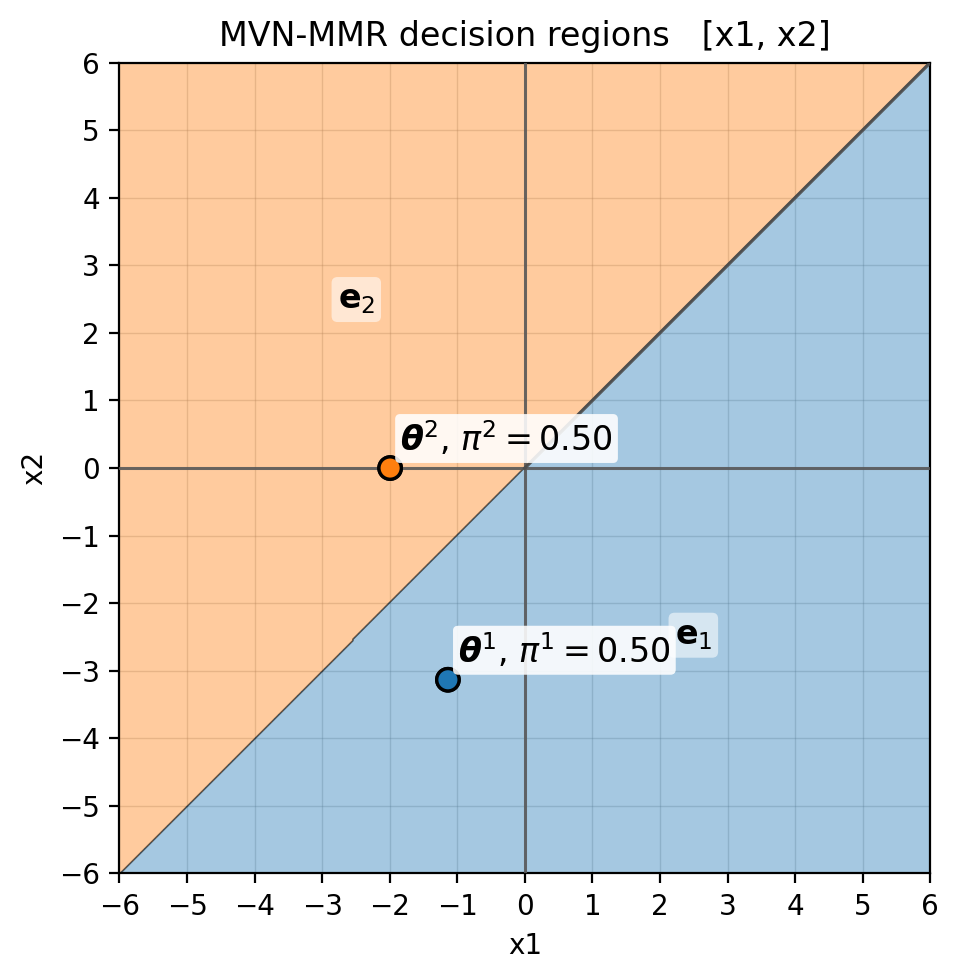}}
\par\end{centering}
{\small Note. The decision regions are calculated by numerically maximizing
the Bayes risk under the 2-point supported least-favorable prior with
$\bm{\Sigma}_{\text{sym}}$ and $\bm{\Sigma}_{\text{asym}}$, respectively.
$\mathbf{e}_{j}$ denotes the MMR decision region of the $j$-th arm.
The support points and the prior probabilities are marked with dots,
respectively. }{\small\par}
\end{figure}

\subsubsection{Decision Boundaries for Three-Arm MVN Experiments}\label{subsec:Decision-Boundaries-for-1}

A more interesting case is $J\geq3$, where homoskedasticity no longer
holds. We illustrate this case for two three-arm experiments. Let
the covariance matrices be as follows:
\[
\bm{\Sigma}_{\text{sym}}=\begin{pmatrix}0.33 & 0 & 0\\
0 & 0.33 & 0\\
0 & 0 & 0.33
\end{pmatrix},\quad\bm{\Sigma}_{\text{asym}}=\begin{pmatrix}1 & 0.5 & 0\\
0.5 & 1 & -0.5\\
0 & -0.5 & 5
\end{pmatrix}.
\]
Arm 3 of $\bm{\Sigma}_{\text{asym}}$ exhibits higher variance than
the other two arms.

In Table \ref{tab:Least-favorable-prior-J3}, we report the least-favorable
prior’s support points with their prior probabilities and the associated
risk values evaluated at those points. We reconfirm the conclusion
of Theorem \ref{thm:(Minimax-Theorem-for} (iv) that the risk values
attained are the same at all three support points of the least-favorable
prior up to a rounding error. Unlike the two-arm experiment, however,
both the prior probabilities and the support points of the least-favorable
prior are no longer symmetric across the arms.

The decision boundaries are illustrated in Figure \ref{fig:Three-arm-sym_decision-boundaries}
for the symmetric covariance matrix $\bm{\Sigma}_{\text{sym}}$ and
in Figure \ref{fig:Three-arm-asym-decision-boundaries} for the asymmetric
covariance matrix $\bm{\Sigma}_{\text{asym}}$. Each subfigure illustrates
the decision boundaries for two coordinates, fixing the remaining
coordinate at $\left(-2,0,2\right)$. We mark the $45^{\circ}$ line
in the first quadrants as a benchmark representing the empirical success
rule's boundary.

For the symmetric covariance matrix $\bm{\Sigma}_{\text{sym}}$, we
find that the boundaries are linear and the empirical success rule
is optimal, as seen in the middle panels of Figures \ref{fig:sym_J3_x3_fixed}-\ref{fig:sym_J3_fixed_x1}.
As reported in Table \ref{tab:sym_lfp_J3}, the least-favorable prior’s
support points are symmetric with respect to the origin, with prior
probabilities of 1/3, up to rounding error.

For the asymmetric covariance matrix $\bm{\Sigma}_{\text{asym}}$,
we find that the empirical success rule is no longer optimal, as seen
in the first quadrants of the middle panels of Figures \ref{fig:,-fixed}-\ref{fig:,-fixed-1}.
The decision boundaries deviate from the $45^{\circ}$ line. The areas
of $\bm{\delta}\left(\mathbf{x}\right)=\mathbf{e}_{3}$ shrink, implying
stronger evidence is required to select the higher-variance arm. Interestingly,
both the support points and the prior probabilities of the least-favorable
prior are no longer symmetric, as reported in Table \ref{tab:asym_lfp_J3}.

\begin{table}[H]
\caption{Least-favorable prior and minimax/Bayes risk values for $J=3$}\label{tab:Least-favorable-prior-J3}

\begin{centering}
{\small\subfloat[\label{tab:sym_lfp_J3}$\bm{\Sigma}_{\text{sym}}$]{{\small}{\small\par}
\centering{}{\small{}%
\begin{tabular}{cccccc}
\toprule 
\multirow{1}{*}{{\small$k$}} & \multirow{1}{*}{{\small$\pi^{k}$}} & \multirow{1}{*}{{\small$R\left(\bm{\theta}^{k},\bm{\delta}\right)$}} & \multirow{1}{*}{{\small Bayes Risk}} & {\small LFP $\bm{\theta}^{k}$ Coordinate} & {\small$d\left(\bm{\theta}^{k},t\mathbf{1}\right)$}\tabularnewline
\midrule
{\small 1} & {\small 0.333} & {\small 0.215} & \multirow{3}{*}{{\small 0.215}} & {\small$\left(0.010,-0.652,-0.663\right)$} & {\small 0.545}\tabularnewline
{\small 2} & {\small 0.336} & {\small 0.213} &  & {\small$\left(-0.657,0.010,-0.657\right)$} & {\small 0.545}\tabularnewline
{\small 3} & {\small 0.331} & {\small 0.217} &  & {\small$\left(-0.663,-0.667,0.000\right)$} & {\small 0.543}\tabularnewline
\bottomrule
\end{tabular}}{\small\par}}}{\small\par}
\par\end{centering}
\begin{centering}
{\small\subfloat[\label{tab:asym_lfp_J3}$\bm{\Sigma}_{\text{asym}}$]{{\small}{\small\par}
\centering{}{\small{}%
\begin{tabular}{cccccc}
\toprule 
\multirow{1}{*}{{\small$k$}} & \multirow{1}{*}{{\small$\pi^{k}$}} & \multirow{1}{*}{{\small$R\left(\bm{\theta}^{k},\bm{\delta}\right)$}} & \multirow{1}{*}{{\small Bayes Risk}} & {\small LFP $\bm{\theta}^{k}$ Coordinate} & {\small$d\left(\bm{\theta}^{k},t\mathbf{1}\right)$}\tabularnewline
\midrule
{\small 1} & {\small 0.249} & {\small 0.524} & \multirow{3}{*}{{\small 0.524}} & {\small$\left(-1.027,-1.762,-2.939\right)$} & {\small 1.364}\tabularnewline
{\small 2} & {\small 0.318} & {\small 0.522} &  & {\small$\left(-1.873,-1.125,-3.259\right)$} & {\small 1.531}\tabularnewline
{\small 3} & {\small 0.433} & {\small 0.526} &  & {\small$\left(-1.965,-2.136,0.000\right)$} & {\small 1.679}\tabularnewline
\bottomrule
\end{tabular}}{\small\par}}}{\small\par}
\par\end{centering}
{\small Note. The tables report the least-favorable prior $\left(\bm{\theta}^{k},\pi^{k}\right)_{k\in\mathcal{J}}$,
pointwise-risk at the support point $\bm{\theta}^{k}$s, and the Bayes
risk. LFP stands for least-favorable prior. $d\left(\bm{\theta}^{k},t\mathbf{1}\right)$
is the Euclidean distance between the LFP support point and the line
$t\mathbf{1}$, calculated as $\left\Vert \bm{\theta}^{k}-\bar{\theta}^{k}\mathbf{1}\right\Vert _{2}$
where $\bar{\theta}^{k}=\frac{1}{J}\sum_{j=1}^{J}\theta_{j}^{k}$
is the average of the coordinates. $\theta_{3}^{3}=0$ is imposed
during the numerical optimization by invoking location invariance.}{\small\par}
\end{table}

\begin{figure}[H]
\caption{Three-arm decision boundaries for $\bm{\Sigma}_{\text{sym}}$}\label{fig:Three-arm-sym_decision-boundaries}

\begin{centering}
\subfloat[\label{fig:sym_J3_x3_fixed}$\left(x_{1},x_{2}\right)$, fixed $x_{3}=-2,0,2$]{
\centering{}\includegraphics[width=0.3\textwidth]{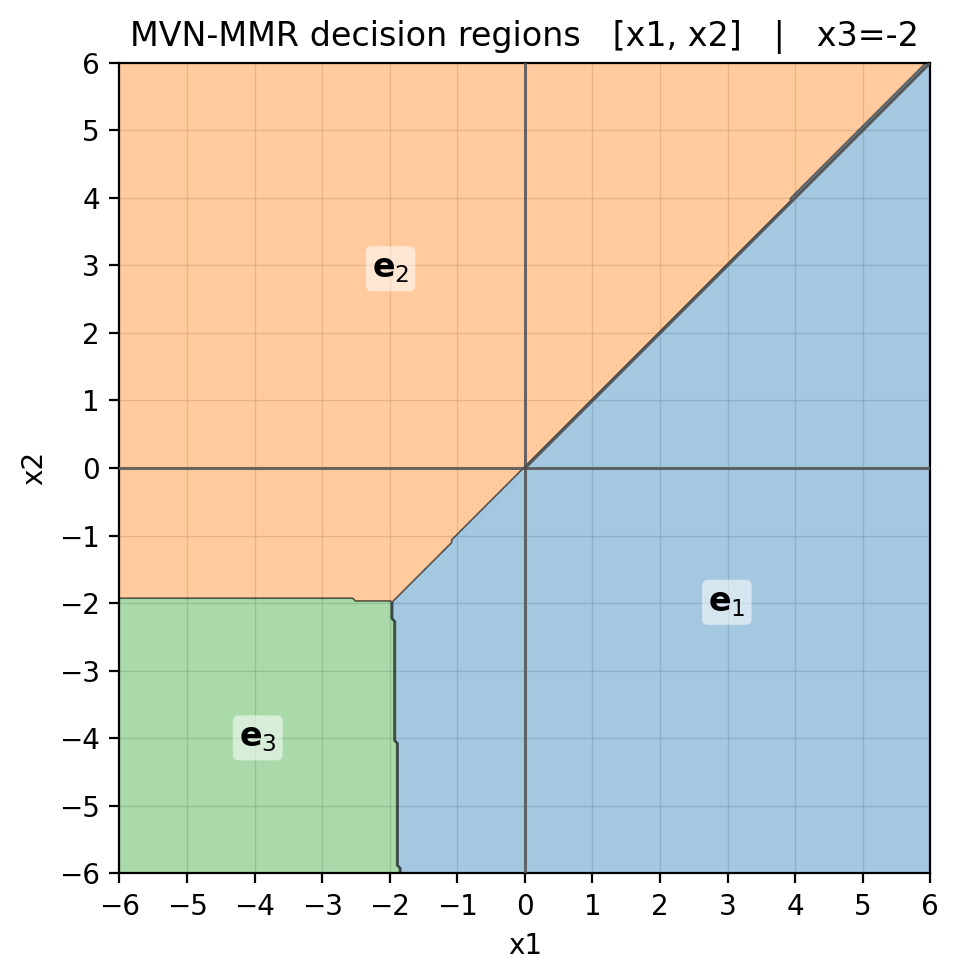}\includegraphics[width=0.3\textwidth]{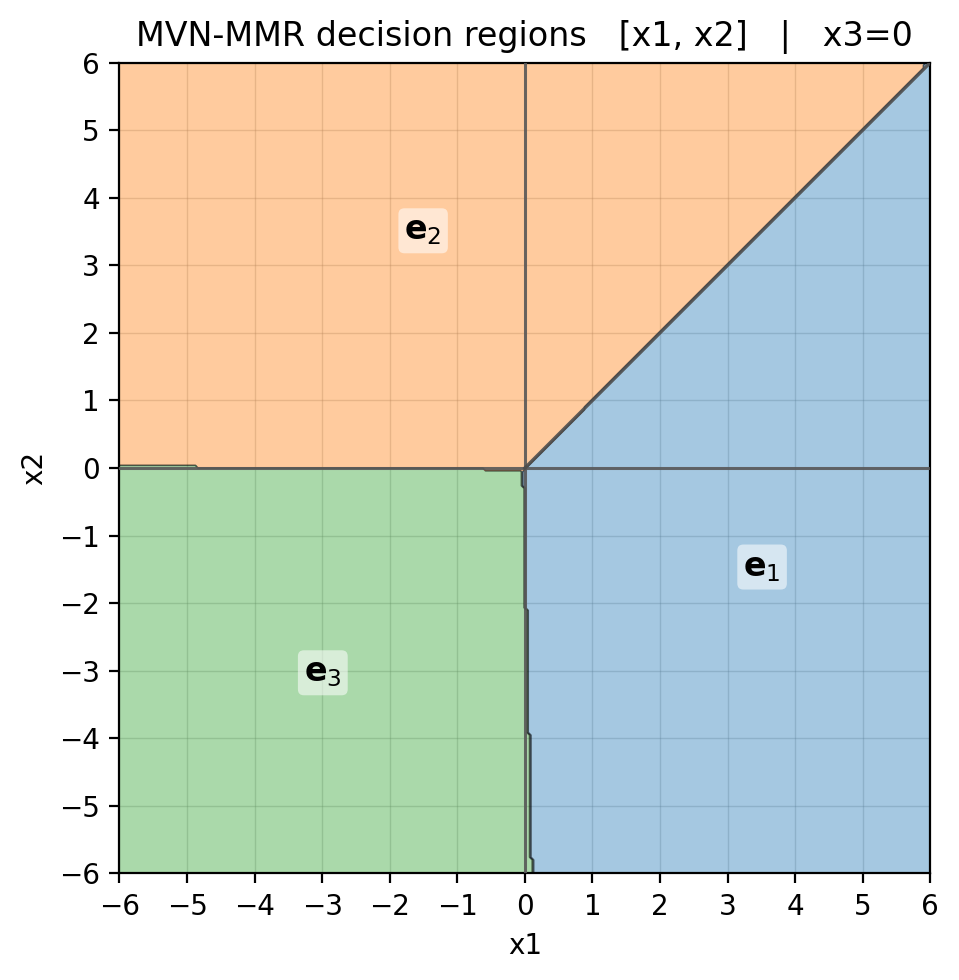}\includegraphics[width=0.3\textwidth]{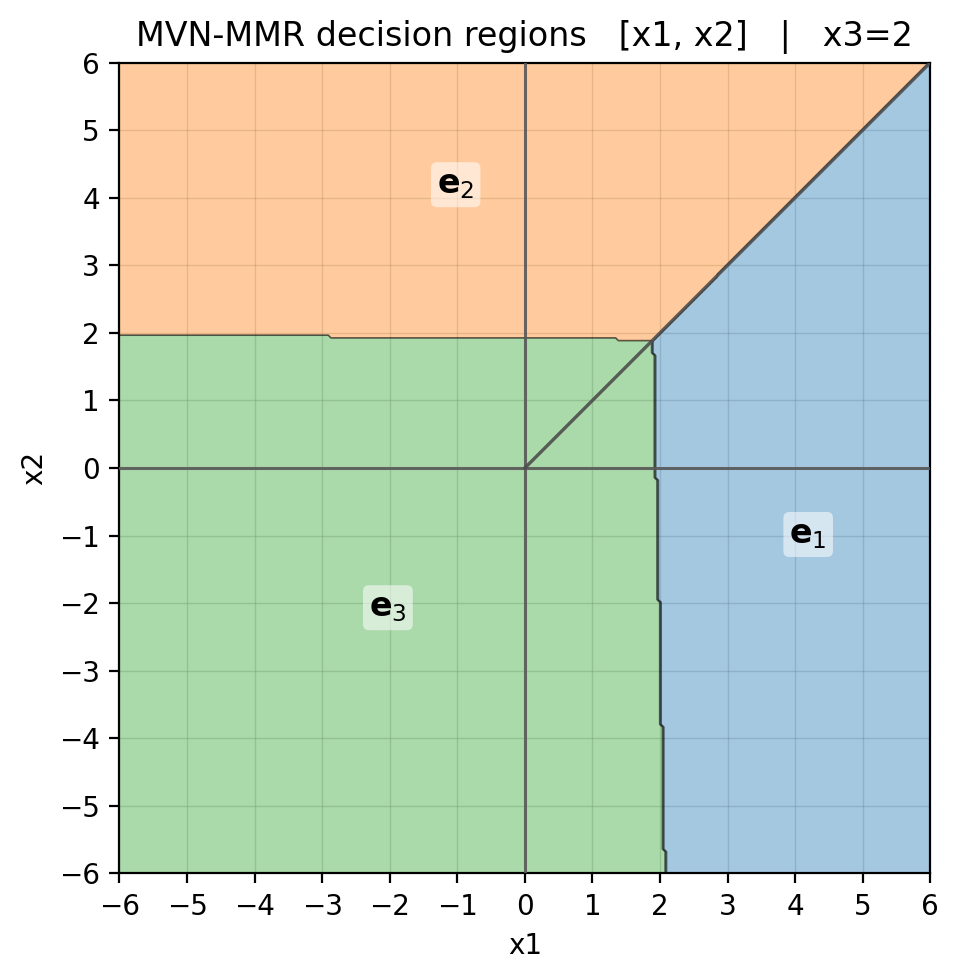}}
\par\end{centering}
\begin{centering}
\subfloat[\label{fig:,-fixed-2}$\left(x_{1},x_{3}\right)$, fixed $x_{2}=-2,0,2$]{
\centering{}\includegraphics[width=0.3\textwidth]{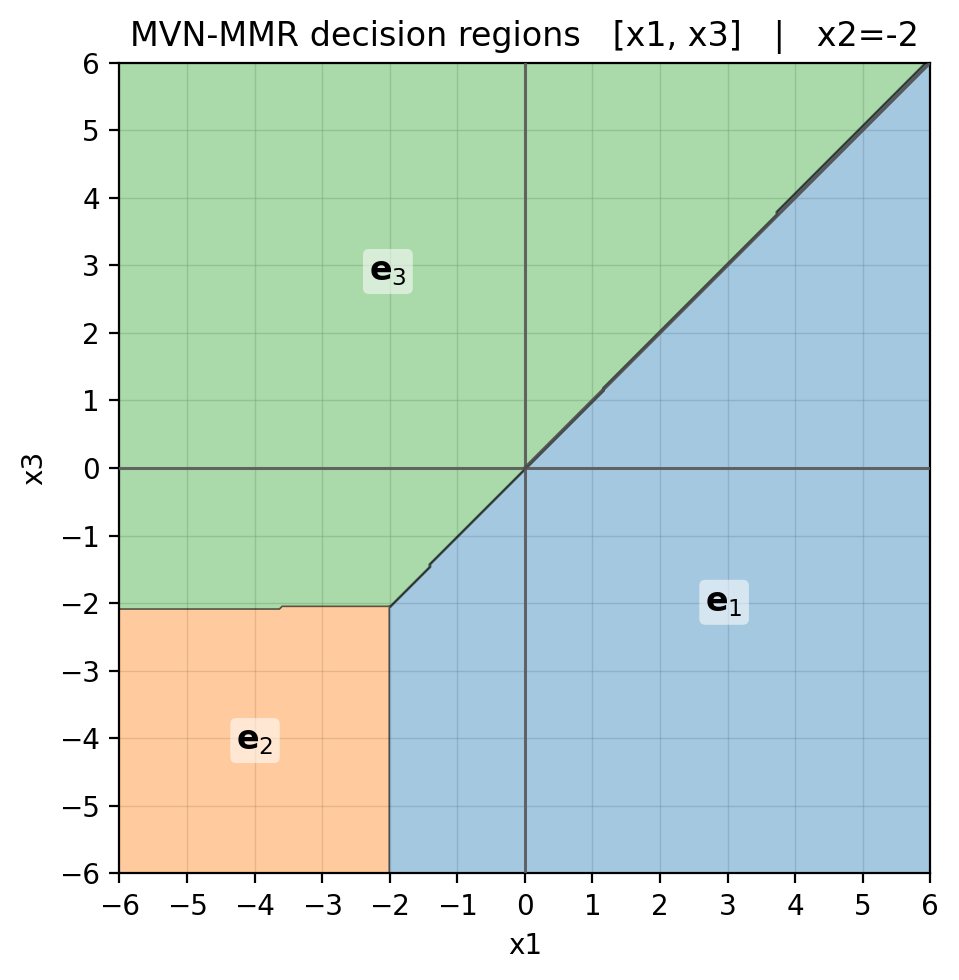}\includegraphics[width=0.3\textwidth]{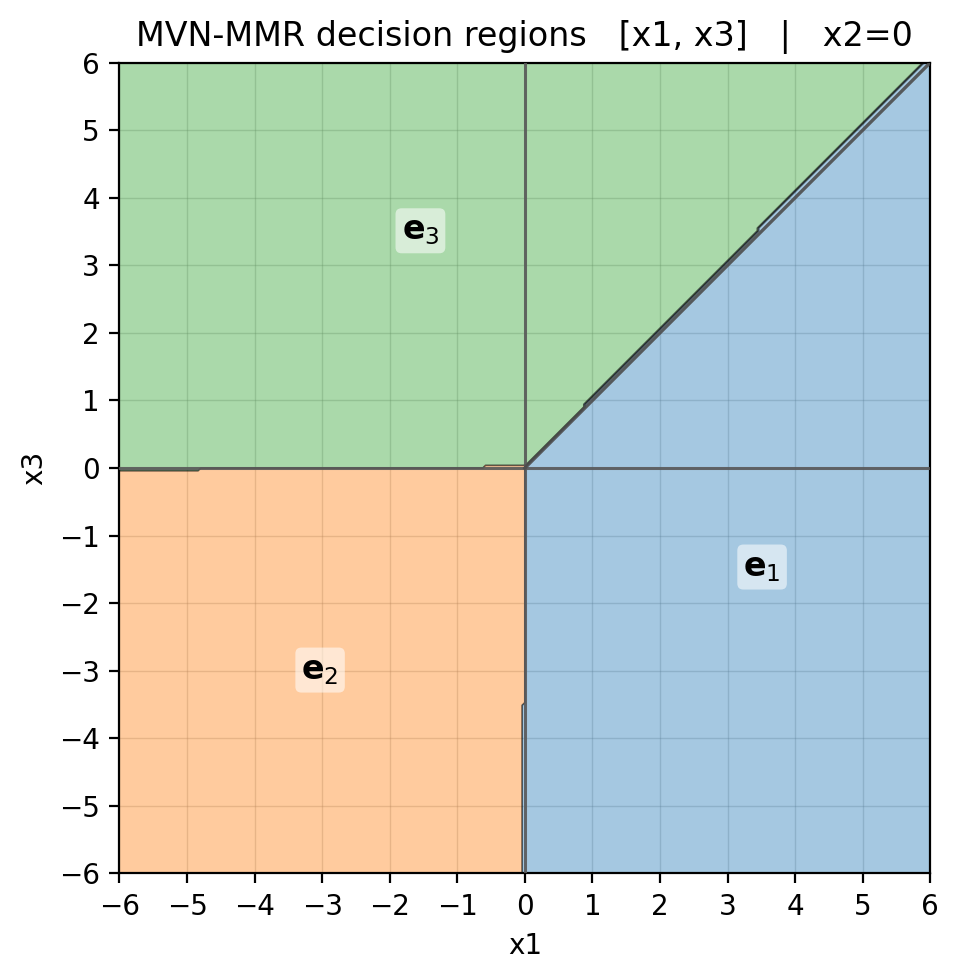}\includegraphics[width=0.3\textwidth]{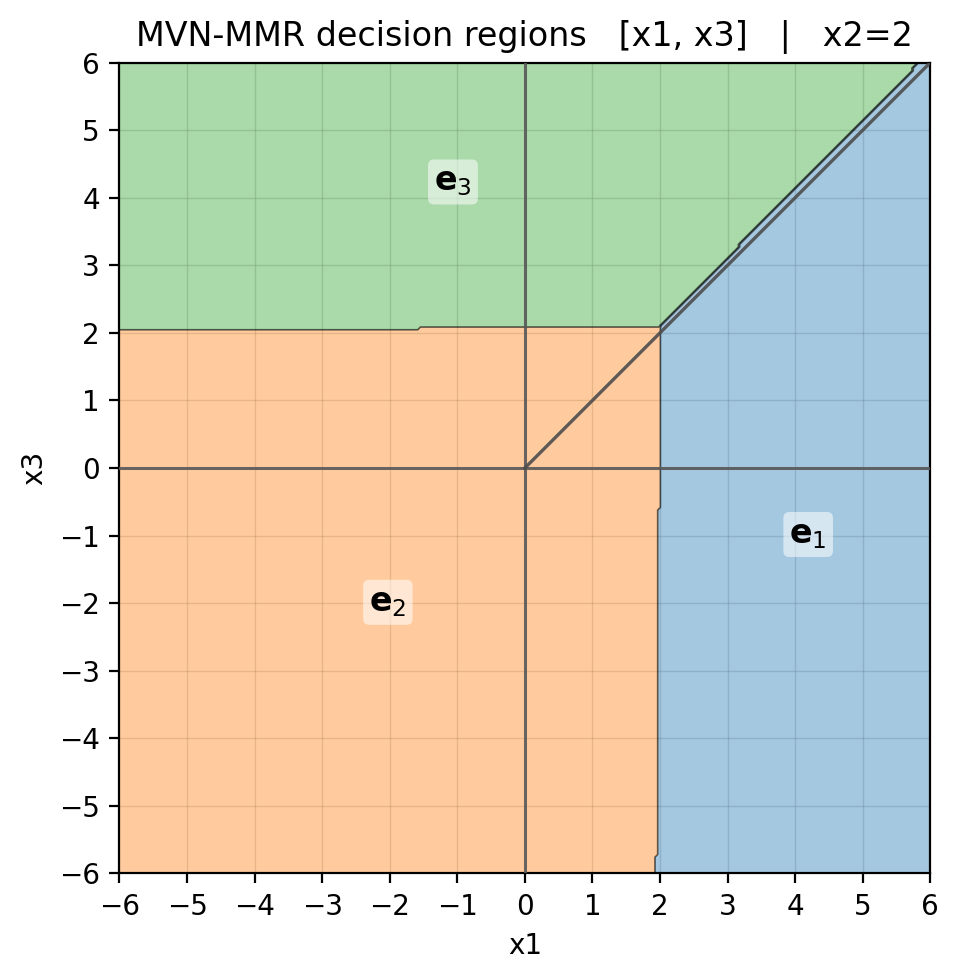}}
\par\end{centering}
\begin{centering}
\subfloat[\label{fig:sym_J3_fixed_x1}$\left(x_{2},x_{3}\right)$, fixed $x_{1}=-2,0,2$]{
\centering{}\includegraphics[width=0.3\textwidth]{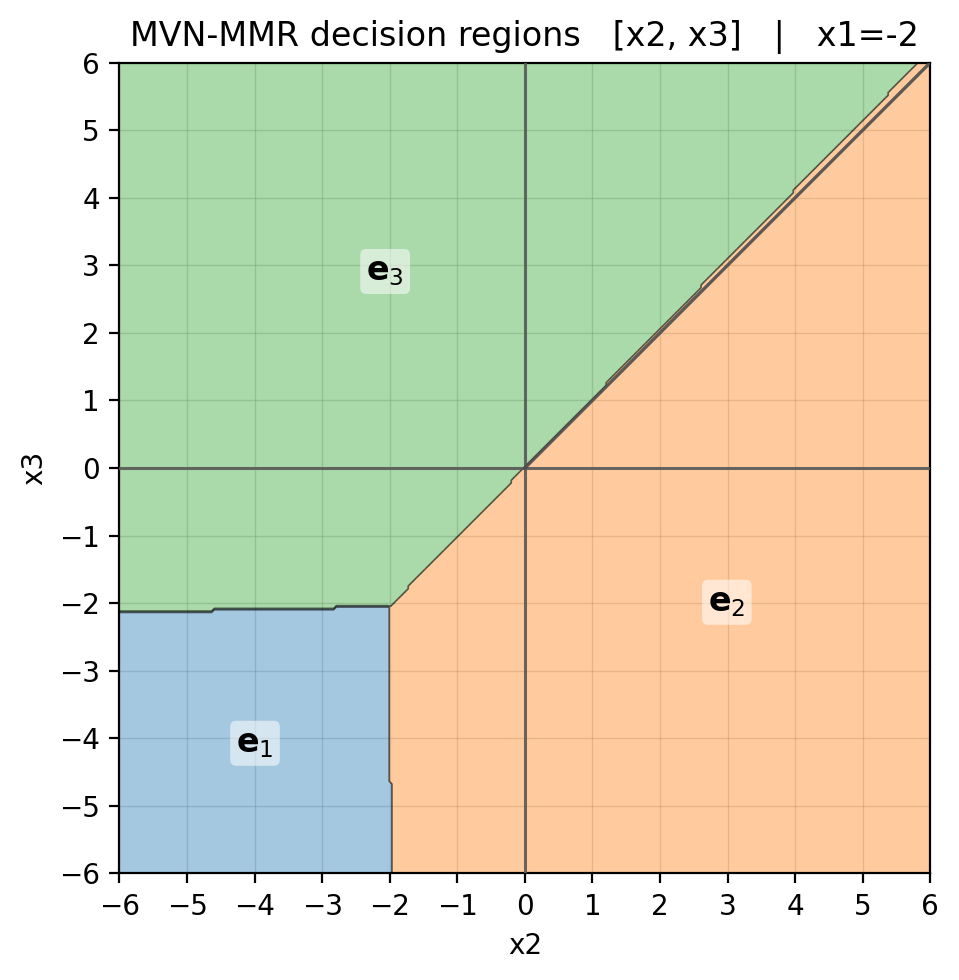}\includegraphics[width=0.3\textwidth]{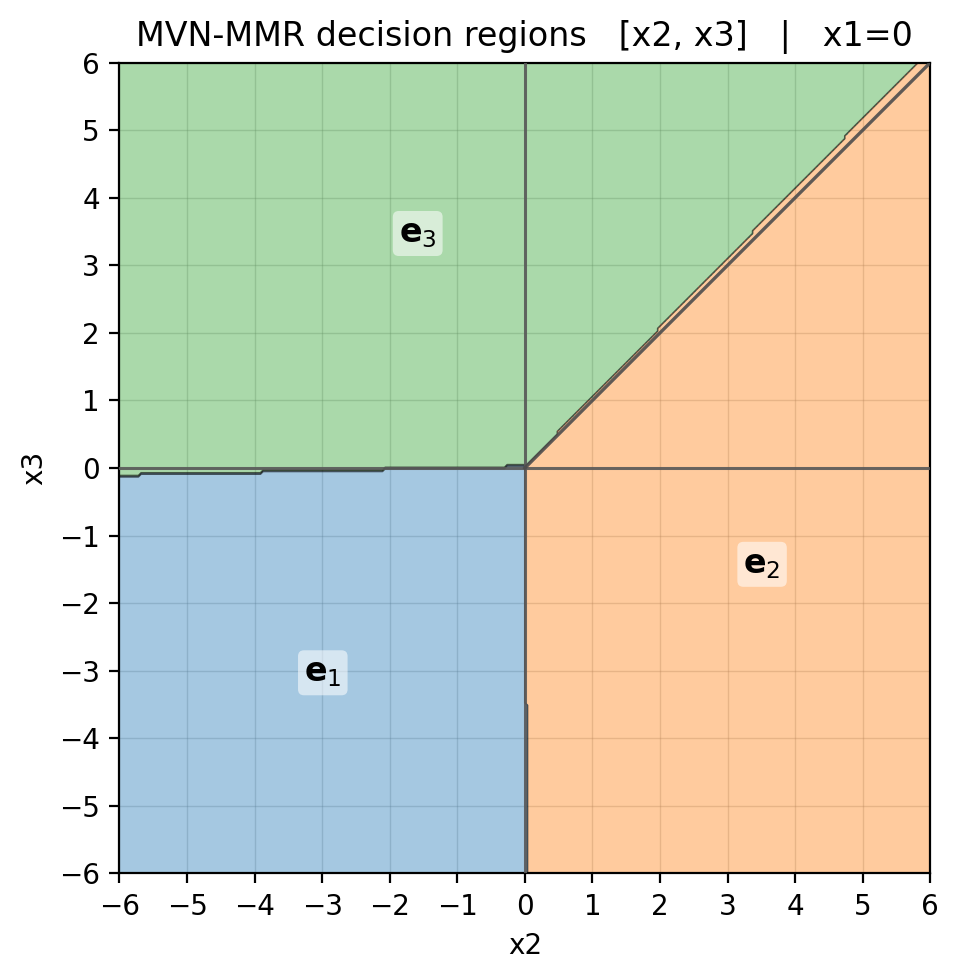}\includegraphics[width=0.3\textwidth]{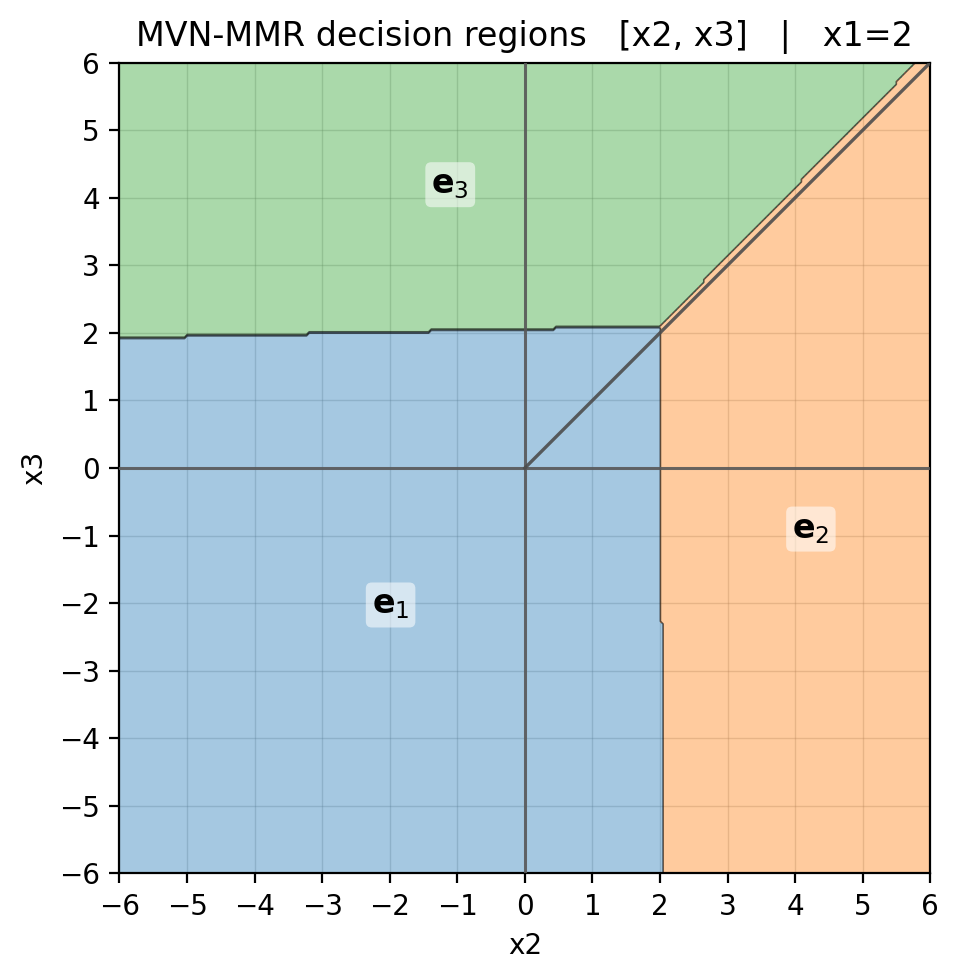}}
\par\end{centering}
{\small Note. The decision regions are calculated by numerically maximizing
the Bayes risk under the 3-point supported least-favorable prior with
the symmetric covariance matrix $\bm{\Sigma}_{\text{sym}}$. $\mathbf{e}_{j}$
denotes the MMR decision region of the $j$-th arm.}{\small\par}
\end{figure}

\begin{figure}[H]
\caption{Three-arm decision boundaries for $\bm{\Sigma}_{\text{asym}}$}\label{fig:Three-arm-asym-decision-boundaries}

\begin{centering}
\subfloat[$\left(x_{1},x_{2}\right)$, fixed $x_{3}=-2,0,2$]{
\centering{}\includegraphics[width=0.3\textwidth]{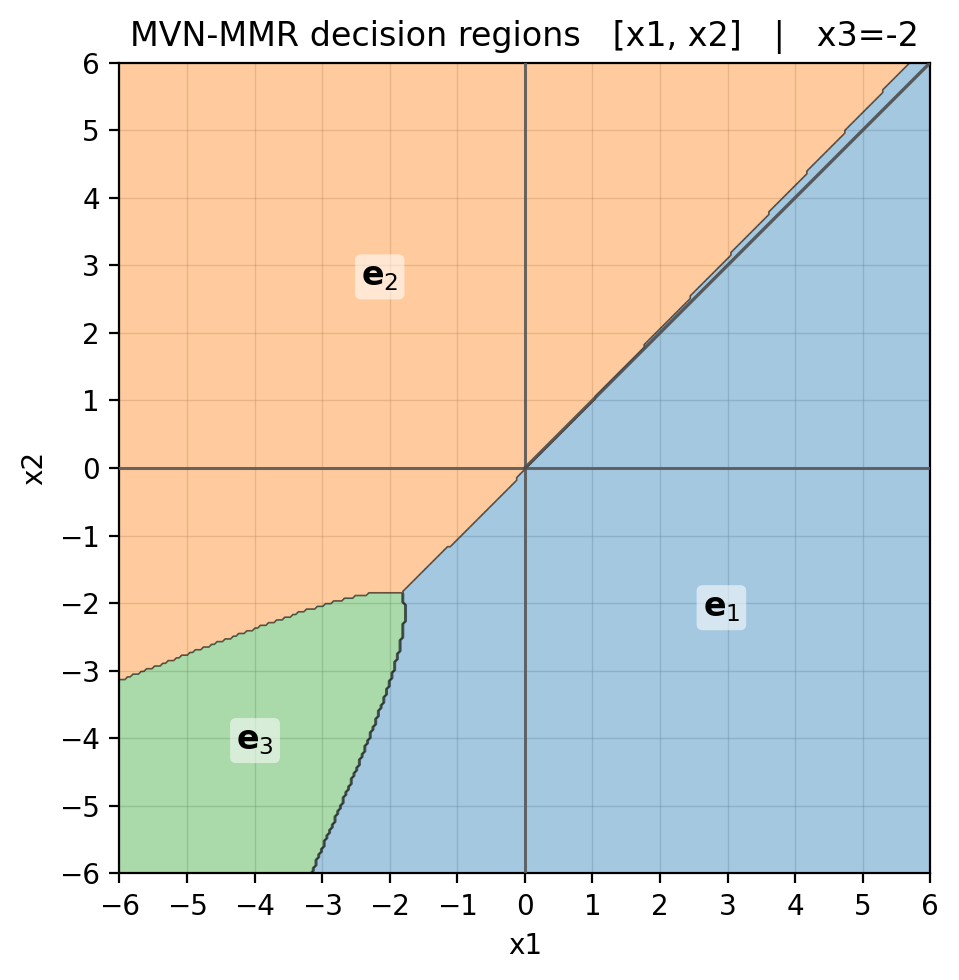}\includegraphics[width=0.3\textwidth]{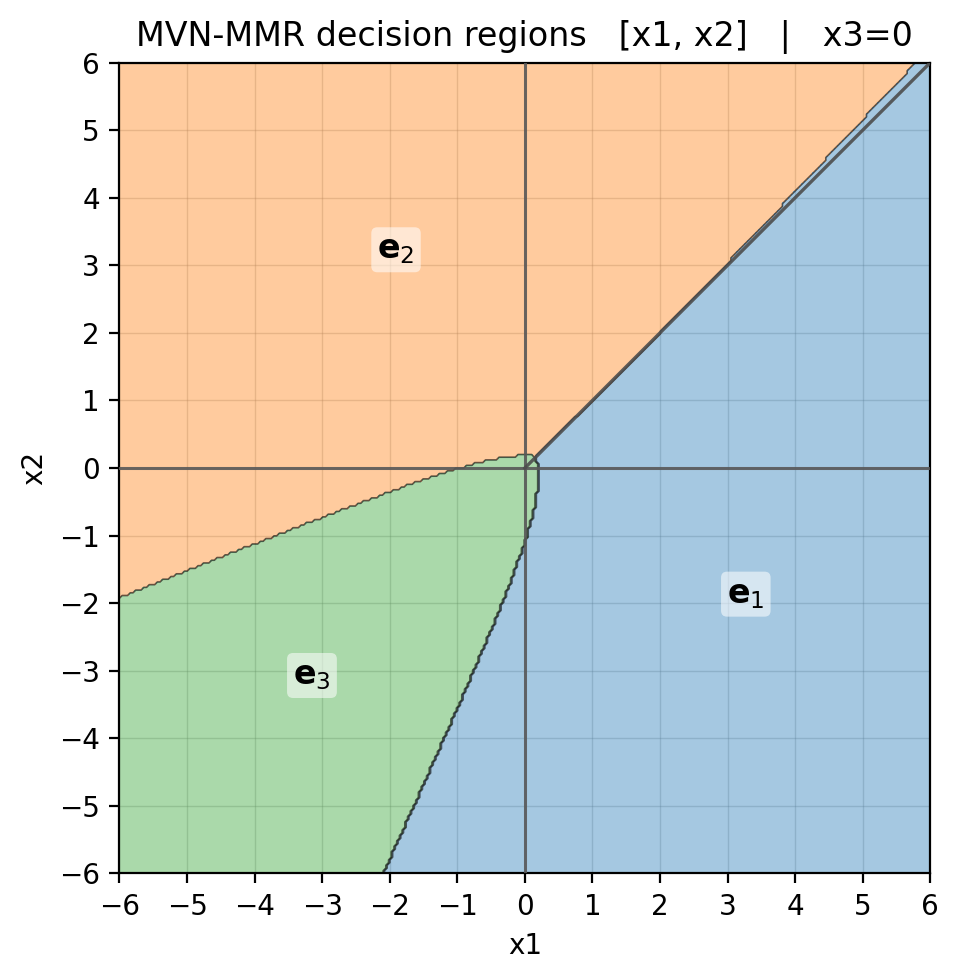}\includegraphics[width=0.3\textwidth]{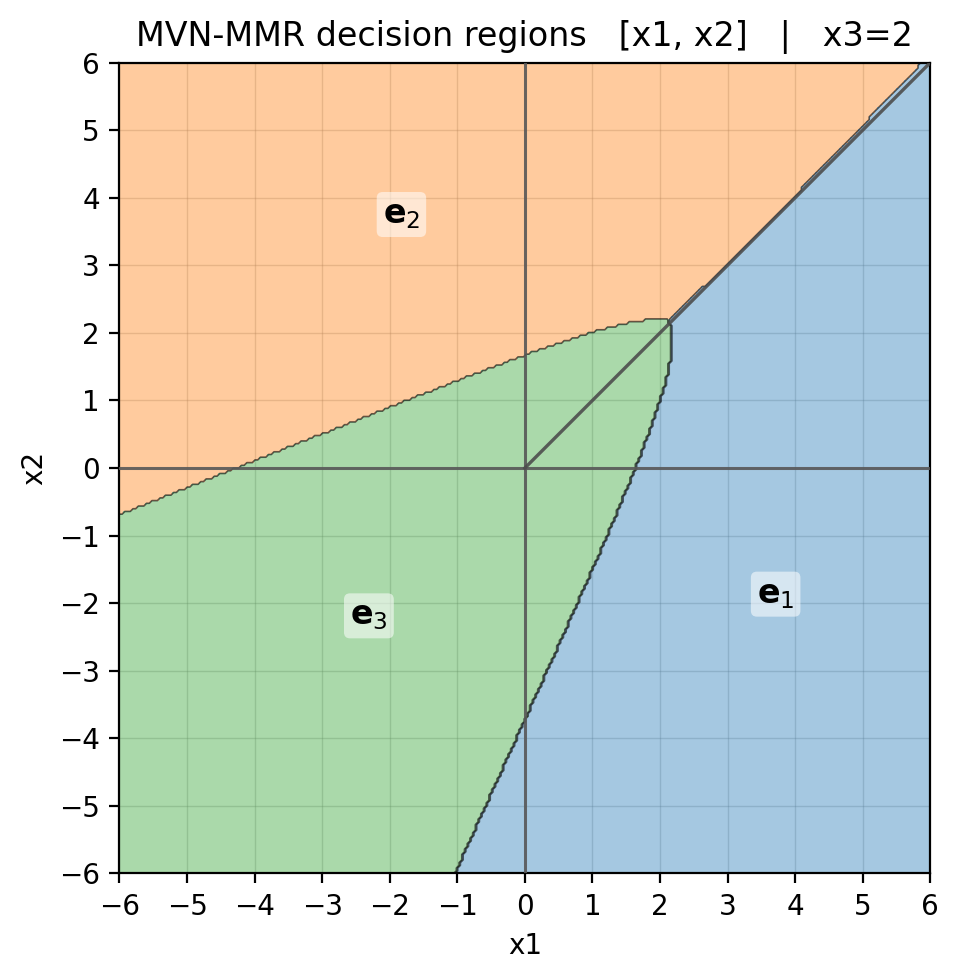}}
\par\end{centering}
\begin{centering}
\subfloat[\label{fig:,-fixed}$\left(x_{1},x_{3}\right)$, fixed $x_{2}=-2,0,2$]{
\centering{}\includegraphics[width=0.3\textwidth]{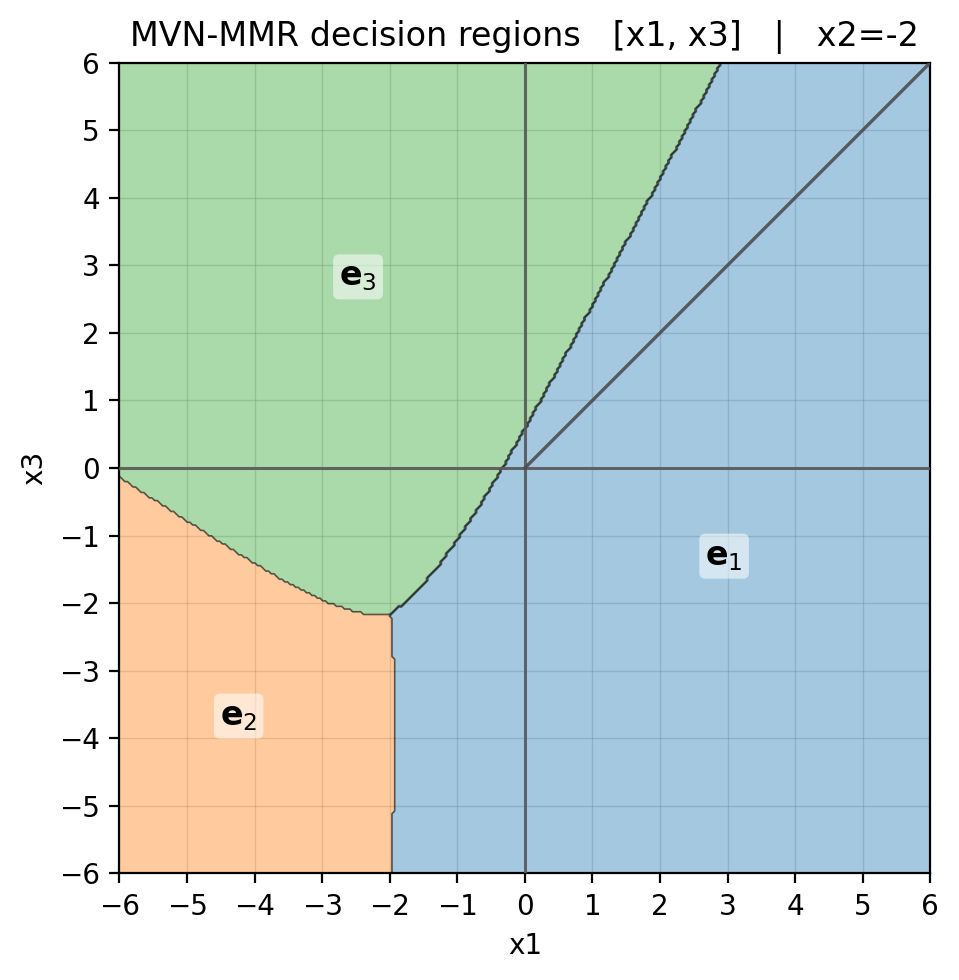}\includegraphics[width=0.3\textwidth]{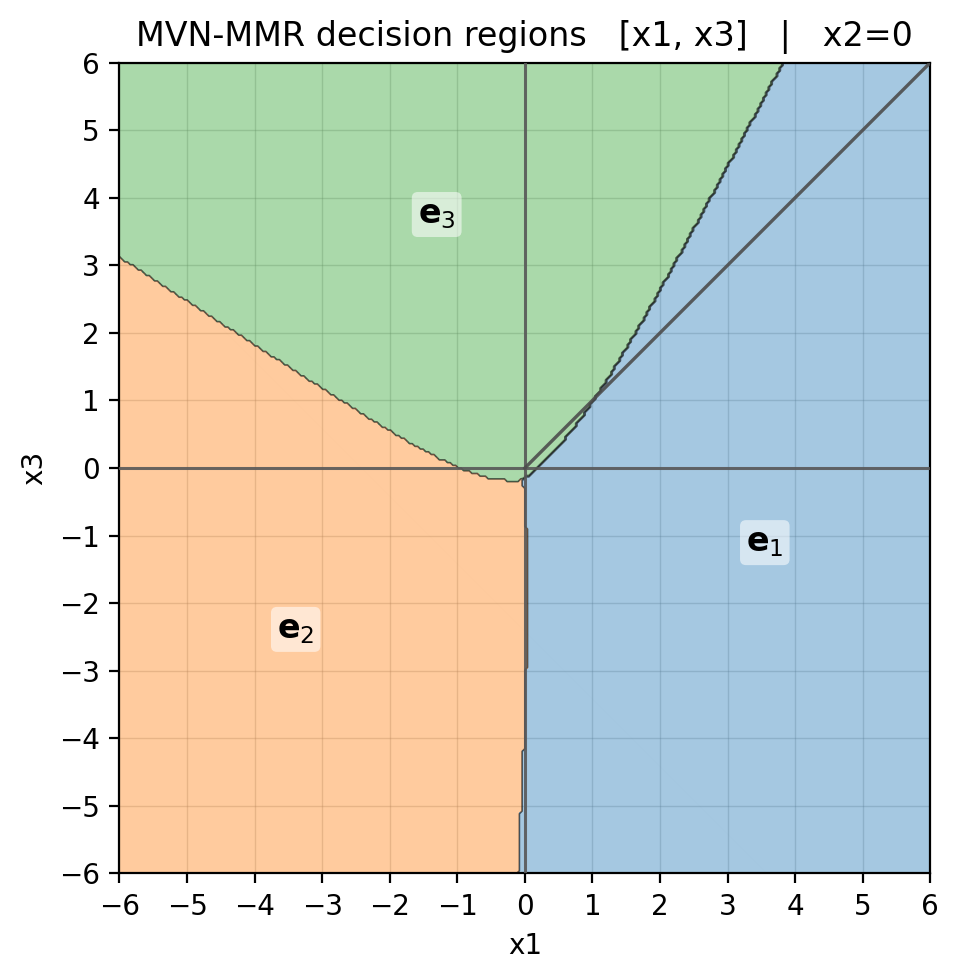}\includegraphics[width=0.3\textwidth]{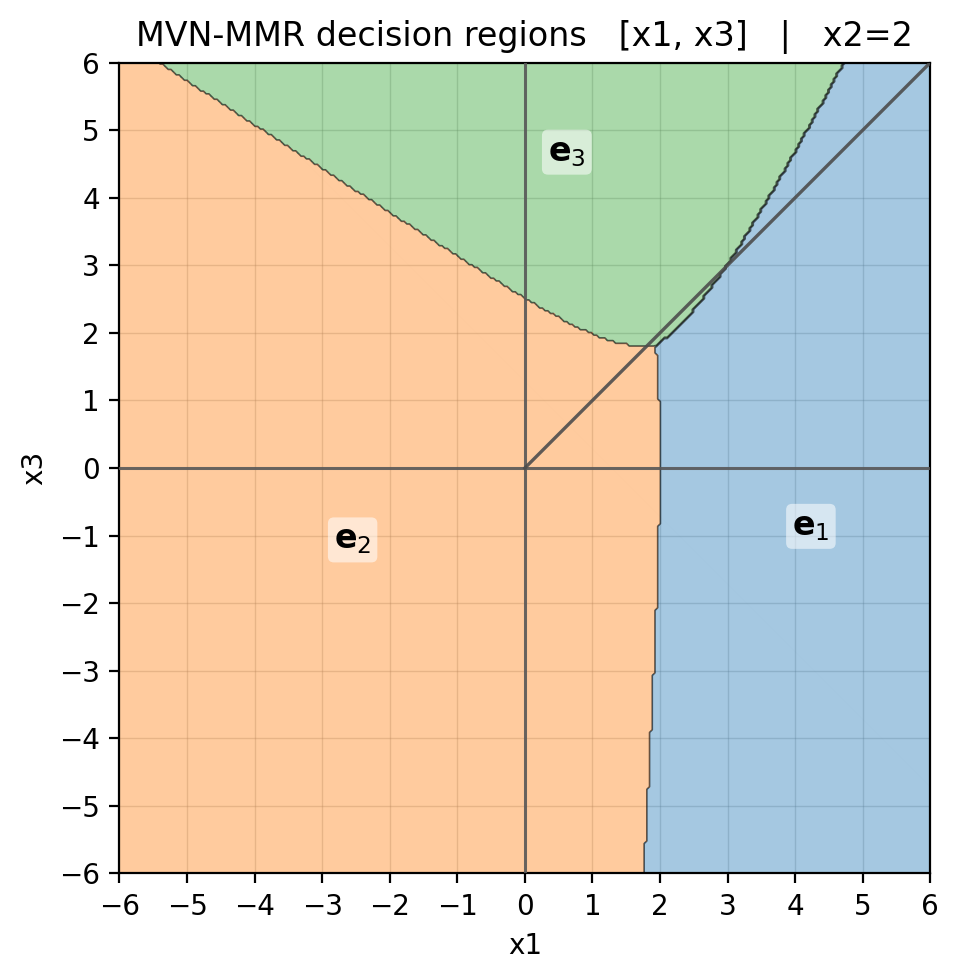}}
\par\end{centering}
\begin{centering}
\subfloat[\label{fig:,-fixed-1}$\left(x_{2},x_{3}\right)$, fixed $x_{1}=-2,0,2$]{
\centering{}\includegraphics[width=0.3\textwidth]{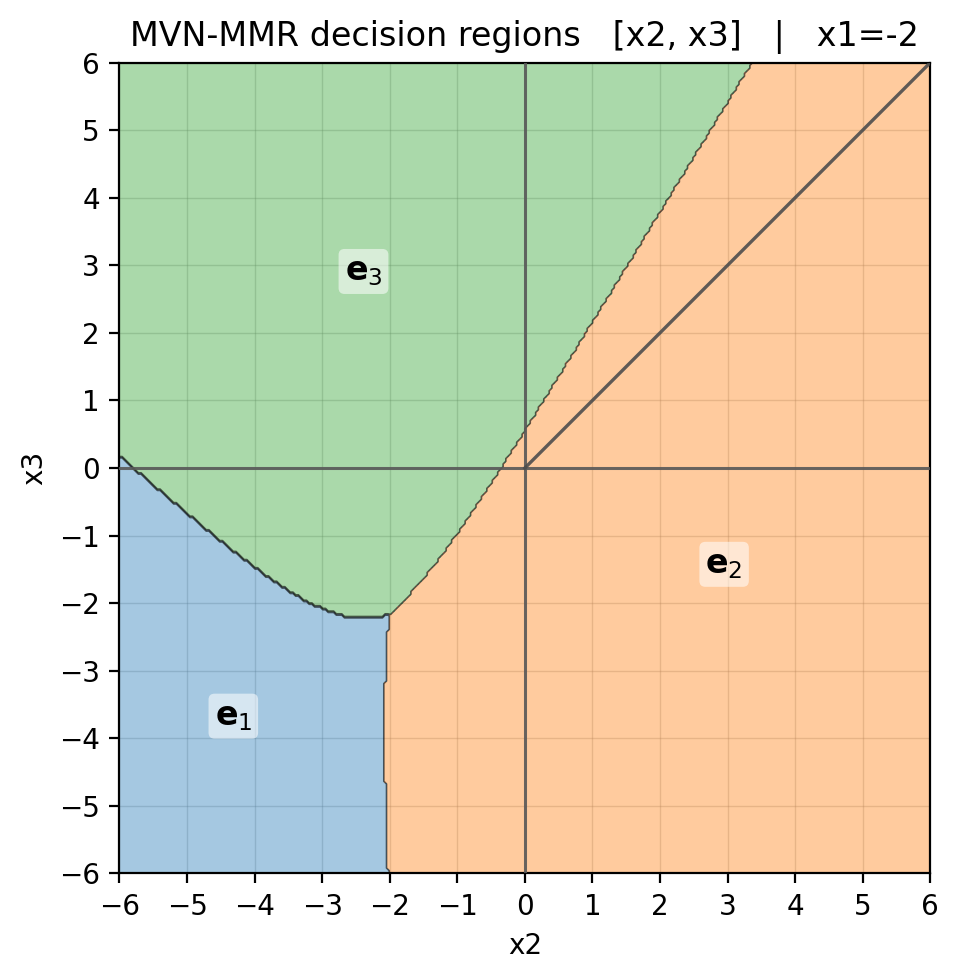}\includegraphics[width=0.3\textwidth]{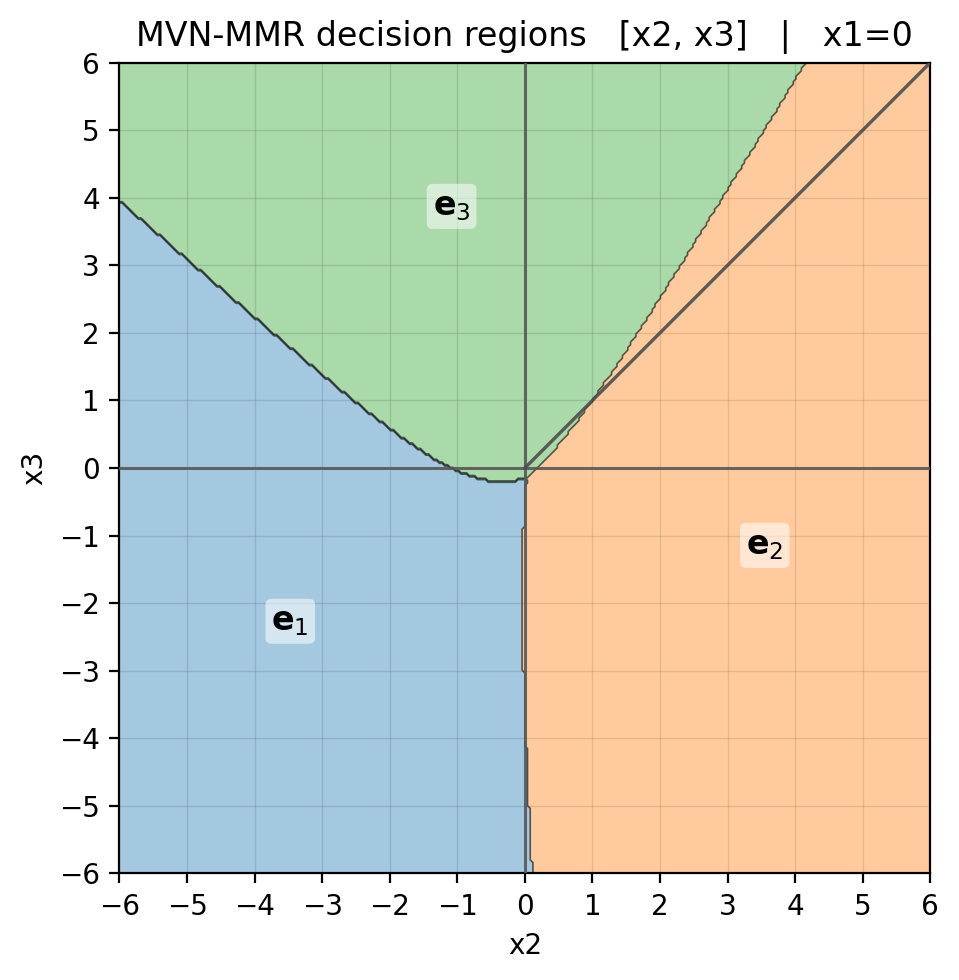}\includegraphics[width=0.3\textwidth]{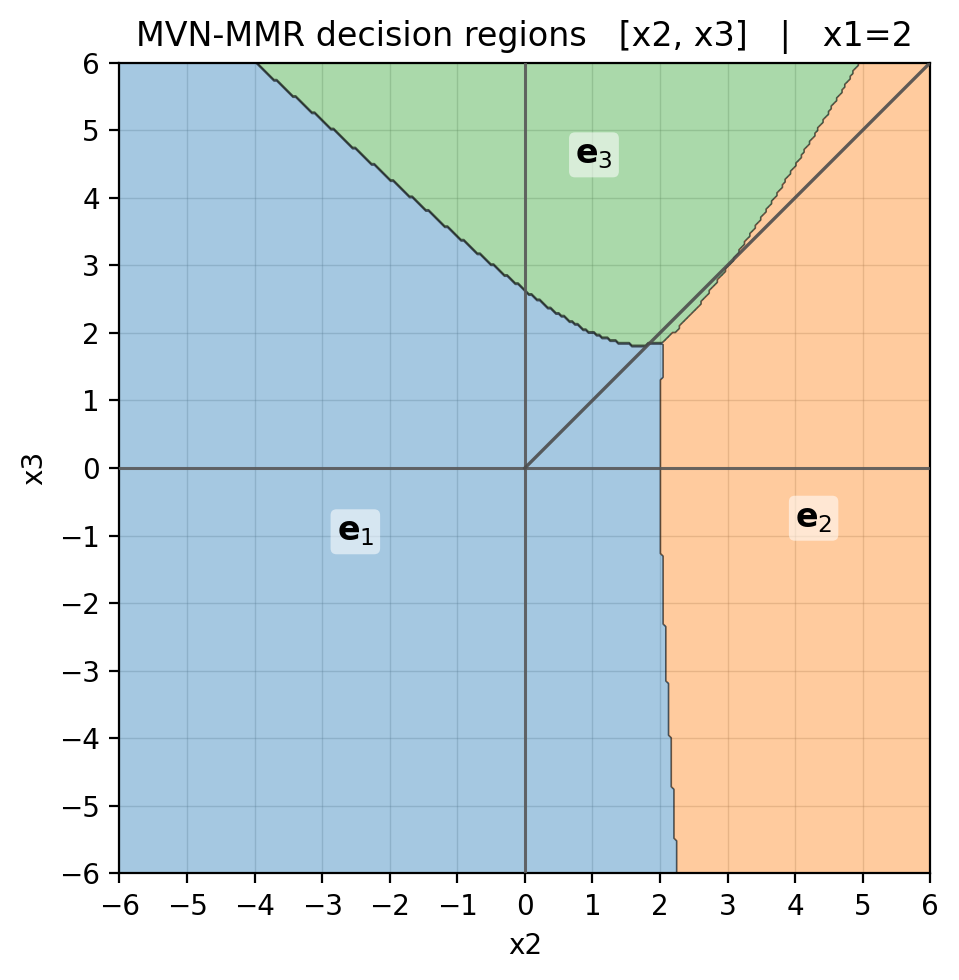}}
\par\end{centering}
{\small Note. The decision regions are calculated by numerically maximizing
the Bayes risk under the 3-point supported least-favorable prior with
the asymmetric covariance matrix $\bm{\Sigma}_{\text{asym}}$. $\mathbf{e}_{j}$
denotes the MMR decision region of the $j$-th arm.}{\small\par}
\end{figure}

\subsubsection{Discussion}

With $J=2$, the joint shift invariance discussed in Section \ref{subsec:Location-Invariance-of}
eliminates the nuisance location $t\mathbf{1}_{2}$ and reduces the
experiment to a single contrast $x_{2}-x_{1}$ with parameter $\theta_{2}-\theta_{1}$.
Under any location family, $x_{2}-x_{1}|\theta_{2}-\theta_{1}$ is
a one-dimensional location model with mean $\theta_{2}-\theta_{1}$
and variance $Var\left(x_{2}-x_{1}\right)$ that depends on the original
covariance matrix $\bm{\Sigma}$ but not on the common shift. Regret
for any rule depends on the location-parameter vector $\bm{\theta}$
only through $\theta_{2}-\theta_{1}$, and on the data only through
$x_{2}-x_{1}$. Therefore, the MMR-risk problem collapses to a one-dimensional
sign decision with symmetric noise, for which the unique MMR rule
is a threshold in $x_{2}-x_{1}$. Equalization of the risk at the
support points (Theorem \ref{thm:(Minimax-Theorem-for} (iv)) forces
the threshold for $x_{2}-x_{1}$ to be 0. In the original $\left(x_{1},x_{2}\right)$
plane, this is precisely the $45^{\circ}$ bisector $x_{1}=x_{2}$.
This is exactly what we see in Table \ref{tab:Least-favorable-prior-and}
and Figure \ref{fig:Two-arm-decision-boundaries}.

For $J\geq3$, the same invariance leaves a $\left(J-1\right)$-dimensional
vector of contrasts $\left(x_{2}-x_{1},...,x_{J}-x_{1}\right)$ with
a joint Gaussian law whose covariance inherits all heteroskedasticity
and correlations from $\bm{\Sigma}$. Deciding that arm $i$ is the
best requires simultaneously beating $J-1$ correlated comparisons
(e.g., arm\,3 must beat both arms\,1 and\,2), so there is no single
one-dimensional sufficient direction and no symmetry that pins each
pairwise threshold to zero. In the dual problem, the $i$-$j$ boundary
solves 
\[
h_{ij}\left(\mathbf{x}\right)\equiv\sum_{k=1}^{J}\pi^{k}p_{\bm{\theta}^{k}}\left(\mathbf{x}\right)\left(\theta_{i}^{k}-\theta_{j}^{k}\right)=0
\]
(see Theorem \ref{thm:(Characterization-of-the}). With two arms,
$h_{12}\left(\mathbf{x}\right)$ contains two exponential terms and
rearranges to a single linear condition, hence a straight line. With
three or more arms, $h_{ij}\left(\mathbf{x}\right)$ is a sum of at
least three exponentials with mixed signs, so the zero-level set is
generically curved. When $\bm{\Sigma}$ is homoskedastic/permutation-invariant,
symmetry still forces linear boundaries with the empirical success
rule (Figure \ref{fig:Three-arm-sym_decision-boundaries} and Table\,\ref{tab:sym_lfp_J3}).
However, when $\bm{\Sigma}$ breaks permutation invariance, nothing
compels symmetry across arms; the minimax rule bends and tilts the
boundaries to penalize high-variance arms, shrinking their decision
regions and requiring stronger evidence to select them. This is the
inward shift of the $\mathbf{e}_{3}$ region in Figure\,\ref{fig:Three-arm-asym-decision-boundaries}
and the asymmetric prior weights in Table\,\ref{tab:asym_lfp_J3}.

The prior probabilities $\pi^{k}$ are $1/J$ in Tables \ref{tab:2-sym},
\ref{tab:2-asym}, and \ref{tab:sym_lfp_J3}, whereas they deviate
from $1/J$ in Table \ref{tab:asym_lfp_J3}. Specifically, even when
the covariance matrix is asymmetric in two arms, $\pi^{k}=1/2$ for
both arms. The intuition for this follows from Theorem \ref{thm:(Minimax-Theorem-for}
(iv) combined with the Bayes risk expression $V=\sum_{k\in\mathcal{J}}\pi^{k}R\left(\bm{\theta}^{k},\bm{\delta}\right)$.
Theorem \ref{thm:(Minimax-Theorem-for} (iv) states that the pointwise
risks at the support points of the least-favorable prior, $R\left(\bm{\theta}^{k},\bm{\delta}\right)$,
equal $V$. Hence, the outer weighted sum in the Bayes risk expression
does not impose any restriction on the values of $\pi^{k}$. What
determines $\pi^{k}$ are the decision boundaries of the MMR decision
rule $\bm{\delta}$, which takes $\left(\pi^{k},\bm{\theta}^{k}\right)_{k\in\mathcal{J}}$
as its argument (see equation (\ref{eq:gaussian_integral})). When
the optimal MMR decision boundaries are symmetric, as in two-arm experiments
or in multi-arm experiments with $\bm{\Sigma}_{\text{sym}}$, the
prior probabilities satisfy $\pi^{k}=1/J$. 

However, when the optimal MMR decision boundaries are asymmetric,
the implied prior probabilities deviate from $1/J$. Intuitively,
because the MMR rule shrinks the selection region for a high-variance
arm, regret for the shrunk region would fall unless nature increases
both the distance of that arm’s support from the line $t\mathbf{1}_{J}$
and its prior mass. The optimized solution therefore loads more prior
weight on the high-variance arm and pushes its support deeper into
$\Theta^{i}$. This is exactly what we see in Table\,\ref{tab:asym_lfp_J3}:
the high-variance arm\,3 receives the largest prior weight ($\pi^{3}=0.433$)
and a more distant support point from the line $t\mathbf{1}_{J}$
($d\left(\bm{\theta}^{k},t\mathbf{1}_{J}\right)=1.679$), jointly
restoring risk equalization stated in Theorem \ref{thm:(Minimax-Theorem-for}
(iv).

\section{Concluding Remarks}\label{sec:Conclusion}

We characterized the MMR-risk-optimal best-population selection rule,
established local asymptotics when the mean rewards are only asymptotically
normal, and illustrated the resulting decision boundaries. Our local
asymptotic approximation allows the plug-in of mean-reward estimates
directly into the MVN-MMR decision boundaries, where a consistent
covariance estimator translates directly into decision boundaries
for mean rewards. These multi-arm plug-in decision boundaries generalize
the two-arm threshold MMR decision rules substantially. Our results
allow researchers and practitioners to compare multiple competing
policies by designing and implementing multi-arm RCTs within a unified
decision-theoretic framework.

\newpage{}

{\small\bibliographystyle{ecta}
\bibliography{multi_arm_ammr}
}{\small\par}

\newpage{}

\appendix

\part*{Appendix}

\section{Preliminary Results}\label{subsec:Preliminary-Results-for}
\begin{lem}[Properties of the Regret Loss]
\label{lem:A1}The regret loss $L\left(\bm{\phi}\left(\mathbf{x}\right),\bm{\theta}\right)$
satisfies the following:

(i) $L\left(\bm{\phi}\left(\mathbf{x}\right),\bm{\theta}\right)$
is bounded below by $0$ and bounded above by $2B$.

(ii) $L\left(\bm{\phi}\left(\mathbf{x}\right),\bm{\theta}\right)$
is 2-Lipschitz in $\bm{\theta}$.
\end{lem}
\begin{proof}
(i) is straightforward from the definition of the regret loss and
the parameter space $\Theta=\left[-B,B\right]^{J}$. (ii) follows
by
\begin{align}
\left|L\left(\bm{\phi}\left(\mathbf{x}\right),\bm{\theta}\right)-L\left(\bm{\phi}\left(\mathbf{x}\right),\bm{\theta}'\right)\right| & \leq\left|\max_{k\in\mathcal{J}}\theta_{k}-\max_{k\in\mathcal{J}}\theta_{k}'\right|+\sum_{j\in\mathcal{J}}\left|\theta_{j}-\theta_{j}'\right|\phi_{j}\left(\mathbf{x}\right)\nonumber \\
 & \leq2\left\Vert \bm{\theta}-\bm{\theta}'\right\Vert _{\infty}\leq2\left\Vert \bm{\theta}-\bm{\theta}'\right\Vert _{2},\label{eq:regretloss_bound}
\end{align}
where the last inequality used $\left\Vert \cdot\right\Vert _{\infty}\leq\left\Vert \cdot\right\Vert _{2}$.
\end{proof}
\begin{lem}[Continuity of Risk in $\bm{\theta}$]
\label{lem:A2}For any decision rule $\bm{\phi}$, the regret risk
$R\left(\bm{\theta},\bm{\phi}\right)$ defined in equation (\ref{eq:risk})
of Section \ref{subsec:Setup} is continuous in $\bm{\theta}$.
\end{lem}
\begin{proof}
$\left\{ P_{\bm{\theta}}\right\} _{\bm{\theta}\in\Theta}$ is a multivariate
location family with fully supported densities $p_{\bm{\theta}}$
in $\mathbb{R}^{J}$. Therefore, for any sequence $\bm{\theta}_{n}\rightarrow\bm{\theta}$,
\[
\int\left|p_{\bm{\theta}_{n}}-p_{\bm{\theta}}\right|d\mu\rightarrow0,
\]
which satisfies condition (A7) of \citet[ch.3]{Liese2008}. Furthermore,
the regret loss is bounded and continuous in $\bm{\theta}$ for any
fixed action $\mathbf{a}\in\bm{\Delta}$. Therefore, Proposition 3.25
of \citet[ch.3]{Liese2008} applies, implying that $R\left(\bm{\theta},\bm{\phi}\right)$
is continuous in $\bm{\theta}$.
\end{proof}
\begin{lem}[Uniform Convergence of MVN Densities Over $\bm{\theta}$ When $\bm{\Sigma}_{n}\rightarrow\bm{\Sigma}$]
\label{lem:A3}Let $\mathcal{S}=\left\{ \bm{\Sigma}\in\mathbb{S}_{++}^{J}:\text{\underbar{\ensuremath{\lambda}}}\mathbf{I}_{J}\preceq\bm{\Sigma}\preceq\bar{\lambda}\mathbf{I}_{J}\right\} $
be a compact set of p.d. covariance matrix for some $0<\underbar{\ensuremath{\lambda}}<\bar{\lambda}<\infty$.
Denote by $p_{\bm{\theta},\bm{\Sigma}}$ the $\mathcal{N}\left(\bm{\theta},\bm{\Sigma}\right)$
density. If $\bm{\Sigma}_{n}\rightarrow\bm{\Sigma}\in\mathcal{S}$,
then 
\[
\sup_{\bm{\theta}\in\Theta}\left\Vert p_{\bm{\theta},\bm{\Sigma}_{n}}-p_{\bm{\theta},\bm{\Sigma}}\right\Vert _{L^{1}\left(\mathbb{R}^{J}\right)}\rightarrow0\qquad\text{as }n\rightarrow\infty.
\]
\end{lem}
\begin{proof}
By the translation $\mathbf{u}=\mathbf{x}-\bm{\theta}$, 
\[
\int\left|p_{\bm{\theta},\bm{\Sigma}_{n}}\left(\mathbf{x}\right)-p_{\bm{\theta},\bm{\Sigma}}\left(\mathbf{x}\right)\right|d\mu\left(\mathbf{x}\right)=\int\left|p_{\mathbf{0},\bm{\Sigma}_{n}}\left(\mathbf{u}\right)-p_{\mathbf{0},\bm{\Sigma}}\left(\mathbf{u}\right)\right|d\mu\left(\mathbf{u}\right),
\]
so the supremum over $\bm{\theta}\in\Theta$ equals the case $\bm{\theta}=\mathbf{0}$.
For each fixed $\mathbf{u}\in\mathbb{R}^{J}$, $p_{\mathbf{0},\bm{\Sigma}_{n}}\left(\mathbf{u}\right)\rightarrow p_{\mathbf{0},\bm{\Sigma}}\left(\mathbf{u}\right)$
by continuity of the MVN density in $\bm{\Sigma}$. Using $\text{\underbar{\ensuremath{\lambda}}}\mathbf{I}_{J}\preceq\bm{\Sigma}\preceq\bar{\lambda}\mathbf{I}_{J}$,
$p_{\mathbf{0},\bm{\Sigma}_{n}}\left(\mathbf{u}\right)\leq\left(2\pi\right)^{-\frac{J}{2}}\text{\underbar{\ensuremath{\lambda}}}^{-\frac{J}{2}}\exp\left(-\frac{1}{2\bar{\lambda}}\left\Vert \mathbf{u}\right\Vert _{2}^{2}\right)\in L^{1}\left(\mathbb{R}^{J}\right)$.
Dominated convergence gives the desired result.
\end{proof}

\section{Proofs of the Main Results}

\subsection{Proof of Theorem \ref{thm:(Minimax-Theorem-for}}\label{subsec:Thm1_proof}

\paragraph{Proof of (i) and (ii):}

The proof of (i) and (ii) proceeds by verifying the conditions for
Theorem 2.4.2 of \citet{Blackwell1954} for the $S$-game formulation.
First, the action set for the selection problem is a finite set $\mathcal{J}$,
corresponding to Player I in the $S$-game. Next, Player II is nature,
who selects $\bm{\theta}\in\Theta$.\footnote{Note that this is opposite to the usual interpretation of $S$-games,
where Player I is nature and Player II is the statistician.} Lastly, the set of attainable risk across the parameter space $\Theta$
$S=\left\{ \left(\max_{k\in\mathcal{J}}\left\{ \theta_{k}\right\} -\theta_{j}\right)_{j\in\mathcal{J}}:\bm{\theta}\in\Theta\right\} $
is a closed and bounded (but not convex) subset of $\mathbb{R}^{J}$.
Theorem 2.4.2 of \citet{Blackwell1954} and the associated discussion
imply that the game is strictly determined and that there exists a
prior $\tilde{\pi}$ supported on at most $J$ distinct points of
$\Theta$ that attains the value $V<\infty$. $0\leq V\leq2B$ is
because the the regret loss $0\leq L\left(\bm{\phi}\left(\mathbf{x}\right),\bm{\theta}\right)\leq2B$,
so the associated risk is also bounded within $\left[0,2B\right]$.%

\paragraph{Proof of (iii): }

(iii) then follows from Proposition 3.56 of \citet{Liese2008}%
{} by invoking the fact that $\left(\tilde{\pi},\bm{\phi}\right)=\left(\pi,\bm{\delta}\right)$
is the saddle point of $\mathbb{E}_{\tilde{\pi}}\left[R\left(\bm{\theta},\bm{\phi}\right)\right]$.%

\paragraph{Proof of (iv):}

Part (i) asserts that $V=\sup_{\bm{\theta}\in\Theta}R\left(\bm{\theta},\bm{\delta}\right)$,
i.e., under the maximin-minimax rule $\bm{\delta}$, the maximum risk
that can be attained across the entire parameter space $\Theta$ is
$V$ so $R\left(\bm{\theta},\bm{\delta}\right)$ cannot exceed $V$.
Suppose, now, $R\left(\bm{\theta},\bm{\delta}\right)<V$ for some
$\bm{\theta}\in\text{supp}\left(\pi\right)$. Then, the Bayes risk
has to be strictly smaller than $V$, a contradiction to (\ref{eq:minimax_value}),
specifically to $\mathbb{E}_{\pi}\left[R\left(\bm{\theta},\bm{\delta}\right)\right]=V$.
This proves (iv).

\subsection{Proof of Lemma \ref{lem:(Pointwise-Bayes-Action}}\label{subsec:Lem1_proof}

The Bayes risk of any decision rule $\bm{\phi}$ under $\pi$ is:
\begin{eqnarray}
 &  & \sum_{k=1}^{m}\pi^{\left(k\right)}R\left(\bm{\theta}^{\left(k\right)},\bm{\phi}\right)\nonumber \\
 & = & \sum_{k=1}^{m}\pi^{\left(k\right)}\int L\left(\bm{\phi}\left(\mathbf{x}\right),\bm{\theta}^{\left(k\right)}\right)p_{\bm{\theta}^{\left(k\right)}}\left(\mathbf{x}\right)d\mu\left(\mathbf{x}\right)\nonumber \\
 & = & \sum_{k=1}^{m}\pi^{\left(k\right)}\int\left[\max_{r\in\mathcal{J}}\theta_{r}^{\left(k\right)}-\sum_{j\in\mathcal{J}}\theta_{j}^{\left(k\right)}\phi_{j}\left(\mathbf{x}\right)\right]p_{\bm{\theta}^{\left(k\right)}}\left(\mathbf{x}\right)d\mu\left(\mathbf{x}\right)\nonumber \\
 & = & \int\sum_{k=1}^{m}\left(\max_{r\in\mathcal{J}}\theta_{r}^{\left(k\right)}\right)\pi^{\left(k\right)}p_{\bm{\theta}^{\left(k\right)}}\left(\mathbf{x}\right)d\mu\left(\mathbf{x}\right)-\sum_{j\in\mathcal{J}}\int\phi_{j}\left(\mathbf{x}\right)\left(\sum_{k=1}^{m}\pi^{\left(k\right)}p_{\bm{\theta}^{\left(k\right)}}\left(\mathbf{x}\right)\theta_{j}^{\left(k\right)}\right)d\mu\left(\mathbf{x}\right),\label{eq:Bayes}
\end{eqnarray}
where we interchanged the finite sum and the integral because both
the regret loss and the risk are bounded by $2B$ (see Lemma \ref{lem:A1}).

Define, for each $\mathbf{x}$:
\begin{align}
C\left(\mathbf{x}\right) & =\sum_{k=1}^{m}\left(\max_{r\in\mathcal{J}}\theta_{r}^{\left(k\right)}\right)\pi^{\left(k\right)}p_{\bm{\theta}^{\left(k\right)}}\left(\mathbf{x}\right),\quad h_{i}\left(\mathbf{x}\right)=\sum_{k=1}^{m}\pi^{\left(k\right)}p_{\bm{\theta}^{\left(k\right)}}\left(\mathbf{x}\right)\theta_{i}^{\left(k\right)}.\label{eq:Bayesrisk}
\end{align}
By construction, $C$ and $h_{i}$ are measurable and integrable.
Simplifying the notation in (\ref{eq:Bayes}) gives the Bayes risk:
\[
\sum_{k=1}^{m}\pi^{\left(k\right)}R\left(\bm{\theta}^{\left(k\right)},\bm{\phi}\right)=\int\left(C\left(\mathbf{x}\right)-\sum_{j\in\mathcal{J}}\phi_{j}\left(\mathbf{x}\right)h_{j}\left(\mathbf{x}\right)\right)d\mu\left(\mathbf{x}\right).
\]

Next, consider the Bayes-risk difference between $\bm{\delta}$ defined
in (\ref{eq:argmax_delta_Bayes}) and any rule $\bm{\phi}$:
\begin{align}
 & \sum_{k=1}^{m}\pi^{\left(k\right)}R\left(\bm{\theta}^{\left(k\right)},\bm{\phi}\right)-\sum_{k=1}^{m}\pi^{\left(k\right)}R\left(\bm{\theta}^{\left(k\right)},\bm{\delta}\right)\nonumber \\
= & \int\left(\sum_{i\in\mathcal{J}}\delta_{i}\left(\mathbf{x}\right)h_{i}\left(\mathbf{x}\right)-\sum_{j\in\mathcal{J}}\phi_{j}\left(\mathbf{x}\right)h_{j}\left(\mathbf{x}\right)\right)d\mu\left(\mathbf{x}\right)\nonumber \\
= & \int\left(\left(\sum_{i\in\mathcal{J}}\delta_{i}\left(\mathbf{x}\right)\right)\left(\sum_{j\in\mathcal{J}}\phi_{j}\left(\mathbf{x}\right)\right)\frac{\sum_{i\in\mathcal{J}}\delta_{i}\left(\mathbf{x}\right)h_{i}\left(\mathbf{x}\right)-\sum_{j\in\mathcal{J}}\phi_{j}\left(\mathbf{x}\right)h_{j}\left(\mathbf{x}\right)}{\left(\sum_{i\in\mathcal{J}}\delta_{i}\left(\mathbf{x}\right)\right)\left(\sum_{j\in\mathcal{J}}\phi_{j}\left(\mathbf{x}\right)\right)}\right)d\mu\left(\mathbf{x}\right)\nonumber \\
= & \int\left(\sum_{i\in\mathcal{J}}\delta_{i}\left(\mathbf{x}\right)\left(\sum_{j\in\mathcal{J}}\phi_{j}\left(\mathbf{x}\right)\right)h_{i}\left(\mathbf{x}\right)-\sum_{j\in\mathcal{J}}\left(\sum_{i\in\mathcal{J}}\delta_{i}\left(\mathbf{x}\right)\right)\phi_{j}\left(\mathbf{x}\right)h_{j}\left(\mathbf{x}\right)\right)d\mu\left(\mathbf{x}\right)\nonumber \\
= & \int\left(\sum_{i\in\mathcal{J}}\sum_{j\in\mathcal{J}}\delta_{i}\left(\mathbf{x}\right)\phi_{j}\left(\mathbf{x}\right)\left(h_{i}\left(\mathbf{x}\right)-h_{j}\left(\mathbf{x}\right)\right)\right)d\mu\left(\mathbf{x}\right)\label{eq:br1}\\
\geq & 0,\nonumber 
\end{align}
where we repeatedly used $\sum_{i\in\mathcal{J}}\delta_{i}\left(\mathbf{x}\right)=\sum_{j\in\mathcal{J}}\phi_{j}\left(\mathbf{x}\right)=1$.
The last inequality is because, by construction, $\delta_{i}\left(\mathbf{x}\right)=1$
whenever $h_{i}\left(\mathbf{x}\right)-h_{j}\left(\mathbf{x}\right)\geq0$
and $\delta_{i}\left(\mathbf{x}\right)=0$ otherwise. Because $\bm{\phi}$
is arbitrary, $\bm{\delta}$ attains the minimum Bayes risk among
all rules. Therefore, $\bm{\delta}$ is Bayes with respect to $\pi$.

\subsection{Proof of Theorem \ref{thm:thm_2}}\label{subsec:Thm2_proof}

\paragraph{Proof of (i):}

We first establish $R\left(\bm{\theta},\bm{\delta}\right)<2B$. Consider
the uniform rule $\phi_{i}\left(\mathbf{x}\right)=\frac{1}{J}$ $\forall i\in\mathcal{J}$,
which ignores the data. Then, for all $\bm{\theta}\in\Theta$, 
\[
R\left(\bm{\theta},\bm{\phi}\right)=\frac{1}{J}\sum_{j=1}^{J}\left(\max_{k}\theta_{k}-\theta_{j}\right)\leq2B\left(1-\frac{1}{J}\right).
\]
Because the upper bound in the RHS does not depend on $\bm{\theta}$,
\[
\sup_{\bm{\theta}\in\Theta}R\left(\bm{\theta},\bm{\phi}\right)\leq2B\left(1-\frac{1}{J}\right).
\]
Invoking that $\bm{\delta}$ is the minimax rule that attains the
value $V$ yields:
\begin{equation}
V=\sup_{\bm{\theta}\in\Theta}R\left(\bm{\theta},\bm{\delta}\right)\leq\sup_{\bm{\theta}\in\Theta}R\left(\bm{\theta},\bm{\phi}\right)=2B\left(1-\frac{1}{J}\right)<2B.\label{eq:MMR_value_strictly_smaller_than_2B}
\end{equation}

Now assume, for a contradiction, that $\delta_{i}\left(\mathbf{x}\right)=0$
$\mu$-a.e. $\mathbf{x}$ for some fixed $i$. Consider the extreme
parameter $\bm{\theta}^{\circ}=B\mathbf{e}_{i}+B\left(\mathbf{e}_{i}-\mathbf{1}_{J}\right)$,
i.e., the $i$-th component is $B$ and all other components are $-B$.
Because $p_{\bm{\theta}^{\circ}}$ is absolutely continuous with respect
to $\mu$, we have $\mathbb{E}_{\bm{\theta}^{\circ}}\left[\delta_{i}\left(\mathbf{x}\right)\right]=0$.
Then we have:
\[
R\left(\bm{\theta}^{\circ},\bm{\delta}\right)=B-\left(B\cdot0+\sum_{j\neq i}\left(-B\right)\mathbb{E}_{\bm{\theta}^{\circ}}\left[\delta_{j}\left(\mathbf{x}\right)\right]\right)=2B,
\]
which contradicts (\ref{eq:MMR_value_strictly_smaller_than_2B}).
If $\delta_{i}\left(\mathbf{x}\right)=1$ $\mu$-a.e., then $\delta_{j}\left(\mathbf{x}\right)=0$
$\mu$-a.e. for all other $j\neq i$. Choosing $\bm{\theta}^{\circ}=B\mathbf{e}_{j}+B\left(\mathbf{e}_{j}-\mathbf{1}_{J}\right)$
for some $j\neq i$ leads to the same contradiction.

\paragraph{Proof of (ii):}

The existence of a least-favorable prior supported on at most $J$
distinct points is guaranteed by Theorem \ref{thm:(Minimax-Theorem-for}
(ii). Throughout, we denote by $\pi$ such least-favorable prior,
by $V$ the value of the game defined in (\ref{eq:minimax_value}),
and by $\bm{\delta}$ the Bayes rule that minimizes the average risk
under $\pi$. Again define the Bayes score functions as in Lemma \ref{lem:(Pointwise-Bayes-Action}:
$h_{j}\left(\mathbf{x}\right)\equiv\sum_{k=1}^{m}\pi^{\left(k\right)}p_{\bm{\theta}^{\left(k\right)}}\left(\mathbf{x}\right)\theta_{j}^{\left(k\right)}$.

\uline{Contradiction under the Bayes rule \mbox{$\bm{\delta}$}}.
Assume, for a contradiction, that $\pi$ assigns no support point
to some region $\Theta^{i}$. Below, we show that we can construct
a Bayes rule $\bm{\delta}$ with respect to the least-favorable prior
$\pi$ such that $\delta_{i}\left(\mathbf{x}\right)=0$ $\mu$-a.e.
$\mathbf{x}$, leading to a contradiction to part (i).

By the contradiction assumption, for every support point $\bm{\theta}^{\left(k\right)}$,
the $i$-th component $\theta_{i}^{\left(k\right)}$ is not strictly
the largest. Suppose $\bm{\alpha}=\left(\alpha_{j}\right)_{j\in\mathcal{J}}$
is a nonnegative weight vector that is independent of $\mathbf{x}$,
and let $w^{\left(k\right)}\left(\bm{\alpha}\right)\equiv\sum_{j\in\mathcal{J}}\alpha_{j}\left(\theta_{j}^{\left(k\right)}-\theta_{i}^{\left(k\right)}\right)$.
Multiplying $h_{j}\left(\mathbf{x}\right)$ by $\alpha_{j}$ and summing
over $j$ yields, for every $\mathbf{x}$, 
\begin{equation}
\sum_{j\in\mathcal{J}}\alpha_{j}h_{j}\left(\mathbf{x}\right)-\left(\sum_{j\in\mathcal{J}}\alpha_{j}\right)h_{i}\left(\mathbf{x}\right)=\sum_{k=1}^{m}\pi^{\left(k\right)}p_{\bm{\theta}^{\left(k\right)}}\left(\mathbf{x}\right)\sum_{j\in\mathcal{J}}\alpha_{j}\left(\theta_{j}^{\left(k\right)}-\theta_{i}^{\left(k\right)}\right)=\sum_{k=1}^{m}\pi^{\left(k\right)}p_{\bm{\theta}^{\left(k\right)}}\left(\mathbf{x}\right)w^{\left(k\right)}\left(\bm{\alpha}\right).\label{eq:master_diff}
\end{equation}
If we manage to choose $\bm{\alpha}$ such that 
\begin{align}
\sum_{j\in\mathcal{J}}\alpha_{j}=1,\quad\sum_{j\neq i}\alpha_{j}>0,\quad\text{and }w^{\left(k\right)}\left(\bm{\alpha}\right) & \geq0\quad\forall k=1,...,m,\label{eq:weight_ineq_alpha}
\end{align}
then (\ref{eq:master_diff}) is nonnegative for every $\mathbf{x}$,
because the densities $p_{\bm{\theta}^{\left(k\right)}}\left(\cdot\right)$
are nonnegative and not identically zero. In turn, we have
\[
\sum_{j\in\mathcal{J}}\alpha_{j}h_{j}\left(\mathbf{x}\right)-h_{i}\left(\mathbf{x}\right)\geq0,
\]
and rearranging gives: 
\begin{align*}
h_{i}\left(\mathbf{x}\right) & \leq\sum_{j\neq i}\left(\frac{\alpha_{j}}{1-\alpha_{i}}\right)h_{j}\left(\mathbf{x}\right)\leq\max_{j\neq i}h_{j}\left(\mathbf{x}\right),
\end{align*}
where the last inequality holds because a convex weighted sum with
$\sum_{j\neq i}\left(\frac{\alpha_{j}}{1-\alpha_{i}}\right)\leq1$
cannot exceed the maximum. Therefore, $i$ is never uniquely optimal.
Lemma \ref{lem:(Pointwise-Bayes-Action} asserts a Bayes rule $\bm{\delta}$
takes the form $\delta_{i}\left(\mathbf{x}\right)=1\left(i=\arg\max_{j\in\mathcal{J}}h_{j}\left(\mathbf{x}\right)\right)$
with any tie-breaking rule. Therefore, we may choose a tie-breaking
rule that never favors arm $i$ at any $\mathbf{x}$ so that $\delta_{i}\left(\mathbf{x}\right)=0$
$\mu$-a.e., which is still Bayes with respect to the least-favorable
prior. However, $\delta_{i}\left(\mathbf{x}\right)=0$ $\mu$-a.e.
is contradiction to part (i) of the present theorem.

\uline{Existence of \mbox{$\bm{\alpha}$}.} It remains to choose the
weight vector $\bm{\alpha}$ that satisfies (\ref{eq:weight_ineq_alpha}).
Let $m\leq J$ be the number of the support points $\left\{ \bm{\theta}^{\left(k\right)}\right\} _{k=1}^{m}$.
Define the $m\times J$ matrix $\mathbf{D}$ with the $k$-th row:
\[
\mathbf{D}_{\left(k,\cdot\right)}=\begin{pmatrix}\theta_{1}^{\left(k\right)}-\theta_{i}^{\left(k\right)} & \theta_{2}^{\left(k\right)}-\theta_{i}^{\left(k\right)} & \cdots & \theta_{J}^{\left(k\right)}-\theta_{i}^{\left(k\right)}\end{pmatrix}.
\]
Note $D_{\left(k,i\right)}=\theta_{i}^{\left(k\right)}-\theta_{i}^{\left(k\right)}=0$
for all $k$, i.e., $i$-th column of $\mathbf{D}$ is a zero vector
by construction.

Now let $\epsilon\in\left(0,1\right)$. Condition (\ref{eq:weight_ineq_alpha})
can be reformulated as the following system of linear inequalities:
\begin{equation}
\bm{\alpha}\geq\mathbf{0},\quad\mathbf{1}^{\top}\bm{\alpha}=1,\quad\alpha_{i}\leq1-\epsilon,\quad\mathbf{D}\bm{\alpha}\geq\mathbf{0}.\label{eq:linear_system}
\end{equation}
$\mathbf{1}^{\top}\bm{\alpha}=1$ with $\alpha_{i}\leq1-\epsilon$
enforces some $\alpha_{j}>0$ for $j\neq i$. Then, exactly one of
the following systems has a solution, by the Farkas' alternative linear
inequality \citep[Theorem 12]{Border2013}: (1) system (\ref{eq:linear_system});
or (2) there exists $\mathbf{y}\geq\mathbf{0}$, $\mu\geq0$, and
$\lambda\in\mathbb{R}$ such that the following (\ref{eq:linear_system_2})
and (\ref{eq:linear_system_2_2}) hold:\footnote{See Theorem \ref{thm:(Farkas's-Alternative-Linear} in Appendix \ref{subsec:Theorems-Used-in}.
These systems are obtained by taking $\mathbf{A}=\begin{pmatrix}\mathbf{1}^{\top} & 0\\
\mathbf{e}_{i}^{\top} & 1
\end{pmatrix}\in\mathbb{R}^{2\times\left(J+1\right)}$, $\mathbf{b}=\begin{pmatrix}1\\
1-\epsilon
\end{pmatrix}\in\mathbb{R}^{2}$, $\mathbf{B}=\begin{pmatrix}-\mathbf{D} & \mathbf{0}\end{pmatrix}\in\mathbb{R}^{m\times\left(J+1\right)}$,
$\mathbf{c}=\mathbf{0}\in\mathbb{R}^{m}$, $\mathbf{x}=\begin{pmatrix}\bm{\alpha}\\
t
\end{pmatrix}\in\mathbb{R}^{J+1}$, $\mathbf{q}=\mathbf{y}\in\mathbb{R}_{+}^{m}$, $\mathbf{p}=\begin{pmatrix}-\lambda\\
-\mu
\end{pmatrix}$ in Theorem \ref{thm:(Farkas's-Alternative-Linear}'s notation, where
$t$ is an auxiliary nonnegative slack variable we introduce to fit
into the system.}%
\begin{align}
\mathbf{D}^{\top}\mathbf{y}+\lambda\mathbf{1}+\mu\mathbf{e}_{i} & \leq\mathbf{0}\label{eq:linear_system_2}\\
\lambda+\mu\left(1-\epsilon\right) & >0.\label{eq:linear_system_2_2}
\end{align}

We show that (2) cannot hold. Because $i$-th column of $\mathbf{D}$
is the zero vector, the $i$-th inequality in (\ref{eq:linear_system_2})
reduces to $0+\lambda\cdot1+\mu\leq0$. Because $\mu\geq0$, this
forces $\lambda\leq-\mu$. Then,
\[
\lambda+\mu\left(1-\epsilon\right)\leq-\mu+\mu\left(1-\epsilon\right)=-\mu\epsilon\leq0,
\]
which contradicts (\ref{eq:linear_system_2_2}). Therefore, (2) cannot
have a solution, and (1) must have a solution $\bm{\alpha}$ that
satisfies (\ref{eq:weight_ineq_alpha}).

Therefore, any least-favorable prior of support size $\leq J$ must
put positive mass on every region $\Theta^{i}$. Since the regions
$\left\{ \Theta^{i}\right\} _{i\in\mathcal{J}}$ are disjoint and
we have at most $J$ support points, the support must contain exactly
one point $\bm{\theta}^{i}$ in each $\Theta^{i}$ with a strictly
positive mass $\pi^{i}$.

\subsection{Proof of Theorem \ref{thm:(Characterization-of-the}}\label{subsec:Thm3_proof}

\paragraph{Proof of (i):}

Applying Lemma \ref{lem:(Pointwise-Bayes-Action} to the prior $\pi=\left(\pi^{j},\bm{\theta}^{j}\right)_{j\in\mathcal{J}}$
implies $\bm{\delta}$ is Bayes with respect to $\pi$, and Theorem
\ref{thm:(Minimax-Theorem-for} then implies $\bm{\delta}$ is minimax.%

\paragraph{Proof of (ii):}

Define $h_{ij}\left(\mathbf{x}\right)=\sum_{k=1}^{J}\pi^{k}p^{k}\left(\mathbf{x}\right)\left(\theta_{i}^{k}-\theta_{j}^{k}\right)$.%
{} Let $\tilde{\bm{\delta}}$ be any Bayes decision rule under $\pi=\left(\pi^{j},\bm{\theta}^{j}\right)_{j\in\mathcal{J}}$.
Consider the risk difference between $\tilde{\bm{\delta}}$ and $\bm{\delta}$,
derived in (\ref{eq:br1}) in the proof of Lemma \ref{lem:(Pointwise-Bayes-Action}:
\begin{equation}
\sum_{k=1}^{J}\pi^{k}R\left(\bm{\theta}^{k},\tilde{\bm{\delta}}\right)-\sum_{k=1}^{J}\pi^{k}R\left(\bm{\theta}^{k},\bm{\delta}\right)=\sum_{i=1}^{J}\sum_{j=1}^{J}\int\delta_{i}\left(\mathbf{x}\right)\tilde{\delta}_{j}\left(\mathbf{x}\right)h_{ij}\left(\mathbf{x}\right)d\mu\left(\mathbf{x}\right).\label{eq:bayes_disagree}
\end{equation}
The integrand in the RHS of (\ref{eq:bayes_disagree}) is nonnegative
as shown in the proof of Lemma \ref{lem:(Pointwise-Bayes-Action}
by definition of $\bm{\delta}$.

Because of our lexicographical assignment of the boundary points to
$\Theta^{i}$s, every $\theta_{i}^{k}-\theta_{j}^{k}$ for $i<j$
is strictly negative when $k=j$ because the surface $\left\{ \bm{\theta}\in\Theta:\theta_{i}=\theta_{j}\right\} $
is already included in $\Theta^{i}$ so it cannot be included in $\Theta^{j}$.
Hence, the vector $\left(\theta_{i}^{k}-\theta_{j}^{k}\right)_{k=1}^{J}$
is not identically zero. Then, under our assumption that $\mu\left(\left\{ \mathbf{x}:h_{ij}\left(\mathbf{x}\right)=0\right\} \right)=0$,
which reduces to the restrictions imposed on the shape of $p^{k}\left(\mathbf{x}\right)$s,
the level set $H=\bigcup_{i<j}\left\{ \mathbf{x}:h_{ij}\left(\mathbf{x}\right)=0\right\} $
is a finite union of $\mu$-null sets in $\mathbb{R}^{J}$.

For any $\mathbf{x}\in\mathbb{R}^{J}\backslash H$, there is a unique
index $i^{*}$ such that $h_{i^{*}j}\left(\mathbf{x}\right)>0\:\forall j\neq i^{*}$.
Define the set $A=\left\{ \mathbf{x}\in\mathbb{R}^{J}\backslash H:\bm{\delta}\left(\mathbf{x}\right)\neq\tilde{\bm{\delta}}\left(\mathbf{x}\right)\right\} $,
where $\bm{\delta}$ and $\tilde{\bm{\delta}}$ disagree on $\mathbb{R}^{J}\backslash H$.
$\mathbf{x}\in A$ implies $\delta_{i^{*}}\left(\mathbf{x}\right)=1$
by construction of $A$, $\bm{\delta}$, and $i^{*}$; while $\tilde{\delta}_{i^{*}}\left(\mathbf{x}\right)<1$
because $\delta_{i^{*}}\left(\mathbf{x}\right)\neq\tilde{\delta}_{i}\left(\mathbf{x}\right)$.
This implies there exists at least one $j\neq i^{*}$ such that $\tilde{\delta}_{j}\left(\mathbf{x}\right)>0$.

Now consider the integrand in the RHS of (\ref{eq:bayes_disagree}).
For every $\mathbf{x}\in A$, all the three terms of $\delta_{i^{*}}\left(\mathbf{x}\right)\tilde{\delta}_{j}\left(\mathbf{x}\right)h_{i^{*}j}\left(\mathbf{x}\right)$
are strictly positive. Thus, if $\bm{\delta}\left(\mathbf{x}\right)\neq\tilde{\bm{\delta}}\left(\mathbf{x}\right)$
on a set of positive Lebesgue measure (so $\mu\left(A\right)>0$),
the integral is strictly positive, yielding $\sum_{k=1}^{J}\pi^{k}R\left(\bm{\theta}^{k},\tilde{\bm{\delta}}\right)>\sum_{k=1}^{J}\pi^{k}R\left(\bm{\theta}^{k},\bm{\delta}\right)$.
However, this is contradiction to the assumption that $\tilde{\bm{\delta}}$
is Bayes with respect to $\pi$. Therefore, we conclude that $\bm{\delta}\left(\mathbf{x}\right)=\tilde{\bm{\delta}}\left(\mathbf{x}\right)$
$\mu$-a.e..

\paragraph{Proof of (iii):}

Since $\bm{\delta}$ is unique Bayes with respect to the least-favorable
prior $\pi$, Theorem 1.4 of \citet[ch. 5]{lehmann1998theory} asserts
that it is also unique minimax.

\paragraph{Proof of (iv):}

$\bm{\delta}$ is admissible because $\bm{\delta}$ is unique Bayes
with respect to a proper prior distribution \citep[e.g.,][ch. 4, Theorem 8]{Berger1985}.

\subsection{Proof of Proposition \ref{prop:Scaling}}\label{subsec:Prop1_proof}

\paragraph{Preliminary results:}

Let $\mathbf{y}=\frac{1}{\sqrt{n}}\mathbf{x}$ and $\bm{\vartheta}=\frac{1}{\sqrt{n}}\bm{\theta}$.
Then, using the affine transformation gives:
\begin{align}
p_{\bm{\theta},\bm{\Sigma}_{\left[1\right]}}\left(\mathbf{x}\right) & =\left|\bm{\Sigma}_{\left[1\right]}\right|^{-\frac{1}{2}}p_{\mathbf{0},\mathbf{I}}\left(\bm{\Sigma}_{\left[1\right]}^{-\frac{1}{2}}\left(\mathbf{x}-\bm{\theta}\right)\right)=\left|\bm{\Sigma}_{\left[1\right]}\right|^{-\frac{1}{2}}p_{\mathbf{0},\mathbf{I}}\left(\bm{\Sigma}_{\left[n\right]}^{-\frac{1}{2}}\left(\frac{1}{\sqrt{n}}\mathbf{x}-\frac{1}{\sqrt{n}}\bm{\theta}\right)\right)\nonumber \\
 & =\frac{\left|\bm{\Sigma}_{\left[1\right]}\right|^{-\frac{1}{2}}}{\left|\bm{\Sigma}_{\left[n\right]}\right|^{-\frac{1}{2}}}p_{\frac{1}{\sqrt{n}}\bm{\theta},\bm{\Sigma}_{\left[n\right]}}\left(\frac{1}{\sqrt{n}}\mathbf{x}\right)=n^{-\frac{J}{2}}p_{\bm{\vartheta},\bm{\Sigma}_{\left[n\right]}}\left(\mathbf{y}\right).\label{eq:density_eq}
\end{align}
We note that the Jacobian $\left|\frac{d\mu\left(\mathbf{y}\right)}{d\mu\left(\mathbf{x}\right)}\right|=n^{-\frac{J}{2}}$,
which yields 
\begin{equation}
p_{\bm{\vartheta},\bm{\Sigma}_{\left[n\right]}}\left(\mathbf{y}\right)d\mu\left(\mathbf{y}\right)=p_{\bm{\theta},\bm{\Sigma}_{\left[1\right]}}\left(\mathbf{x}\right)d\mu\left(\mathbf{x}\right).\label{eq:density_measure_eq}
\end{equation}

Because the regret loss is homogeneous of degree 1 in $\bm{\theta}$,
for any rule $\bm{\phi}$ we have
\begin{equation}
L\left(\bm{\phi}\left(\mathbf{x}\right),\bm{\theta}\right)=\sqrt{n}L\left(\bm{\phi}\left(\mathbf{x}\right),\frac{1}{\sqrt{n}}\bm{\theta}\right)=\sqrt{n}L\left(\bm{\phi}\left(\mathbf{x}\right),\bm{\vartheta}\right).\label{eq:loss_rescaling}
\end{equation}
Define the scaled rule 
\[
\bm{\varphi}\left(\mathbf{y}\right)=\bm{\phi}\left(\sqrt{n}\mathbf{y}\right)=\bm{\phi}\left(\mathbf{x}\right)=\bm{\varphi}\left(\frac{1}{\sqrt{n}}\mathbf{x}\right).
\]

\paragraph{Proof of (i):}

The risk-scaling identity is derived as follows:
\begin{align}
R_{\left[1\right]}\left(\bm{\theta},\bm{\phi}\right) & =\int L\left(\bm{\phi}\left(\mathbf{x}\right),\bm{\theta}\right)p_{\bm{\theta},\bm{\Sigma}_{\left[1\right]}}\left(\mathbf{x}\right)d\mu\left(\mathbf{x}\right)\nonumber \\
 & =\int\sqrt{n}L\left(\bm{\phi}\left(\mathbf{x}\right),\frac{1}{\sqrt{n}}\bm{\theta}\right)p_{\frac{1}{\sqrt{n}}\bm{\theta},\bm{\Sigma}_{\left[n\right]}}\left(\frac{1}{\sqrt{n}}\mathbf{x}\right)d\mu\left(\mathbf{x}\right)\nonumber \\
 & =\sqrt{n}\int L\left(\bm{\phi}\left(\sqrt{n}\mathbf{y}\right),\bm{\vartheta}\right)p_{\bm{\vartheta},\bm{\Sigma}_{\left[n\right]}}\left(\mathbf{y}\right)d\mu\left(\mathbf{y}\right)\label{eq:scaling}\\
 & =\sqrt{n}\int L\left(\bm{\varphi}\left(\mathbf{y}\right),\bm{\vartheta}\right)p_{\bm{\vartheta},\bm{\Sigma}_{\left[n\right]}}\left(\mathbf{y}\right)d\mu\left(\mathbf{y}\right)\label{eq:varphi}\\
 & =\sqrt{n}R_{\left[n\right]}\left(\bm{\vartheta},\bm{\varphi}\right),\label{eq:risk_scaling}
\end{align}
where we used $\mathbf{x}=\sqrt{n}\mathbf{y}$ with (\ref{eq:density_measure_eq})
in (\ref{eq:scaling}) and $\bm{\varphi}\left(\mathbf{y}\right)=\bm{\phi}\left(\sqrt{n}\mathbf{y}\right)$
in (\ref{eq:varphi}). Because $\bm{\vartheta}=\frac{1}{\sqrt{n}}\bm{\theta}$,
$R_{\left[n\right]}\left(\bm{\vartheta},\bm{\varphi}\right)$ is defined
on the rescaled parameter set $\Theta_{\left[n\right]}=\frac{1}{\sqrt{n}}\left[-B,B\right]^{J}$.

\paragraph{Proof of (ii):}

Let $\Phi\left(\bm{\phi}\right)=\bm{\varphi}$, which is a bijection
with the rescaled parameter set $\Theta_{\left[n\right]}=\frac{1}{\sqrt{n}}\Theta_{\left[1\right]}$.
Then, $R_{\left[1\right]}\left(\bm{\theta},\bm{\phi}\right)=\sqrt{n}R_{\left[n\right]}\left(\bm{\vartheta},\Phi\left(\bm{\phi}\right)\right)$.
Let 
\[
F_{\left[1\right]}\left(\bm{\phi}\right)=\sup_{\bm{\theta}}R_{\left[1\right]}\left(\bm{\theta},\bm{\phi}\right),\quad F_{\left[n\right]}\left(\bm{\phi}\right)=\sup_{\bm{\vartheta}}R_{\left[n\right]}\left(\bm{\vartheta},\Phi\left(\bm{\phi}\right)\right).
\]
Because $R_{\left[1\right]}\left(\bm{\theta},\bm{\phi}\right)=\sqrt{n}R_{\left[n\right]}\left(\frac{1}{\sqrt{n}}\bm{\theta},\Phi\left(\bm{\phi}\right)\right)$
for any $\bm{\theta}$ and $\bm{\phi}$,
\[
F_{\left[1\right]}\left(\bm{\phi}\right)=\sup_{\bm{\theta}}R_{\left[1\right]}\left(\bm{\theta},\bm{\phi}\right)=\sup_{\bm{\theta}}\sqrt{n}R_{\left[n\right]}\left(\frac{1}{\sqrt{n}}\bm{\theta},\Phi\left(\bm{\phi}\right)\right)=\sqrt{n}F_{\left[n\right]}\left(\Phi\left(\bm{\phi}\right)\right).
\]
Taking infima by invoking the fact that $\Phi\left(\cdot\right)$
is a bijection gives:
\[
V_{\left[1\right]}\equiv\inf_{\bm{\phi}}\sup_{\bm{\theta}}R_{\left[1\right]}\left(\bm{\theta},\bm{\phi}\right)=\inf_{\bm{\phi}}F_{\left[1\right]}\left(\bm{\phi}\right)=\sqrt{n}\inf_{\bm{\phi}}F_{\left[n\right]}\left(\Phi\left(\bm{\phi}\right)\right)=\sqrt{n}\inf_{\bm{\varphi}}\sup_{\bm{\vartheta}}R_{\left[n\right]}\left(\bm{\vartheta},\bm{\varphi}\right)\equiv\sqrt{n}V_{\left[n\right]}.
\]
Therefore, $V_{\left[n\right]}=\frac{1}{\sqrt{n}}V_{\left[1\right]}$.

\paragraph{Proof of (iii) and (iv):}

Let $\bm{\delta}_{\left[1\right]}$ be the minimax/maximin rule when
$m=1$ and $\pi_{\left[1\right]}=\left(\pi_{\left[1\right]}^{k},\bm{\theta}_{\left[1\right]}^{k}\right)_{k\in\mathcal{J}}$
is the associated least-favorable prior. Let $\bm{\delta}_{\left[n\right]}\left(\mathbf{y}\right)=\bm{\delta}_{\left[1\right]}\left(\sqrt{n}\mathbf{y}\right)=\bm{\delta}_{\left[1\right]}\left(\mathbf{x}\right)$
be our candidate rule. In the following, we first show that $\bm{\delta}_{\left[n\right]}\left(\mathbf{y}\right)=\bm{\delta}_{\left[1\right]}\left(\sqrt{n}\mathbf{y}\right)=\bm{\delta}_{\left[1\right]}\left(\mathbf{x}\right)$
is Bayes with respect to $\left(\pi^{k},\frac{1}{\sqrt{n}}\bm{\theta}_{\left[1\right]}^{k}\right)_{k\in\mathcal{J}}$,
and then the pair $\left(\bm{\delta}_{\left[n\right]},\left(\pi^{k},\frac{1}{\sqrt{n}}\bm{\theta}_{\left[1\right]}^{k}\right)_{k\in\mathcal{J}}\right)$
achieves the minimax value $V_{\left[n\right]}=\frac{1}{\sqrt{n}}V_{\left[1\right]}$
at each support point $\bm{\theta}_{\left[n\right]}^{k}\equiv\frac{1}{\sqrt{n}}\bm{\theta}_{\left[1\right]}^{k}$,
to conclude $\bm{\delta}_{\left[n\right]}$ is minimax and $\left(\pi^{k},\frac{1}{\sqrt{n}}\bm{\theta}_{\left[1\right]}^{k}\right)_{k\in\mathcal{J}}$
is least-favorable.

By Theorem \ref{thm:(Characterization-of-the}, the Bayes rule under
the priors $\pi_{\left[m\right]}=\left(\pi_{\left[m\right]}^{k},\bm{\theta}_{\left[m\right]}^{k}\right)_{k\in\mathcal{J}}$
when $m=1$ and $m=n$ are, respectively:
\begin{align}
\delta_{\left[1\right],i}\left(\mathbf{x}\right) & =\begin{cases}
1 & \text{if }\sum_{k=1}^{J}\pi^{k}\theta_{\left[1\right],i}^{k}p_{\bm{\theta}_{\left[1\right]}^{k},\bm{\Sigma}_{\left[1\right]}}\left(\mathbf{x}\right)>\sum_{k=1}^{J}\pi^{k}\theta_{\left[1\right],j}^{k}p_{\bm{\theta}_{\left[1\right]}^{k},\bm{\Sigma}_{\left[1\right]}}\left(\mathbf{x}\right)\:\forall j\neq i\\
0 & \text{otherwise}
\end{cases}\qquad\forall\mathbf{x}\in\mathbb{R}^{J}\label{eq:bayes_under_1}\\
\delta_{\left[n\right],i}\left(\mathbf{y}\right) & =\begin{cases}
1 & \text{if }\sum_{k=1}^{J}\pi^{k}\theta_{\left[n\right],i}^{k}p_{\bm{\theta}_{\left[n\right]}^{k},\bm{\Sigma}_{\left[n\right]}}\left(\mathbf{y}\right)>\sum_{k=1}^{J}\pi^{k}\theta_{\left[n\right],j}^{k}p_{\bm{\theta}_{\left[n\right]}^{k},\bm{\Sigma}_{\left[n\right]}}\left(\mathbf{y}\right)\:\forall j\neq i\\
0 & \text{otherwise}
\end{cases}\qquad\forall\mathbf{y}\in\mathbb{R}^{J}.\label{eq:bayes_under_n}
\end{align}
Suppose $\pi_{\left[n\right]}=\left(\pi^{k},\bm{\theta}_{\left[n\right]}^{k}\right)_{k\in\mathcal{J}}=\left(\pi^{k},\frac{1}{\sqrt{n}}\bm{\theta}_{\left[1\right]}^{k}\right)_{k\in\mathcal{J}}$.
Then, for each $i\in\mathcal{J}$,
\begin{align}
\sum_{k=1}^{J}\pi^{k}\theta_{\left[1\right],i}^{k}p_{\bm{\theta}_{\left[1\right]}^{k},\bm{\Sigma}_{\left[1\right]}}\left(\mathbf{x}\right) & =n^{-\frac{J}{2}}\sum_{k=1}^{J}\pi^{k}\theta_{\left[1\right],i}^{k}p_{\frac{1}{\sqrt{n}}\bm{\theta}_{\left[1\right]}^{k},\bm{\Sigma}_{\left[n\right]}}^{k}\left(\frac{1}{\sqrt{n}}\mathbf{x}\right)\label{eq:used_density_eq}\\
 & =n^{\frac{1}{2}-\frac{J}{2}}\sum_{k=1}^{J}\pi^{k}\left(\frac{\theta_{\left[1\right],i}^{k}}{\sqrt{n}}\right)p_{\frac{1}{\sqrt{n}}\bm{\theta}_{\left[1\right]}^{k},\bm{\Sigma}_{\left[n\right]}}^{k}\left(\frac{1}{\sqrt{n}}\mathbf{x}\right)\label{eq:assumption_imposed}\\
 & =n^{\frac{1}{2}-\frac{J}{2}}\sum_{k=1}^{J}\pi^{k}\theta_{\left[n\right],i}^{k}p_{\bm{\theta}_{\left[n\right]}^{k},\bm{\Sigma}_{\left[n\right]}}^{k}\left(\frac{1}{\sqrt{n}}\mathbf{x}\right),\label{eq:assumption_imposed2}
\end{align}
where (\ref{eq:used_density_eq}) used (\ref{eq:density_eq}) and
(\ref{eq:assumption_imposed2}) invoked $\bm{\theta}_{\left[n\right]}^{k}=\frac{1}{\sqrt{n}}\bm{\theta}_{\left[1\right]}^{k}$.
Therefore, $\delta_{\left[n\right],i}\left(\frac{1}{\sqrt{n}}\mathbf{x}\right)=1$
iff $\delta_{\left[1\right],i}\left(\mathbf{x}\right)=1$ for all
$i$, i.e., $\delta_{\left[1\right],i}\left(\sqrt{n}\mathbf{y}\right)=\delta_{\left[n\right],i}\left(\mathbf{y}\right)$
for all $i$.

Finally, the risk-scaling identity (\ref{eq:risk_scaling}) gives,
$\forall k\in\mathcal{J}$,
\begin{align*}
R_{\left[n\right]}\left(\frac{1}{\sqrt{n}}\bm{\theta}_{\left[1\right]}^{k},\bm{\delta}_{\left[n\right]}\right) & =\frac{1}{\sqrt{n}}R_{\left[1\right]}\left(\bm{\theta}_{\left[1\right]}^{k},\bm{\delta}_{\left[1\right]}\right)=\frac{1}{\sqrt{n}}V_{\left[1\right]}=V_{\left[n\right]}
\end{align*}
where the last equality is by part (ii) of the current proposition
already established. Corollary 1.6 of \citet[ch. 5]{lehmann1998theory}
yields $\bm{\delta}_{\left[n\right]}$ is minimax and the associated
$\pi_{\left[n\right]}=\left(\pi^{k},\bm{\theta}_{\left[n\right]}^{k}\right)_{k\in\mathcal{J}}=\left(\pi^{k},\frac{1}{\sqrt{n}}\bm{\theta}_{\left[1\right]}^{k}\right)_{k\in\mathcal{J}}$
is least-favorable.

\subsection{Proof of Proposition \ref{prop:Prop2}}\label{subsec:Proof-of-Proposition}

The key challenge in establishing the proposition is the least-favorable
prior $\pi_{\bm{\Sigma}_{n}}=\left(\pi_{\bm{\Sigma}_{n}}^{k},\bm{\theta}_{\bm{\Sigma}_{n}}^{k}\right)_{k\in\mathcal{J}}$
moves along $\bm{\Sigma}_{n}$, without a guarantee that $\pi_{\bm{\Sigma}_{n}}$
converges even if $\bm{\Sigma}_{n}$ converges. In the following,
we invoke compactness of $\bm{\Delta}\times\Theta^{J}$ to extract
a convergent subsequence of least-favorable priors, prove the associated
Bayes rules converge pointwise a.e. along the convergent subsequence,
and show that the limiting Bayes rule is minimax with respect to $\bm{\Sigma}$.
The a.e. equality of the limiting rule $\bm{\delta}_{\bm{\Sigma}_{n}}$
with $\bm{\delta}_{\bm{\Sigma}}$ follows by the a.e. uniqueness of
the MVN-MMR rule (Corollary of Theorem \ref{thm:(Characterization-of-the}).

\paragraph{Step 1. Extraction of a convergent subsequence of the least-favorable
prior and the associated maximin decision rule:}

Let $\pi_{\bm{\Sigma}_{n}}=\left(\pi_{\bm{\Sigma}_{n}}^{k},\bm{\theta}_{\bm{\Sigma}_{n}}^{k}\right)_{k\in\mathcal{J}}$
be any sequence of least-favorable prior indexed by $\bm{\Sigma}_{n}$
supported on $J$ distinct points where each $\bm{\theta}_{\bm{\Sigma}_{n}}^{k}\in\Theta^{k}$.
For the sequence of priors $\pi_{\bm{\Sigma}_{n}}$, pick any accumulation
point $\bar{\pi}=\left(\bar{\pi}^{k},\bar{\bm{\theta}}^{k}\right)_{k\in\mathcal{J}}$
and extract the associated subsequence $\pi_{\bm{\Sigma}_{l}}=\left(\pi_{\bm{\Sigma}_{l}}^{k},\bm{\theta}_{\bm{\Sigma}_{l}}^{k}\right)_{k\in\mathcal{J}}$
such that $\pi_{\bm{\Sigma}_{l}}\rightarrow\bar{\pi}$, where $l$
denotes the index of the convergent subsequence for the priors. This
convergence subsequence extraction is always possible because $\bm{\Delta}\times\Theta\times...\times\Theta$
is compact.

Note, the limit $\bar{\pi}$ need not be least-favorable for $\bm{\Sigma}$;
in the subsequent steps, we establish that $\bar{\pi}$ is actually
least-favorable for $\bm{\Sigma}$.

\paragraph{Step 2. Continuity of MVN Bayes risk in $\pi_{\bm{\Sigma}_{l}}$:}

For any $l$, for the least-favorable prior $\pi_{\bm{\Sigma}_{l}}=\left(\pi_{\bm{\Sigma}_{l}}^{k},\bm{\theta}_{\bm{\Sigma}_{l}}^{k}\right)_{k\in\mathcal{J}}$,
define: 
\begin{align*}
C_{\bm{\Sigma}_{l}}\left(\mathbf{x}\right) & =\sum_{k=1}^{J}\left(\max_{r\in\mathcal{J}}\theta_{\bm{\Sigma}_{l},r}^{k}\right)\pi_{\bm{\Sigma}_{l}}^{k}p_{\bm{\theta}_{\bm{\Sigma}_{l}}^{k},\bm{\Sigma}_{l}}\left(\mathbf{x}\right)\\
h_{\bm{\Sigma}_{l},i}\left(\mathbf{x}\right) & =\sum_{k=1}^{J}\theta_{\bm{\Sigma}_{l},i}^{k}\pi_{\bm{\Sigma}_{l}}^{k}p_{\bm{\theta}_{\bm{\Sigma}_{l}}^{k},\bm{\Sigma}_{l}}\left(\mathbf{x}\right).
\end{align*}
The MVN-MMR decision rule is:
\[
\delta_{\bm{\Sigma}_{l},i}\left(\mathbf{x}\right)\equiv1\left(h_{\bm{\Sigma}_{l},i}\left(\mathbf{x}\right)\geq h_{\bm{\Sigma}_{l},j}\left(\mathbf{x}\right)\quad\forall j\neq i\text{ with arbitrary tie-breaking}\right).
\]
Define the Bayes risk under the MVN-MMR decision rule, which attains
the value $V_{\bm{\Sigma}_{l}}$:
\begin{align*}
\mathcal{R}\left(\left(\pi_{\bm{\Sigma}_{l}}^{k},\bm{\theta}_{\bm{\Sigma}_{l}}^{k}\right)_{k\in\mathcal{J}},\bm{\Sigma}_{l}\right)=\sum_{k=1}^{J}\pi_{\bm{\Sigma}_{l}}^{k}R\left(\bm{\theta}_{\bm{\Sigma}_{l}}^{k},\bm{\delta}_{\bm{\Sigma}_{l}}\right) & =\int\left(C_{\bm{\Sigma}_{l}}\left(\mathbf{x}\right)-\max_{i\in\mathcal{J}}h_{\bm{\Sigma}_{l},i}\left(\mathbf{x}\right)\right)d\mu\left(\mathbf{x}\right).
\end{align*}

We claim the map $\left(\left(\pi_{\bm{\Sigma}_{l}}^{k},\bm{\theta}_{\bm{\Sigma}_{l}}^{k}\right)_{k\in\mathcal{J}},\bm{\Sigma}_{l}\right)\mapsto\mathcal{R}\left(\left(\pi_{\bm{\Sigma}_{l}}^{k},\bm{\theta}_{\bm{\Sigma}_{l}}^{k}\right)_{k\in\mathcal{J}},\bm{\Sigma}_{l}\right)$
is continuous on the compact domain $\bm{\Delta}\times\Theta\times...\times\Theta\times\mathcal{S}$.
To see why, whenever $\bm{\Sigma}_{l}\rightarrow\bm{\Sigma}$, 
\begin{align*}
 & \left|\mathcal{R}\left(\left(\pi_{\bm{\Sigma}_{l}}^{k},\bm{\theta}_{\bm{\Sigma}_{l}}^{k}\right)_{k\in\mathcal{J}},\bm{\Sigma}_{l}\right)-\mathcal{R}\left(\left(\pi_{\bm{\Sigma}}^{k},\bm{\theta}_{\bm{\Sigma}}^{k}\right)_{k\in\mathcal{J}},\bm{\Sigma}\right)\right|\\
\leq & \int\left|C_{\bm{\Sigma}_{l}}\left(\mathbf{x}\right)-C_{\bm{\Sigma}}\left(\mathbf{x}\right)\right|d\mu\left(\mathbf{x}\right)+\int\left|\max_{i\in\mathcal{J}}h_{\bm{\Sigma}_{l},i}\left(\mathbf{x}\right)-\max_{i\in\mathcal{J}}h_{\bm{\Sigma},i}\left(\mathbf{x}\right)\right|d\mu\left(\mathbf{x}\right).
\end{align*}
The first term is bounded by:
\begin{align}
 & \int\left|C_{\bm{\Sigma}_{l}}-C_{\bm{\Sigma}}\right|d\mu\nonumber \\
\leq & \sum_{k=1}^{J}\int\left|\left(\max_{r\in\mathcal{J}}\theta_{\bm{\Sigma}_{l},r}^{k}\right)\pi_{\bm{\Sigma}_{l}}^{k}p_{\bm{\theta}_{\bm{\Sigma}_{l}}^{k},\bm{\Sigma}_{l}}-\left(\max_{r\in\mathcal{J}}\bar{\theta}_{r}^{k}\right)\bar{\pi}^{k}p_{\bar{\bm{\theta}}^{k},\bm{\Sigma}}\right|d\mu\nonumber \\
\leq & \sum_{k=1}^{J}\left|\left(\max_{r\in\mathcal{J}}\theta_{\bm{\Sigma}_{l},r}^{k}\right)\pi_{\bm{\Sigma}_{l}}^{k}-\left(\max_{r\in\mathcal{J}}\bar{\theta}_{r}^{k}\right)\bar{\pi}^{k}\right|\int p_{\bm{\theta}_{\bm{\Sigma}_{l}}^{k},\bm{\Sigma}_{l}}d\mu+\left|\bar{\pi}^{k}\right|\left|\left(\max_{r\in\mathcal{J}}\bar{\theta}_{r}^{k}\right)\right|\left\Vert p_{\bm{\theta}_{\bm{\Sigma}_{l}}^{k},\bm{\Sigma}_{l}}-p_{\bar{\bm{\theta}}^{k},\bm{\Sigma}}\right\Vert _{L^{1}\left(\mathbb{R}^{J}\right)}\nonumber \\
\leq & \sum_{k=1}^{J}\left|\left(\max_{r\in\mathcal{J}}\theta_{\bm{\Sigma}_{l},r}^{k}\right)\pi_{\bm{\Sigma}_{l}}^{k}-\left(\max_{r\in\mathcal{J}}\bar{\theta}_{r}^{k}\right)\bar{\pi}^{k}\right|+B\left\Vert p_{\bm{\theta}_{\bm{\Sigma}_{l}}^{k},\bm{\Sigma}_{l}}-p_{\bar{\bm{\theta}}^{k},\bm{\Sigma}}\right\Vert _{L^{1}\left(\mathbb{R}^{J}\right)},\label{eq:C_difference}
\end{align}
where t(\ref{eq:C_difference}) used $\int p_{\bm{\theta}_{\bm{\Sigma}_{l}}^{k}}d\mu=1$,
$\left|\max_{r\in\mathcal{J}}\bar{\theta}_{r}^{k}\right|\leq B$,
and $\left|\bar{\pi}^{k}\right|\leq1$. The second term is bounded
by:
\begin{align}
\int\left|\max_{i\in\mathcal{J}}h_{\bm{\Sigma}_{l},i}-\max_{i\in\mathcal{J}}h_{\bm{\Sigma},i}\right|d\mu & \leq\sum_{i=1}^{J}\int\left|h_{\bm{\Sigma}_{l},i}\left(\mathbf{x}\right)-h_{\bm{\Sigma},i}\left(\mathbf{x}\right)\right|d\mu\nonumber \\
 & \leq\sum_{i=1}^{J}\sum_{k=1}^{J}\left(\left|\pi_{\bm{\Sigma}_{l}}^{k}\theta_{\bm{\Sigma}_{l},i}^{k}-\bar{\pi}^{k}\bar{\theta}_{i}^{k}\right|\int p_{\bm{\theta}_{\bm{\Sigma}_{l}}^{k}}d\mu+\left|\bar{\pi}^{k}\right|\left|\bar{\theta}_{i}^{k}\right|\left\Vert p_{\bm{\theta}_{\bm{\Sigma}_{l}}^{k},\bm{\Sigma}_{l}}-p_{\bar{\bm{\theta}}^{k},\bm{\Sigma}}\right\Vert _{L^{1}\left(\mathbb{R}^{J}\right)}\right)\nonumber \\
 & \leq\sum_{i=1}^{J}\sum_{k=1}^{J}\left(\left|\pi_{\bm{\Sigma}_{l}}^{k}\theta_{\bm{\Sigma}_{l},i}^{k}-\bar{\pi}^{k}\bar{\theta}_{i}^{k}\right|+B\left\Vert p_{\bm{\theta}_{\bm{\Sigma}_{l}}^{k},\bm{\Sigma}_{l}}-p_{\bar{\bm{\theta}}^{k},\bm{\Sigma}}\right\Vert _{L^{1}\left(\mathbb{R}^{J}\right)}\right).\label{eq:h_difference}
\end{align}
We have $\pi_{\bm{\Sigma}_{l}}^{k}\rightarrow\bar{\pi}^{k}$, $\theta_{\bm{\Sigma}_{l},i}^{k}\rightarrow\bar{\theta}^{k}$,
and 
\begin{align*}
\left\Vert p_{\bm{\theta}_{\bm{\Sigma}_{l}}^{k},\bm{\Sigma}_{l}}-p_{\bar{\bm{\theta}}^{k},\bm{\Sigma}}\right\Vert _{L^{1}\left(\mathbb{R}^{J}\right)} & \leq\left\Vert p_{\bm{\theta}_{\bm{\Sigma}_{l}}^{k},\bm{\Sigma}_{l}}-p_{\bm{\theta}_{\bm{\Sigma}_{l}}^{k},\bm{\Sigma}}\right\Vert _{L^{1}\left(\mathbb{R}^{J}\right)}+\left\Vert p_{\bm{\theta}_{\bm{\Sigma}_{l}}^{k},\bm{\Sigma}}-p_{\bar{\bm{\theta}}^{k},\bm{\Sigma}}\right\Vert _{L^{1}\left(\mathbb{R}^{J}\right)},
\end{align*}
where the first term converges to 0 by Lemma \ref{lem:A3} and the
second term converges to 0 by $L^{1}$ continuity of Gaussian densities
and by dominated convergence. Therefore, $\mathcal{R}\left(\pi_{\bm{\Sigma}_{l}},\bm{\Sigma}_{l}\right)$
is continuous in $\left(\pi_{\bm{\Sigma}_{l}},\bm{\Sigma}_{l}\right)$.

\paragraph{Step 3. Accumulation points of least-favorable priors are least-favorable:}

For any prior $\pi'$ and for all $l$,
\[
\mathcal{R}\left(\pi_{\bm{\Sigma}_{l}},\bm{\Sigma}_{l}\right)\geq\mathcal{R}\left(\pi',\bm{\Sigma}_{l}\right),
\]
by definition of the MVN-MMR decision rule and the Bayes risk. Letting
$l\rightarrow\infty$ and using continuity of $\mathcal{R}\left(\pi_{\bm{\Sigma}_{l}},\bm{\Sigma}_{l}\right)$
in both arguments yields, for any prior $\pi'$:
\[
\mathcal{R}\left(\bar{\pi},\bm{\Sigma}\right)\geq\mathcal{R}\left(\pi',\bm{\Sigma}\right).
\]
Therefore, the accumulation point $\bar{\pi}=\left(\bar{\pi}^{k},\bar{\bm{\theta}}^{k}\right)_{k\in\mathcal{J}}$
is least-favorable for $\bm{\Sigma}$.

\paragraph{Step 4. Convergence of the decision rule along the original sequence:}

Define 
\begin{align}
\bar{h}_{ij}\left(\mathbf{x}\right) & \equiv\sum_{k=1}^{J}\bar{\pi}^{k}p_{\bar{\bm{\theta}}^{k},\bm{\Sigma}}\left(\mathbf{x}\right)\left(\bar{\theta}_{i}^{k}-\bar{\theta}_{j}^{k}\right)\nonumber \\
\bar{\delta}_{i}\left(\mathbf{x}\right) & \equiv1\left(\bar{h}_{ij}\left(\mathbf{x}\right)\geq0\quad\forall j\neq i\text{ with arbitrary tie-breaking}\right).\label{eq:tiebreaking_delta_sigma_2}
\end{align}
By Step 3, $\bar{\bm{\delta}}$ is Bayes with respect to the least-favorable
prior $\bar{\pi}$, implying:
\begin{equation}
V_{\bm{\Sigma}}=\sup_{\bm{\theta}\in\Theta}R_{\bm{\Sigma}}\left(\bm{\theta},\bar{\bm{\delta}}\right),\label{eq:minimax_continuity}
\end{equation}
i.e., the rule $\bar{\bm{\delta}}$ achieves the minimax value $V_{\bm{\Sigma}}$.
By the a.e. uniqueness of the MVN-MMR rule (Theorem \ref{thm:(Characterization-of-the}
and its Corollary), $\bar{\bm{\delta}}=\bm{\delta}_{\bm{\Sigma}}$
a.e.. Furthermore, because our choice of the convergent subsequence
(indexed by $l$) was arbitrary and all the subsequences of the decision
rules $\bm{\delta}_{\bm{\Sigma}_{l}}$ converge to the same decision
rule $\bm{\delta}_{\bm{\Sigma}}$ a.e., we conclude that 
\[
\bm{\delta}_{\bm{\Sigma}_{n}}\left(\mathbf{x}\right)\rightarrow\bm{\delta}_{\bm{\Sigma}}\left(\mathbf{x}\right)\quad\text{a.e.-}\mathbf{x}\;\text{as }n\rightarrow\infty.
\]
This establishes part (i) of the proposition.

\paragraph{Step 5. Convergence of the risk along the original sequence: }

Along the subsequence $l$, 
\begin{align}
\lim_{l\rightarrow\infty}\sup_{\bm{\theta}\in\Theta}R_{\bm{\Sigma}_{l}}\left(\bm{\theta},\bm{\delta}_{\bm{\Sigma}_{l}}\right) & =\lim_{l\rightarrow\infty}V_{\bm{\Sigma}_{l}}\nonumber \\
 & =\lim_{l\rightarrow\infty}\mathcal{R}\left(\pi_{\bm{\Sigma}_{l}},\bm{\Sigma}_{l}\right)\nonumber \\
 & =\mathcal{R}\left(\bar{\pi},\bm{\Sigma}\right)\label{eq:step2_prop32}\\
 & =\sup_{\bm{\theta}\in\Theta}R_{\bm{\Sigma}}\left(\bm{\theta},\bar{\bm{\delta}}\right)\nonumber \\
 & =V_{\bm{\Sigma}},\label{eq:V_Sigma_attained}
\end{align}
where (\ref{eq:step2_prop32}) used the continuity established in
Step 2, and (\ref{eq:V_Sigma_attained}) invoked (\ref{eq:minimax_continuity}).
Because our choice of convergent subsequence was arbitrary and all
the subsequences of the values converge to the same point $V_{\bm{\Sigma}}$,
we get 
\[
\lim_{n\rightarrow\infty}\sup_{\bm{\theta}\in\Theta}R_{\bm{\Sigma}_{n}}\left(\bm{\theta},\bm{\delta}_{\bm{\Sigma}_{n}}\right)=\sup_{\bm{\theta}\in\Theta}R_{\bm{\Sigma}}\left(\bm{\theta},\bm{\delta}_{\bm{\Sigma}}\right)=V_{\bm{\Sigma}}.
\]
This establishes part (ii) of the proposition.

\paragraph{Step 6. Pointwise $L^{1}$-convergence of the decision rule }

Let $g_{n}\left(\mathbf{x}\right)=\sum_{j=1}^{J}\left|\delta_{\bm{\Sigma}_{n},i}\left(\mathbf{x}\right)-\delta_{\bm{\Sigma},i}\left(\mathbf{x}\right)\right|$.
By Step 4, for a.e.-$\mathbf{x}$, $g_{n}\left(\mathbf{x}\right)\rightarrow0$
and $0\leq g_{n}\left(\mathbf{x}\right)\leq2$ for all $\mathbf{x}$.
For any fixed $\bm{\theta}$ and $\bm{\Lambda}\in\mathcal{S}$
\[
\lim_{n\rightarrow\infty}\int g_{n}\left(\mathbf{x}\right)p_{\bm{\theta},\bm{\Lambda}}\left(\mathbf{x}\right)d\mu\left(\mathbf{x}\right)=\int\lim_{n\rightarrow\infty}g_{n}\left(\mathbf{x}\right)p_{\bm{\theta},\bm{\Lambda}}\left(\mathbf{x}\right)d\mu\left(\mathbf{x}\right)=0
\]
by dominated convergence. This establishes $\left\Vert \bm{\delta}_{\bm{\Sigma}}-\bm{\delta}_{\bm{\Sigma}_{n}}\right\Vert _{L^{1}\left(P_{\bm{\theta},\bm{\Lambda}}\right)}\rightarrow0$
as $n\rightarrow\infty$ pointwise-$\bm{\theta}$.

\paragraph{Step 7. Uniform $L^{1}$-convergence of the decision rule in $\bm{\theta}$
under Gaussian laws:}

The family $\left\{ P_{\bm{\theta},\bm{\Lambda}}:\bm{\theta}\in\tilde{\Theta}\right\} $
is uniformly tight for a compact $\tilde{\Theta}$. Therefore, for
any $\epsilon>0$, we can choose an $\eta>0$ large enough so that
$\sup_{\bm{\theta}\in\tilde{\Theta}}P_{\bm{\theta},\bm{\Lambda}}\left(\left\Vert \mathbf{x}\right\Vert >\eta\right)<\frac{\epsilon}{4}$.
Then we have:
\begin{align}
\sup_{\bm{\theta}\in\tilde{\Theta}}\int g_{n}\left(\mathbf{x}\right)dP_{\bm{\theta},\bm{\Lambda}}\left(\mathbf{x}\right) & \leq\sup_{\bm{\theta}\in\tilde{\Theta}}\left(\int_{\left\{ \mathbf{x}:\left\Vert \mathbf{x}\right\Vert \leq\eta\right\} }g_{n}\left(\mathbf{x}\right)p_{\bm{\theta},\bm{\Lambda}}\left(\mathbf{x}\right)d\mu\left(\mathbf{x}\right)+\int_{\left\{ \mathbf{x}:\left\Vert \mathbf{x}\right\Vert >\eta\right\} }g_{n}\left(\mathbf{x}\right)p_{\bm{\theta},\bm{\Lambda}}\left(\mathbf{x}\right)d\mu\left(\mathbf{x}\right)\right)\nonumber \\
 & \leq\int_{\left\{ \mathbf{x}:\left\Vert \mathbf{x}\right\Vert \leq\eta\right\} }\left(2\pi\right)^{-\frac{J}{2}}\left|\bm{\Lambda}\right|^{-\frac{1}{2}}\left|g_{n}\left(\mathbf{x}\right)\right|d\mu\left(\mathbf{x}\right)+2\sup_{\bm{\theta}\in\tilde{\Theta}}P_{\bm{\theta},\bm{\Lambda}}\left(\left\{ \mathbf{x}:\left\Vert \mathbf{x}\right\Vert >\eta\right\} \right)\label{eq:line3_twosplit}\\
 & \leq\frac{\epsilon}{2}+\frac{\epsilon}{2}\quad\text{for }n\text{ large enough,}\label{eq:dct_uniformL1}
\end{align}
where the first term of (\ref{eq:line3_twosplit}) used $\sup_{\bm{\theta}}p_{\bm{\theta},\bm{\Lambda}}\left(\mathbf{x}\right)=\left(2\pi\right)^{-\frac{J}{2}}\left|\bm{\Lambda}\right|^{-\frac{1}{2}}$,
the second term of (\ref{eq:line3_twosplit}) used $\left|g_{n}\left(\mathbf{x}\right)\right|\leq2$,
and (\ref{eq:dct_uniformL1}) invoked dominated convergence on the
finite-measure set $\left\{ \mathbf{x}:\left\Vert \mathbf{x}\right\Vert \leq\eta\right\} $. 

\paragraph{Step 8. Uniform $L^{1}$-convergence of the decision rule over Gaussian
laws:}

Fix $\epsilon>0$ and a compact $\tilde{\Theta}$. By Lemma \ref{lem:A3}
and compactness of $\mathcal{S}$, there exists a finite set $\left\{ \bm{\Lambda}^{\left(1\right)},...,\bm{\Lambda}^{\left(M\right)}\right\} \subset\mathcal{S}$
such that, for every $\bm{\Lambda}\in\mathcal{S}$, we can find an
index $m\left(\bm{\Lambda}\right)$ with 
\begin{equation}
\sup_{\bm{\theta}\in\tilde{\Theta}}\left\Vert p_{\bm{\theta},\bm{\Lambda}}-p_{\bm{\theta},\bm{\Lambda}^{\left(m\right)}}\right\Vert _{L^{1}\left(\mathbb{R}^{J}\right)}<\frac{\epsilon}{4}.\label{eq:epsilonover4}
\end{equation}
For each $m=1,2,...,M$, Step 7 implies there exists $N_{m}$ such
that $\forall n\geq N_{m}$,
\[
\sup_{\bm{\theta}\in\tilde{\Theta}}\int g_{n}p_{\bm{\theta},\bm{\Lambda}^{\left(m\right)}}d\mu<\frac{\epsilon}{2}.
\]
Let $N=\max_{m}N_{m}$. Then, $\forall n\geq N$, any $\bm{\Lambda}\in\mathcal{S}$,
any $\bm{\theta}\in\tilde{\Theta}$, choose $m=m\left(\bm{\Lambda}\right)$
such that (\ref{eq:epsilonover4}) holds, and write 
\begin{align*}
\int g_{n}p_{\bm{\theta},\bm{\Lambda}}d\mu & \leq\int g_{n}p_{\bm{\theta},\bm{\Lambda}^{\left(m\right)}}d\mu+\int g_{n}\left|p_{\bm{\theta},\bm{\Lambda}}-p_{\bm{\theta},\bm{\Lambda}^{\left(m\right)}}\right|d\mu\\
 & \leq\int g_{n}p_{\bm{\theta},\bm{\Lambda}^{\left(m\right)}}d\mu+2\int\left|p_{\bm{\theta},\bm{\Lambda}}-p_{\bm{\theta},\bm{\Lambda}^{\left(m\right)}}\right|d\mu\\
 & \leq\frac{\epsilon}{2}+2\cdot\frac{\epsilon}{4}=\epsilon,
\end{align*}
where we used $g_{n}\leq2$ in the second inequality. This $\epsilon$-bound
does not depend on $\bm{\theta}$ or $\bm{\Lambda}$. Therefore, taking
suprema in $\bm{\theta}$ and in $\bm{\Lambda}$ gives:
\[
\sup_{\bm{\Lambda}\in\mathcal{S}}\sup_{\bm{\theta}\in\tilde{\Theta}}\left\Vert \bm{\delta}_{\bm{\Sigma}_{n}}\left(\mathbf{x}\right)-\bm{\delta}_{\bm{\Sigma}}\left(\mathbf{x}\right)\right\Vert _{L^{1}\left(P_{\bm{\theta},\bm{\Lambda}}\right)}=\sup_{\bm{\Lambda}\in\mathcal{S},\bm{\theta}\in\tilde{\Theta}}\int g_{n}p_{\bm{\theta},\bm{\Lambda}}d\mu<\epsilon.
\]
This establishes part (iii) of the proposition.

\subsection{Proof of Proposition \ref{prop:(Consequences-of-Local}}\label{subsec:Prop2_proof}

Let $n\in\mathbb{N}$. Because $\Theta$ is a compact hyperrectangle,
there exists $H<\infty$ with $\Theta\subset\text{cl}\left(B_{H}\left(0\right)\right)$.
Assumption \ref{assu:For-every-} implies that, for any $\bm{\theta}\in\text{int}\left(\Theta\right)$
and any compact set $K\subset\text{cl}\left(B_{H}\left(0\right)\right)$
(e.g., $K=\Theta$),
\[
\sup_{\mathbf{h}\in K}d_{LP}\left(\mathcal{L}_{\bm{\theta}+\frac{\mathbf{h}}{\sqrt{n}}}\left(\sqrt{n}\left(\hat{\bm{\theta}}_{n}-\left(\bm{\theta}+\frac{\mathbf{h}}{\sqrt{n}}\right)\right)\right),\mathcal{N}\left(\mathbf{0},\bm{\Sigma}\right)\right)\rightarrow0\quad\text{as }n\rightarrow\infty.
\]
Setting $K=\text{cl}\left(B_{H}\left(0\right)\right)$, $\bm{\theta}=\mathbf{0}$,
and $\bm{\vartheta}=\mathbf{h}$ gives (\ref{eq:Q_central}). A location
shift invoking the translation invariance of $d_{LP}$ then gives
(\ref{eq:P_central}).

\subsection{Proof of Theorem \ref{thm:(Uniform-Convergence-of} and Its Corollary}\label{subsec:Thm4_proof}

\paragraph{Step 0. Notations and preliminary facts:}

Fix a decision rule $\bm{\phi}$ that satisfies Assumption \ref{assu:For-each-,}.
For simplicity of notation, denote by $g_{\bm{\vartheta}}\left(\mathbf{y}\right)=L\left(\bm{\phi}\left(\mathbf{y}\right),\bm{\vartheta}\right)$.
$g_{\bm{\vartheta}}$ is a.e. continuous in $\mathbf{y}$ with respect
to $P_{\bm{\vartheta},\bm{\Sigma}}\left(\ll\mu\right)$ by Assumption
\ref{assu:For-each-,}, 2-Lipschitz in $\bm{\vartheta}$, and $\left|g_{\bm{\vartheta}}\right|\leq2B$
by Lemma \ref{lem:A1}.

Under any coupling $\left(\mathbf{y}_{\bm{\vartheta},n},\mathbf{z}_{\bm{\vartheta}}\right)$
of $\left(P_{\bm{\vartheta},n},P_{\bm{\vartheta},\bm{\Sigma}}\right)$,
\begin{align}
\left|R_{\left[n\right]}\left(\bm{\vartheta},\bm{\phi}\right)-R_{\bm{\Sigma}}\left(\bm{\vartheta},\bm{\phi}\right)\right| & =\left|\int g_{\bm{\vartheta}}dP_{\bm{\vartheta},n}-\int g_{\bm{\vartheta}}dP_{\bm{\vartheta},\bm{\Sigma}}\right|\nonumber \\
 & =\left|\mathbb{E}\left[g_{\bm{\vartheta}}\left(\mathbf{y}_{\bm{\vartheta},n}\right)-g_{\bm{\vartheta}}\left(\mathbf{z}_{\bm{\vartheta}}\right)\right]\right|\nonumber \\
 & \leq\mathbb{E}\left[\left|g_{\bm{\vartheta}}\left(\mathbf{y}_{\bm{\vartheta},n}\right)-g_{\bm{\vartheta}}\left(\mathbf{z}_{\bm{\vartheta}}\right)\right|\right],\label{eq:coupling_bound}
\end{align}
and our goal is to control (\ref{eq:coupling_bound}) below an arbitrary
$\epsilon>0$ uniformly across $\Theta$.

\paragraph{Step 1. Uniform bound on coupling discrepancy probability:}

Define the uniform Levy-Prokhorov gap $\epsilon_{n}\equiv\sup_{\bm{\vartheta}\in\Theta}d_{LP}\left(P_{\bm{\vartheta},n},P_{\bm{\vartheta},\bm{\Sigma}}\right)$.
By Proposition \ref{prop:(Consequences-of-Local}, for any $\eta>0$,
there exists $N_{\eta}$ such that $n\geq N_{\eta}$ implies $\epsilon_{n}<\eta$.

Fix $\eta>0$ (to be chosen). Because $P_{\bm{\vartheta},n}\rightsquigarrow P_{\bm{\vartheta},\bm{\Sigma}}$,
Prokhorov's theorem \citep[Theorem 2.4]{vanderVaart1998} asserts
that $\left\{ P_{\bm{\vartheta},n}\right\} $ is uniformly tight.
By Strassen's coupling theorem (Theorem \ref{thm:(Strassen-Prokhorov/Strassen-Dud}
in Appendix \ref{subsec:Theorems-Used-in}) applied pointwise in $\bm{\vartheta}$,
for each $\bm{\vartheta}$ there exist random variables $\mathbf{y}_{\bm{\vartheta},n}\sim P_{\bm{\vartheta},n}$
and $\mathbf{z}_{\bm{\vartheta}}\sim P_{\bm{\vartheta}}$ (possibly
on $\bm{\vartheta}$-dependent probability spaces) such that 
\[
\Pr\left(\left\Vert \mathbf{y}_{\bm{\vartheta},n}-\mathbf{z}_{\bm{\vartheta}}\right\Vert >\eta\right)\leq\Pr\left(\left\Vert \mathbf{y}_{\bm{\vartheta},n}-\mathbf{z}_{\bm{\vartheta}}\right\Vert >\epsilon_{n}\right)\leq\epsilon_{n},\quad\forall n\geq N_{\eta},
\]
where the first inequality used the tail event $\left\{ \left\Vert \mathbf{y}_{\bm{\vartheta},n}-\mathbf{z}_{\bm{\vartheta}}\right\Vert >\eta\right\} \subset\left\{ \left\Vert \mathbf{y}_{\bm{\vartheta},n}-\mathbf{z}_{\bm{\vartheta}}\right\Vert >\epsilon_{n}\right\} $
for $\epsilon_{n}<\eta$. Because $\eta$ and $N_{\eta}$ do not depend
on the choice of $\bm{\vartheta}$, the bound is uniform:
\begin{equation}
\sup_{\bm{\vartheta}\in\Theta}\Pr\left(\left\Vert \mathbf{y}_{\bm{\vartheta},n}-\mathbf{z}_{\bm{\vartheta}}\right\Vert >\eta\right)\leq\epsilon_{n},\quad\forall n\geq N_{\eta}.\label{eq:coupling_uniform}
\end{equation}
We note that each $\bm{\vartheta}$ may use possibly different coupling,
which does not pose a problem for our purpose as we only want to uniformly
bound the probability of discrepancy across $\bm{\vartheta}\in\Theta$.

\paragraph{Step 2. Oscillation bound of the loss:}

For the bounded regret loss function $g_{\bm{\vartheta}}$ and the
$\eta>0$, define the local oscillation: 
\[
\underset{\eta}{\text{osc}}\,g_{\bm{\vartheta}}\left(\mathbf{y}\right)\equiv\sup_{\left\Vert \mathbf{h}\right\Vert \leq\eta}\left|g_{\bm{\vartheta}}\left(\mathbf{y}+\mathbf{h}\right)-g_{\bm{\vartheta}}\left(\mathbf{y}\right)\right|.
\]

Let $A\equiv\left\{ \left\Vert \mathbf{y}_{\bm{\vartheta},n}-\mathbf{z}_{\bm{\vartheta}}\right\Vert \leq\eta\right\} $.
(\ref{eq:coupling_bound}) is bounded above by: 
\begin{align}
\mathbb{E}\left[\left|g_{\bm{\vartheta}}\left(\mathbf{y}_{\bm{\vartheta},n}\right)-g_{\bm{\vartheta}}\left(\mathbf{z}_{\bm{\vartheta}}\right)\right|\right] & =\mathbb{E}\left[\left|g_{\bm{\vartheta}}\left(\mathbf{y}_{\bm{\vartheta},n}\right)-g_{\bm{\vartheta}}\left(\mathbf{z}_{\bm{\vartheta}}\right)\right|\cdot\mathbf{1}_{A}\right]+\mathbb{E}\left[\left|g_{\bm{\vartheta}}\left(\mathbf{y}_{\bm{\vartheta},n}\right)-g_{\bm{\vartheta}}\left(\mathbf{z}_{\bm{\vartheta}}\right)\right|\cdot\mathbf{1}_{A^{c}}\right]\nonumber \\
 & \leq\mathbb{E}\left[\underset{\eta}{\text{osc}}\,g_{\bm{\vartheta}}\left(\mathbf{z}_{\bm{\vartheta}}\right)\right]+2\left\Vert g_{\bm{\vartheta}}\right\Vert _{\infty}\Pr\left(\left\Vert \mathbf{y}_{\bm{\vartheta},n}-\mathbf{z}_{\bm{\vartheta}}\right\Vert >\eta\right)\nonumber \\
 & \leq\mathbb{E}\left[\underset{\eta}{\text{osc}}\,g_{\bm{\vartheta}}\left(\mathbf{z}_{\bm{\vartheta}}\right)\right]+4B\epsilon_{n}\quad\forall n\geq N_{\eta},\label{eq:intermediate_close_to_final}
\end{align}
where the last inequality used $\left\Vert g_{\bm{\vartheta}}\right\Vert _{\infty}\leq2B$
and (\ref{eq:coupling_uniform}).

In the following Steps 3 and 4, we bound the first term of (\ref{eq:intermediate_close_to_final})
uniformly over $\bm{\vartheta}$.

\paragraph{Step 3. A uniform modulus of continuity:}

Let $\delta>0$ (to be chosen). Because $\Theta=\left[-B,B\right]^{J}$
is compact, there exists a finite $\delta$-net $\left\{ \bm{\vartheta}^{\left(m\right)}\right\} _{m=1}^{M}\subset\Theta$
such that for each $\bm{\vartheta}\in\Theta$ there is an index $m=m\left(\bm{\vartheta}\right)$
with $\left\Vert \bm{\vartheta}-\bm{\vartheta}^{\left(m\right)}\right\Vert \leq\delta$.
Then, 
\begin{align}
 & \underset{\eta}{\text{osc}}\,g_{\bm{\vartheta}}\left(\mathbf{z}_{\bm{\vartheta}}\right)\nonumber \\
= & \sup_{\left\Vert \mathbf{h}\right\Vert \leq\eta}\left|g_{\bm{\vartheta}}\left(\mathbf{z}_{\bm{\vartheta}}+\mathbf{h}\right)-g_{\bm{\vartheta}}\left(\mathbf{z}_{\bm{\vartheta}}\right)\right|\nonumber \\
\leq & \sup_{\left\Vert \mathbf{h}\right\Vert \leq\eta}\left(\left|g_{\bm{\vartheta}^{\left(m\right)}}\left(\mathbf{z}_{\bm{\vartheta}}+\mathbf{h}\right)-g_{\bm{\vartheta}^{\left(m\right)}}\left(\mathbf{z}_{\bm{\vartheta}}\right)\right|+\left|g_{\bm{\vartheta}}\left(\mathbf{z}_{\bm{\vartheta}}+\mathbf{h}\right)-g_{\bm{\vartheta}^{\left(m\right)}}\left(\mathbf{z}_{\bm{\vartheta}}+\mathbf{h}\right)\right|+\left|g_{\bm{\vartheta}^{\left(m\right)}}\left(\mathbf{z}_{\bm{\vartheta}}\right)-g_{\bm{\vartheta}}\left(\mathbf{z}_{\bm{\vartheta}}\right)\right|\right)\nonumber \\
\leq & \underset{\eta}{\text{osc}}\,g_{\bm{\vartheta}^{\left(m\right)}}\left(\mathbf{z}_{\bm{\vartheta}}\right)+2\cdot2\left\Vert \bm{\vartheta}-\bm{\vartheta}^{\left(m\right)}\right\Vert _{2}\leq\underset{\eta}{\text{osc}}\,g_{\bm{\vartheta}^{\left(m\right)}}\left(\mathbf{z}_{\bm{\vartheta}}\right)+4\delta,\label{eq:uniform_modulus}
\end{align}
where the last inequality invoked Lemma \ref{lem:A1}.

\paragraph{Step 4. Uniform bound of $\mathbb{E}_{P_{\bm{\vartheta},\bm{\Sigma}}}\left[\protect\underset{\eta}{\text{osc}}\,g_{\bm{\vartheta}^{\left(m\right)}}\left(\mathbf{z}_{\bm{\vartheta}}\right)\right]$
over $\bm{\vartheta}$ using Dini's theorem:}

Taking expectations under $P_{\bm{\vartheta},\bm{\Sigma}}$ and suprema
over $\bm{\vartheta}$ gives:
\begin{equation}
\sup_{\bm{\vartheta}\in\Theta}\mathbb{E}_{P_{\bm{\vartheta},\bm{\Sigma}}}\left[\underset{\eta}{\text{osc}}\,g_{\bm{\vartheta}}\left(\mathbf{z}_{\bm{\vartheta}}\right)\right]\leq\max_{1\leq m\leq M}\sup_{\bm{\vartheta}\in\Theta}\mathbb{E}_{P_{\bm{\vartheta},\bm{\Sigma}}}\left[\underset{\eta}{\text{osc}}\,g_{\bm{\vartheta}^{\left(m\right)}}\left(\mathbf{z}_{\bm{\vartheta}}\right)\right]+4\delta.\label{eq:step3_conc}
\end{equation}

For each fixed $m$ and $\bm{\vartheta}^{\left(m\right)}$, define
\[
F_{\eta}^{\left(m\right)}\left(\bm{\vartheta}\right)\equiv\mathbb{E}_{P_{\bm{\vartheta},\bm{\Sigma}}}\left[\underset{\eta}{\text{osc}}\,g_{\bm{\vartheta}^{\left(m\right)}}\left(\mathbf{z}_{\bm{\vartheta}}\right)\right]=\int\underset{\eta}{\text{osc}}\,g_{\bm{\vartheta}^{\left(m\right)}}\left(\mathbf{y}\right)dP_{\bm{\vartheta},\bm{\Sigma}}\left(\mathbf{y}\right).
\]
We claim that $F_{\eta}^{\left(m\right)}\left(\bm{\vartheta}\right)$
is (i) monotonically decreasing to 0 for fixed $\bm{\vartheta}$ as
$\eta\downarrow0$, and (ii) continuous in $\bm{\vartheta}$ for $\eta>0$
fixed.

\uline{(1) (Monotonicity in \mbox{$\eta$})}: Fix $\bm{\vartheta}$.
For any fixed $\bm{\vartheta}^{\left(m\right)}$, and a.e.-$\mathbf{y}$,
$\eta<\eta'$ implies $\underset{\eta}{\text{osc}}\,g_{\bm{\vartheta}^{\left(m\right)}}\left(\mathbf{y}\right)\leq\underset{\eta'}{\text{osc}}\,g_{\bm{\vartheta}^{\left(m\right)}}\left(\mathbf{y}\right)$
because the domain of suprema in the oscillation $\left\{ \mathbf{h}:\left\Vert \mathbf{h}\right\Vert \leq\eta\right\} \subseteq\left\{ \mathbf{h}:\left\Vert \mathbf{h}\right\Vert \leq\eta'\right\} $,
and $\lim_{\eta\downarrow0}\underset{\eta}{\text{osc}}\,g_{\bm{\vartheta}^{\left(m\right)}}\left(\mathbf{y}\right)=0$.
Furthermore, $\left|\underset{\eta}{\text{osc}}\,g_{\bm{\vartheta}^{\left(m\right)}}\left(\mathbf{y}\right)\right|\leq4B$.
Therefore, for fixed $\bm{\vartheta}$ and for $\eta<\eta'$,
\[
F_{\eta}^{\left(m\right)}\left(\bm{\vartheta}\right)=\int\underset{\eta}{\text{osc}}\,g_{\bm{\vartheta}^{\left(m\right)}}\left(\mathbf{y}\right)dP_{\bm{\vartheta},\bm{\Sigma}}\left(\mathbf{y}\right)\leq\int\underset{\eta'}{\text{osc}}\,g_{\bm{\vartheta}^{\left(m\right)}}\left(\mathbf{y}\right)dP_{\bm{\vartheta},\bm{\Sigma}}\left(\mathbf{y}\right)=F_{\eta'}^{\left(m\right)}\left(\bm{\vartheta}\right),
\]
and $F_{\eta}^{\left(m\right)}\left(\bm{\vartheta}\right)\downarrow0$
as $\eta\downarrow0$ by dominated convergence.

\uline{(2) (Continuity in \mbox{$\bm{\vartheta}$})}: For $\eta>0$
fixed, 
\begin{align*}
\left|F_{\eta}^{\left(m\right)}\left(\bm{\vartheta}\right)-F_{\eta}^{\left(m\right)}\left(\bm{\vartheta}'\right)\right| & =\left|\int\underset{\eta}{\text{osc}}\,g_{\bm{\vartheta}^{\left(m\right)}}\left(\mathbf{y}\right)dP_{\bm{\vartheta},\bm{\Sigma}}\left(\mathbf{y}\right)-\int\underset{\eta}{\text{osc}}\,g_{\bm{\vartheta}^{\left(m\right)}}\left(\mathbf{y}\right)dP_{\bm{\vartheta}',\bm{\Sigma}}\left(\mathbf{y}\right)\right|\\
 & \leq4B\left\Vert P_{\bm{\vartheta}}-P_{\bm{\vartheta}'}\right\Vert _{TV}.
\end{align*}
Because $P_{\bm{\vartheta},\bm{\Sigma}}=\mathcal{N}\left(\bm{\vartheta},\bm{\Sigma}\right)$
is continuous in total variation in $\bm{\vartheta}$ on compact $\Theta$,
$\bm{\vartheta}\mapsto F_{\eta}^{\left(m\right)}\left(\bm{\vartheta}\right)$
is continuous on $\Theta$.

\uline{(3) (Applying Dini's theorem)}: (1) and (2) are sufficient
for Dini's theorem to apply, which yields the uniform convergence
over $\bm{\vartheta}$:
\begin{equation}
\lim_{\eta\downarrow0}\sup_{\bm{\vartheta}\in\Theta}F_{\eta}^{\left(m\right)}\left(\bm{\vartheta}\right)=\lim_{\eta\downarrow0}\sup_{\bm{\vartheta}\in\Theta}\mathbb{E}_{P_{\bm{\vartheta},\bm{\Sigma}}}\left[\underset{\eta}{\text{osc}}\,g_{\bm{\vartheta}^{\left(m\right)}}\left(\mathbf{z}_{\bm{\vartheta}}\right)\right]=0,\quad\text{for each fixed }m.\label{eq:uniform_converg_expectation_osc}
\end{equation}

\uline{(4) (Maximum Over Finitely Many \mbox{$m$} to Obtain the Uniform
Bound)}: With finitely many $m$, taking $\max_{1\leq m\leq M}$ preserves
the uniform convergence to 0. Combining (\ref{eq:step3_conc}) and
(\ref{eq:uniform_converg_expectation_osc}) gives, for any $\epsilon>0$,
we can first pick $\delta>0$ so that $4\delta\leq\frac{\epsilon}{4}$,
and then pick $\eta>0$ small enough so that:
\[
\max_{1\leq m\leq M}\sup_{\bm{\vartheta}\in\Theta}\mathbb{E}_{P_{\bm{\vartheta},\bm{\Sigma}}}\left[\underset{\eta}{\text{osc}}\,g_{\bm{\vartheta}^{\left(m\right)}}\left(\mathbf{z}_{\bm{\vartheta}}\right)\right]\leq\frac{\epsilon}{4}.
\]
Therefore, 
\begin{equation}
\sup_{\bm{\vartheta}\in\Theta}\mathbb{E}_{P_{\bm{\vartheta},\bm{\Sigma}}}\left[\underset{\eta}{\text{osc}}\,g_{\bm{\vartheta}}\left(\mathbf{z}_{\bm{\vartheta}}\right)\right]\leq\max_{1\leq m\leq M}\sup_{\bm{\vartheta}\in\Theta}\mathbb{E}_{P_{\bm{\vartheta},\bm{\Sigma}}}\left[\underset{\eta}{\text{osc}}\,g_{\bm{\vartheta}^{\left(m\right)}}\left(\mathbf{z}_{\bm{\vartheta}}\right)\right]+4\delta\leq\frac{\epsilon}{4}+\frac{\epsilon}{4}=\frac{\epsilon}{2},\label{eq:step4_ineq}
\end{equation}
bounding the first term of (\ref{eq:intermediate_close_to_final})
under $\frac{\epsilon}{2}$ uniformly over $\bm{\vartheta}$.

\paragraph{Step 5. Uniform bound of the regret difference:}

Recall (\ref{eq:coupling_bound}) and (\ref{eq:intermediate_close_to_final})
with suprema over $\Theta$ taken:
\begin{align*}
\sup_{\bm{\vartheta}\in\Theta}\left|R_{\left[n\right]}\left(\bm{\vartheta},\bm{\phi}\right)-R_{\bm{\Sigma}}\left(\bm{\vartheta},\bm{\phi}\right)\right| & \leq\sup_{\bm{\vartheta}\in\Theta}\mathbb{E}\left[\left|g_{\bm{\vartheta}}\left(\mathbf{y}_{\bm{\vartheta},n}\right)-g_{\bm{\vartheta}}\left(\mathbf{z}_{\bm{\vartheta}}\right)\right|\right]\\
 & \leq\sup_{\bm{\vartheta}\in\Theta}\mathbb{E}_{P_{\bm{\vartheta},\bm{\Sigma}}}\left[\underset{\eta}{\text{osc}}\,g_{\bm{\vartheta}}\left(\mathbf{z}_{\bm{\vartheta}}\right)\right]+4B\epsilon_{n}\quad\forall n\geq N_{\eta}.
\end{align*}

Let $\epsilon>0$, we can pick $\delta>0$ small enough in the inequality
(\ref{eq:step3_conc}), and then pick $\eta>0$ small enough in the
inequality (\ref{eq:step4_ineq}) so that $\mathbb{E}_{P_{\bm{\vartheta},\bm{\Sigma}}}\left[\underset{\eta}{\text{osc}}\,g_{\bm{\vartheta}}\left(\mathbf{z}_{\bm{\vartheta}}\right)\right]<\frac{\epsilon}{2}$.
Then pick $n\geq N_{\eta}$ large enough so $4B\epsilon_{n}<\frac{\epsilon}{2}$.
This completes the proof of the main theorem.

\paragraph{Proof of the Corollary:}

We have 
\begin{align}
\left|V_{\bm{\Sigma}}-\sup_{\bm{\vartheta}\in\Theta}R_{\left[n\right]}\left(\bm{\vartheta},\bm{\delta}_{\bm{\Sigma}}\right)\right| & =\left|\sup_{\bm{\vartheta}\in\Theta}R_{\bm{\Sigma}}\left(\bm{\vartheta},\bm{\delta}_{\bm{\Sigma}}\right)-\sup_{\bm{\vartheta}\in\Theta}R_{\left[n\right]}\left(\bm{\vartheta},\bm{\delta}_{\bm{\Sigma}}\right)\right|\nonumber \\
 & \leq\sup_{\bm{\vartheta}\in\Theta}\left|R_{\bm{\Sigma}}\left(\bm{\vartheta},\bm{\delta}_{\bm{\Sigma}}\right)-R_{\left[n\right]}\left(\bm{\vartheta},\bm{\delta}_{\bm{\Sigma}}\right)\right|\label{eq:diff_max_max_diff}\\
 & \rightarrow0\quad as\quad n\rightarrow\infty,\label{eq:thm4_application}
\end{align}
where we used $\left|\sup_{x}f\left(x\right)-\sup_{y}g\left(y\right)\right|\leq\sup_{x}\left|f\left(x\right)-g\left(x\right)\right|$
in (\ref{eq:diff_max_max_diff}), and (\ref{eq:thm4_application})
follows by applying the theorem with $\bm{\phi}=\bm{\delta}_{\bm{\Sigma}}$,
as $\bm{\delta}_{\bm{\Sigma}}$ satisfies Assumption \ref{assu:For-each-,}.

\subsection{Proof of Lemma \ref{lem:41}}\label{subsec:Proof-of-Lemma41}

We closely follow the proof steps of Theorem \ref{thm:(Uniform-Convergence-of}.
Denote by $g_{\bm{\vartheta}}^{\bm{\Gamma}}\left(\mathbf{y}\right)=L\left(\bm{\delta}_{\bm{\Gamma}}\left(\mathbf{y}\right),\bm{\vartheta}\right)$.
Steps 0-3 of the proof are identical with those of the proof of Theorem
\ref{thm:(Uniform-Convergence-of}, which does not concern uniformity
over $\bm{\Gamma}$. 

\paragraph{Step 4'. Uniform bound of $\mathbb{E}_{P_{\bm{\vartheta},\bm{\Sigma}}}\left[\protect\underset{\eta}{\text{osc}}\,g_{\bm{\vartheta}^{\left(m\right)}}^{\bm{\Gamma}}\left(\mathbf{z}_{\bm{\vartheta}}\right)\right]$
over $\left(\bm{\vartheta},\bm{\Gamma}\right)$ using Dini's theorem:}

For the $\delta$-net $\left\{ \bm{\vartheta}^{\left(m\right)}\right\} $,
taking expectations under $P_{\bm{\vartheta},\bm{\Sigma}}$ and suprema
over $\left(\bm{\vartheta},\bm{\Gamma}\right)$ gives:
\[
\sup_{\bm{\vartheta},\bm{\Gamma}}\mathbb{E}_{P_{\bm{\vartheta},\bm{\Sigma}}}\left[\underset{\eta}{\text{osc}}\,g_{\bm{\vartheta}}^{\bm{\Gamma}}\left(\mathbf{z}_{\bm{\vartheta}}\right)\right]\leq\max_{1\leq m\leq M}\sup_{\bm{\vartheta},\bm{\Gamma}}\mathbb{E}_{P_{\bm{\vartheta},\bm{\Sigma}}}\left[\underset{\eta}{\text{osc}}\,g_{\bm{\vartheta}^{\left(m\right)}}^{\bm{\Gamma}}\left(\mathbf{z}_{\bm{\vartheta}}\right)\right]+4\delta.
\]
For each fixed $m$ and $\bm{\vartheta}^{\left(m\right)}$, define
\[
F_{\eta}^{\left(m\right)}\left(\bm{\vartheta},\bm{\Gamma}\right)\equiv\mathbb{E}_{P_{\bm{\vartheta},\bm{\Sigma}}}\left[\underset{\eta}{\text{osc}}\,g_{\bm{\vartheta}^{\left(m\right)}}^{\bm{\Gamma}}\left(\mathbf{z}_{\bm{\vartheta}}\right)\right]=\int\underset{\eta}{\text{osc}}\,g_{\bm{\vartheta}^{\left(m\right)}}^{\bm{\Gamma}}\left(\mathbf{y}\right)dP_{\bm{\vartheta},\bm{\Sigma}}\left(\mathbf{y}\right).
\]

\uline{(1) (Monotonicity in \mbox{$\eta$})}: For any fixed $\bm{\vartheta}^{\left(m\right)}$,
$\bm{\Gamma}$, and a.e-$\mathbf{y}$, $\eta<\eta'$ implies $\underset{\eta}{\text{osc}}\,g_{\bm{\vartheta}^{\left(m\right)}}^{\bm{\Gamma}}\left(\mathbf{y}\right)\leq\underset{\eta'}{\text{osc}}\,g_{\bm{\vartheta}^{\left(m\right)}}^{\bm{\Gamma}}\left(\mathbf{y}\right)$
because $\left\{ \mathbf{h}:\left\Vert \mathbf{h}\right\Vert \leq\eta\right\} \subseteq\left\{ \mathbf{h}:\left\Vert \mathbf{h}\right\Vert \leq\eta'\right\} $,
and $\lim_{\eta\downarrow0}\underset{\eta}{\text{osc}}\,g_{\bm{\vartheta}^{\left(m\right)}}^{\bm{\Gamma}}\left(\mathbf{y}\right)=0$.
Furthermore, $\left|\underset{\eta}{\text{osc}}\,g_{\bm{\vartheta}^{\left(m\right)}}^{\bm{\Gamma}}\left(\mathbf{y}\right)\right|\leq4B$.
Therefore, for fixed $\bm{\vartheta}$ and for $\eta<\eta'$,
\[
F_{\eta}^{\left(m\right)}\left(\bm{\vartheta},\bm{\Gamma}\right)=\int\underset{\eta}{\text{osc}}\,g_{\bm{\vartheta}^{\left(m\right)}}^{\bm{\Gamma}}\left(\mathbf{y}\right)dP_{\bm{\vartheta},\bm{\Sigma}}\leq\int\underset{\eta'}{\text{osc}}\,g_{\bm{\vartheta}^{\left(m\right)}}^{\bm{\Gamma}}\left(\mathbf{y}\right)dP_{\bm{\vartheta},\bm{\Sigma}}=F_{\eta'}^{\left(m\right)}\left(\bm{\vartheta},\bm{\Gamma}\right),
\]
and $F_{\eta}^{\left(m\right)}\left(\bm{\vartheta},\bm{\Gamma}\right)\downarrow0$
as $\eta\downarrow0$ by dominated convergence.

\uline{(2) (Continuity in \mbox{$\left(\bm{\vartheta},\bm{\Gamma}\right)$})}:
For each $\eta>0$ fixed, we show $\left(\bm{\vartheta},\bm{\Gamma}\right)\mapsto F_{\eta}^{\left(m\right)}\left(\bm{\vartheta},\bm{\Gamma}\right)$
is continuous on the compact product $\Theta\times\mathcal{S}$. For
$\left(\bm{\vartheta}',\bm{\Gamma}'\right)\neq\left(\bm{\vartheta},\bm{\Gamma}\right)$:
\begin{eqnarray}
\left|F_{\eta}^{\left(m\right)}\left(\bm{\vartheta}',\bm{\Gamma}'\right)-F_{\eta}^{\left(m\right)}\left(\bm{\vartheta},\bm{\Gamma}\right)\right| & \leq & \left|F_{\eta}^{\left(m\right)}\left(\bm{\vartheta}',\bm{\Gamma}'\right)-F_{\eta}^{\left(m\right)}\left(\bm{\vartheta},\bm{\Gamma}'\right)\right|+\left|F_{\eta}^{\left(m\right)}\left(\bm{\vartheta},\bm{\Gamma}'\right)-F_{\eta}^{\left(m\right)}\left(\bm{\vartheta},\bm{\Gamma}\right)\right|\nonumber \\
 & = & \left|\int\underset{\eta}{\text{osc}}\,g_{\bm{\vartheta}^{\left(m\right)}}^{\bm{\Gamma}'}\left(\mathbf{y}\right)dP_{\bm{\vartheta}',\bm{\Sigma}}\left(\mathbf{y}\right)-\int\underset{\eta}{\text{osc}}\,g_{\bm{\vartheta}^{\left(m\right)}}^{\bm{\Gamma}'}\left(\mathbf{y}\right)dP_{\bm{\vartheta},\bm{\Sigma}}\left(\mathbf{y}\right)\right|\nonumber \\
 &  & +\left|\int\underset{\eta}{\text{osc}}\,g_{\bm{\vartheta}^{\left(m\right)}}^{\bm{\Gamma}'}\left(\mathbf{y}\right)dP_{\bm{\vartheta},\bm{\Sigma}}\left(\mathbf{y}\right)-\int\underset{\eta}{\text{osc}}\,g_{\bm{\vartheta}^{\left(m\right)}}^{\bm{\Gamma}}\left(\mathbf{y}\right)dP_{\bm{\vartheta},\bm{\Sigma}}\left(\mathbf{y}\right)\right|\nonumber \\
 & \leq & 4B\left\Vert P_{\bm{\vartheta},\bm{\Sigma}}-P_{\bm{\vartheta}',\bm{\Sigma}}\right\Vert _{TV}\nonumber \\
 &  & +\left|\int\underset{\eta}{\text{osc}}\,g_{\bm{\vartheta}^{\left(m\right)}}^{\bm{\Gamma}}\left(\mathbf{y}\right)dP_{\bm{\vartheta},\bm{\Sigma}}\left(\mathbf{y}\right)-\int\underset{\eta}{\text{osc}}\,g_{\bm{\vartheta}^{\left(m\right)}}^{\bm{\Gamma}'}\left(\mathbf{y}\right)dP_{\bm{\vartheta},\bm{\Sigma}}\left(\mathbf{y}\right)\right|.\label{eq:RHS_F_eta_m}
\end{eqnarray}
$4B\left\Vert P_{\bm{\vartheta}}-P_{\bm{\vartheta}'}\right\Vert _{TV}\rightarrow0$
as $\bm{\vartheta}'\rightarrow\bm{\vartheta}$ for Gaussian measures
$P_{\bm{\vartheta},\bm{\Sigma}},P_{\bm{\vartheta}',\bm{\Sigma}}$.

Using $\left|\sup_{u}f\left(u\right)-\sup_{u}g\left(u\right)\right|\leq\sup_{u}\left|f\left(u\right)-g\left(u\right)\right|$
and $\left|\vartheta_{j}^{\left(m\right)}\right|\leq B$, the oscillation-difference
term of (\ref{eq:RHS_F_eta_m}) is bounded further by:
\begin{eqnarray}
 &  & \int\left|\underset{\eta}{\text{osc}}\,g_{\bm{\vartheta}^{\left(m\right)}}^{\bm{\Gamma}}\left(\mathbf{y}\right)-\underset{\eta}{\text{osc}}\,g_{\bm{\vartheta}^{\left(m\right)}}^{\bm{\Gamma}'}\left(\mathbf{y}\right)\right|dP_{\bm{\vartheta},\bm{\Sigma}}\left(\mathbf{y}\right)\nonumber \\
 & = & \int\left|\sup_{\left\Vert \mathbf{h}\right\Vert \leq\eta}\left|g_{\bm{\vartheta}^{\left(m\right)}}^{\bm{\Gamma}}\left(\mathbf{y}+\mathbf{h}\right)-g_{\bm{\vartheta}^{\left(m\right)}}^{\bm{\Gamma}}\left(\mathbf{y}\right)\right|-\sup_{\left\Vert \mathbf{h}\right\Vert \leq\eta}\left|g_{\bm{\vartheta}^{\left(m\right)}}^{\bm{\Gamma}'}\left(\mathbf{y}+\mathbf{h}\right)-g_{\bm{\vartheta}^{\left(m\right)}}^{\bm{\Gamma}'}\left(\mathbf{y}\right)\right|\right|dP_{\bm{\vartheta},\bm{\Sigma}}\left(\mathbf{y}\right)\nonumber \\
 & \leq & \int\sup_{\left\Vert \mathbf{h}\right\Vert \leq\eta}\left|\left(g_{\bm{\vartheta}^{\left(m\right)}}^{\bm{\Gamma}}\left(\mathbf{y}+\mathbf{h}\right)-g_{\bm{\vartheta}^{\left(m\right)}}^{\bm{\Gamma}'}\left(\mathbf{y}+\mathbf{h}\right)\right)-\left(g_{\bm{\vartheta}^{\left(m\right)}}^{\bm{\Gamma}}\left(\mathbf{y}\right)-g_{\bm{\vartheta}^{\left(m\right)}}^{\bm{\Gamma}'}\left(\mathbf{y}\right)\right)\right|dP_{\bm{\vartheta},\bm{\Sigma}}\left(\mathbf{y}\right)\nonumber \\
 & \leq & \int\sup_{\left\Vert \mathbf{h}\right\Vert \leq\eta}\left|\sum_{k=1}^{J}\vartheta_{k}^{\left(m\right)}\left(\delta_{\bm{\Gamma},k}\left(\mathbf{y}+\mathbf{h}\right)-\delta_{\bm{\Gamma}',k}\left(\mathbf{y}+\mathbf{h}\right)\right)+\sum_{k=1}^{J}\vartheta_{k}^{\left(m\right)}\left(\delta_{\bm{\Gamma}',k}\left(\mathbf{y}\right)-\delta_{\bm{\Gamma},k}\left(\mathbf{y}\right)\right)\right|dP_{\bm{\vartheta},\bm{\Sigma}}\left(\mathbf{y}\right)\nonumber \\
 & \leq & \int\sup_{\left\Vert \mathbf{h}\right\Vert \leq\eta}\left(\sum_{k=1}^{J}\left|\vartheta_{k}^{\left(m\right)}\right|\left|\delta_{\bm{\Gamma},k}\left(\mathbf{y}+\mathbf{h}\right)-\delta_{\bm{\Gamma}',k}\left(\mathbf{y}+\mathbf{h}\right)\right|+\sum_{k=1}^{J}\left|\vartheta_{k}^{\left(m\right)}\right|\left|\delta_{\bm{\Gamma}',k}\left(\mathbf{y}\right)-\delta_{\bm{\Gamma},k}\left(\mathbf{y}\right)\right|\right)dP_{\bm{\vartheta},\bm{\Sigma}}\left(\mathbf{y}\right)\nonumber \\
 & \leq & B\int\sup_{\left\Vert \mathbf{h}\right\Vert \leq\eta}\sum_{k=1}^{J}\left(\left|\delta_{\bm{\Gamma},k}\left(\mathbf{y}+\mathbf{h}\right)-\delta_{\bm{\Gamma}',k}\left(\mathbf{y}+\mathbf{h}\right)\right|+\left|\delta_{\bm{\Gamma}',k}\left(\mathbf{y}\right)-\delta_{\bm{\Gamma},k}\left(\mathbf{y}\right)\right|\right)dP_{\bm{\vartheta},\bm{\Sigma}}\left(\mathbf{y}\right)\nonumber \\
 & \leq & B\int\sup_{\left\Vert \mathbf{h}\right\Vert \leq\eta}\sum_{k=1}^{J}\left|\delta_{\bm{\Gamma},k}\left(\mathbf{y}+\mathbf{h}\right)-\delta_{\bm{\Gamma}',k}\left(\mathbf{y}+\mathbf{h}\right)\right|dP_{\bm{\vartheta},\bm{\Sigma}}\left(\mathbf{y}\right)+B\int\sum_{k=1}^{J}\left|\delta_{\bm{\Gamma}',k}\left(\mathbf{y}\right)-\delta_{\bm{\Gamma},k}\left(\mathbf{y}\right)\right|dP_{\bm{\vartheta},\bm{\Sigma}}\left(\mathbf{y}\right)\nonumber \\
 & \leq & B\int\sup_{\left\Vert \mathbf{h}\right\Vert \leq\eta}\sum_{k=1}^{J}\left|\delta_{\bm{\Gamma},k}\left(\mathbf{z}\right)-\delta_{\bm{\Gamma}',k}\left(\mathbf{z}\right)\right|dP_{\bm{\vartheta}+\mathbf{h},\bm{\Sigma}}\left(\mathbf{z}\right)+B\int\sum_{k=1}^{J}\left|\delta_{\bm{\Gamma}',k}\left(\mathbf{y}\right)-\delta_{\bm{\Gamma},k}\left(\mathbf{y}\right)\right|dP_{\bm{\vartheta},\bm{\Sigma}}\left(\mathbf{y}\right)\nonumber \\
 & \leq & B\int\sup_{\left\Vert \mathbf{h}\right\Vert \leq\eta}\left\Vert \bm{\delta}_{\bm{\Gamma}}\left(\mathbf{z}\right)-\bm{\delta}_{\bm{\Gamma}'}\left(\mathbf{z}\right)\right\Vert _{1}p_{\bm{\vartheta}+\mathbf{h},\bm{\Sigma}}\left(\mathbf{z}\right)d\mu\left(\mathbf{z}\right)+B\int\left\Vert \bm{\delta}_{\bm{\Gamma}}-\bm{\delta}_{\bm{\Gamma}'}\right\Vert _{1}dP_{\bm{\vartheta},\bm{\Sigma}}.\label{eq:finalresult_rightbefore}
\end{eqnarray}

We further bound the first term of (\ref{eq:finalresult_rightbefore})
below. Consider the equality
\begin{equation}
p_{\bm{\vartheta}+\mathbf{h},\bm{\Sigma}}\left(\mathbf{z}\right)=\Lambda_{\mathbf{h}}\left(\mathbf{z}\right)p_{\bm{\vartheta},\bm{\Sigma}}\left(\mathbf{z}\right)\label{eq:density_upper_bound_1}
\end{equation}
where 
\[
\Lambda_{\mathbf{h}}\left(\mathbf{z}\right)=\exp\left(\mathbf{h}^{\top}\bm{\Sigma}^{-1}\left(\mathbf{z}-\bm{\vartheta}\right)-\frac{1}{2}\mathbf{h}^{\top}\bm{\Sigma}^{-1}\mathbf{h}\right)\leq\exp\left(\mathbf{h}^{\top}\bm{\Sigma}^{-1}\left(\mathbf{z}-\bm{\vartheta}\right)\right).
\]
For any $\bm{\Sigma}\in\mathcal{S}$ and $\left\Vert \mathbf{h}\right\Vert \leq\eta$,
by applying Cauchy-Schwarz inequality in the inner product:
\[
\Lambda_{\mathbf{h}}\left(\mathbf{z}\right)\leq\exp\left(\left\Vert \bm{\Sigma}^{-\frac{1}{2}}\mathbf{h}\right\Vert \left\Vert \bm{\Sigma}^{-\frac{1}{2}}\left(\mathbf{z}-\bm{\vartheta}\right)\right\Vert \right)\leq\exp\left(\frac{\eta}{\sqrt{\text{\underbar{\ensuremath{\lambda}}}}}\left\Vert \bm{\Sigma}^{-\frac{1}{2}}\left(\mathbf{z}-\bm{\vartheta}\right)\right\Vert \right),
\]
where the last inequality is because $\text{\underbar{\ensuremath{\lambda}}}\mathbf{I}_{J}\preceq\bm{\Sigma}$
implies $\left\Vert \bm{\Sigma}^{-\frac{1}{2}}\right\Vert \leq\frac{1}{\sqrt{\text{\underbar{\ensuremath{\lambda}}}}}$.
Therefore, 
\begin{equation}
\sup_{\left\Vert \mathbf{h}\right\Vert \leq\eta}\Lambda_{\mathbf{h}}\left(\mathbf{z}\right)\leq\exp\left(\frac{\eta}{\sqrt{\text{\underbar{\ensuremath{\lambda}}}}}\left\Vert \bm{\Sigma}^{-\frac{1}{2}}\left(\mathbf{z}-\bm{\vartheta}\right)\right\Vert \right).\label{eq:density_upper_bound_2}
\end{equation}
With (\ref{eq:density_upper_bound_1}) and (\ref{eq:density_upper_bound_2}),
the first term of (\ref{eq:finalresult_rightbefore}) is bounded further
by:
\begin{align}
 & B\int\sup_{\left\Vert \mathbf{h}\right\Vert \leq\eta}\left\Vert \bm{\delta}_{\bm{\Gamma}}\left(\mathbf{z}\right)-\bm{\delta}_{\bm{\Gamma}'}\left(\mathbf{z}\right)\right\Vert _{1}p_{\bm{\vartheta}+\mathbf{h},\bm{\Sigma}}\left(\mathbf{z}\right)d\mu\left(\mathbf{z}\right)\nonumber \\
= & B\int\sup_{\left\Vert \mathbf{h}\right\Vert \leq\eta}\left\Vert \bm{\delta}_{\bm{\Gamma}}\left(\mathbf{z}\right)-\bm{\delta}_{\bm{\Gamma}'}\left(\mathbf{z}\right)\right\Vert _{1}\Lambda_{\mathbf{h}}\left(\mathbf{z}\right)p_{\bm{\vartheta},\bm{\Sigma}}\left(\mathbf{z}\right)d\mu\left(\mathbf{z}\right)\nonumber \\
\leq & B\mathbb{E}_{P_{\bm{\theta},\bm{\Sigma}}}\left[\left\Vert \bm{\delta}_{\bm{\Gamma}}\left(\mathbf{z}\right)-\bm{\delta}_{\bm{\Gamma}'}\left(\mathbf{z}\right)\right\Vert _{1}\exp\left(\frac{\eta}{\sqrt{\text{\underbar{\ensuremath{\lambda}}}}}\left\Vert \bm{\Sigma}^{-\frac{1}{2}}\left(\mathbf{z}-\bm{\vartheta}\right)\right\Vert \right)\right]\label{eq:firstterm_dini_Gamma1}\\
\leq & B\left(\mathbb{E}_{P_{\bm{\theta},\bm{\Sigma}}}\left[\left\Vert \bm{\delta}_{\bm{\Gamma}}\left(\mathbf{z}\right)-\bm{\delta}_{\bm{\Gamma}'}\left(\mathbf{z}\right)\right\Vert _{1}^{2}\right]\right)^{\frac{1}{2}}\left(\mathbb{E}_{P_{\bm{\theta},\bm{\Sigma}}}\left[\exp\left(2\frac{\eta}{\sqrt{\text{\underbar{\ensuremath{\lambda}}}}}\left\Vert \bm{\Sigma}^{-\frac{1}{2}}\left(\mathbf{z}-\bm{\vartheta}\right)\right\Vert \right)\right]\right)^{\frac{1}{2}}\label{eq:firstterm_dini_Gamma2}\\
= & B\left(\mathbb{E}_{P_{\bm{\theta},\bm{\Sigma}}}\left[\left\Vert \bm{\delta}_{\bm{\Gamma}}\left(\mathbf{z}\right)-\bm{\delta}_{\bm{\Gamma}'}\left(\mathbf{z}\right)\right\Vert _{1}^{2}\right]\right)^{\frac{1}{2}}\left(\mathbb{E}_{P_{\mathbf{0},\mathbf{I}_{J}}}\left[\exp\left(2\frac{\eta}{\sqrt{\text{\underbar{\ensuremath{\lambda}}}}}\left\Vert \mathbf{q}\right\Vert \right)\right]\right)^{\frac{1}{2}},\label{eq:firstterm_dini_Gamma3}
\end{align}
where (\ref{eq:firstterm_dini_Gamma1}) invoked (\ref{eq:density_upper_bound_2}),
(\ref{eq:firstterm_dini_Gamma2}) invoked Cauchy-Schwarz inequality,
(\ref{eq:firstterm_dini_Gamma3}) used $\mathbf{q}\equiv\bm{\Sigma}^{-\frac{1}{2}}\left(\mathbf{z}-\bm{\vartheta}\right)\sim\mathcal{N}\left(0,\mathbf{I}_{J}\right)$.
The second factor of (\ref{eq:firstterm_dini_Gamma3}) depends only
on $\left(\eta,\text{\underbar{\ensuremath{\lambda}}},J\right)$ and
is finite, which we denote by $M$. Because $0\leq\left\Vert \bm{\delta}_{\bm{\Gamma}}\left(\mathbf{z}\right)-\bm{\delta}_{\bm{\Gamma}'}\left(\mathbf{z}\right)\right\Vert _{1}\leq2$,
$\left\Vert \bm{\delta}_{\bm{\Gamma}}\left(\mathbf{z}\right)-\bm{\delta}_{\bm{\Gamma}'}\left(\mathbf{z}\right)\right\Vert _{1}^{2}\leq2\left\Vert \bm{\delta}_{\bm{\Gamma}}\left(\mathbf{z}\right)-\bm{\delta}_{\bm{\Gamma}'}\left(\mathbf{z}\right)\right\Vert _{1}$,
and therefore, (\ref{eq:finalresult_rightbefore}) is bounded above
by:
\begin{align*}
 & \sqrt{2}BM\left(\mathbb{E}_{P_{\bm{\theta},\bm{\Sigma}}}\left[\left\Vert \bm{\delta}_{\bm{\Gamma}}\left(\mathbf{z}\right)-\bm{\delta}_{\bm{\Gamma}'}\left(\mathbf{z}\right)\right\Vert _{1}\right]\right)^{\frac{1}{2}}+B\int\left\Vert \bm{\delta}_{\bm{\Gamma}}-\bm{\delta}_{\bm{\Gamma}'}\right\Vert _{1}dP_{\bm{\vartheta},\bm{\Sigma}}\\
= & \sqrt{2}BM\left\Vert \bm{\delta}_{\bm{\Gamma}}-\bm{\delta}_{\bm{\Gamma}'}\right\Vert _{L^{1}\left(P_{\bm{\theta},\bm{\Sigma}}\right)}^{\frac{1}{2}}+B\left\Vert \bm{\delta}_{\bm{\Gamma}}-\bm{\delta}_{\bm{\Gamma}'}\right\Vert _{L^{1}\left(P_{\bm{\theta},\bm{\Sigma}}\right)}\\
\leq & \sqrt{2}BM\sup_{\bm{\theta}\in\Theta}\left\Vert \bm{\delta}_{\bm{\Gamma}}-\bm{\delta}_{\bm{\Gamma}'}\right\Vert _{L^{1}\left(P_{\bm{\theta},\bm{\Sigma}}\right)}^{\frac{1}{2}}+B\sup_{\bm{\theta}\in\Theta}\left\Vert \bm{\delta}_{\bm{\Gamma}}-\bm{\delta}_{\bm{\Gamma}'}\right\Vert _{L^{1}\left(P_{\bm{\theta},\bm{\Sigma}}\right)}\\
\rightarrow & 0\qquad\text{as }\bm{\Gamma}'\rightarrow\bm{\Gamma},
\end{align*}
where the last step is by invoking Proposition \ref{prop:Prop2} (iii)
for a fixed $\bm{\Sigma}$.

Therefore, (\ref{eq:RHS_F_eta_m}) converges to zero as $\left(\bm{\vartheta}',\bm{\Gamma}'\right)\rightarrow\left(\bm{\vartheta},\bm{\Gamma}\right)$,
implying continuity of $F_{\eta}^{\left(m\right)}\left(\bm{\vartheta},\bm{\Gamma}\right)$
in $\left(\bm{\vartheta},\bm{\Gamma}\right)$.

\uline{(3) (Applying Dini's theorem)}: (1) and (2) are sufficient
for Dini's theorem to apply. Therefore. 
\[
\sup_{\left(\bm{\vartheta},\bm{\Gamma}\right)\in\Theta\times\mathcal{S}}F_{\eta}^{\left(m\right)}\left(\bm{\vartheta},\bm{\Gamma}\right)\downarrow0\qquad\text{as }\eta\downarrow0.
\]
The remainder of the proof logic is identical to the proof of Theorem
\ref{thm:(Uniform-Convergence-of}.

\subsection{Proof of Theorem \ref{thm:(Plug-in-of-the}}\label{subsec:Proof-of-Theorem42}

Take the triangle inequality decomposition:
\begin{eqnarray}
 &  & \left|\sup_{\bm{\vartheta}\in\Theta}R_{\left[n\right]}\left(\bm{\vartheta},\bm{\delta}_{\hat{\bm{\Sigma}}}\right)-\sup_{\bm{\vartheta}\in\Theta}R_{\bm{\Sigma}}\left(\bm{\vartheta},\bm{\delta}_{\bm{\Sigma}}\right)\right|\nonumber \\
 & \leq & \sup_{\bm{\vartheta}\in\Theta}\left|R_{\left[n\right]}\left(\bm{\vartheta},\bm{\delta}_{\hat{\bm{\Sigma}}}\right)-R_{\bm{\Sigma}}\left(\bm{\vartheta},\bm{\delta}_{\bm{\Sigma}}\right)\right|\nonumber \\
 & \leq & \sup_{\bm{\vartheta}\in\Theta}\left|R_{\left[n\right]}\left(\bm{\vartheta},\bm{\delta}_{\hat{\bm{\Sigma}}}\right)-R_{\bm{\Sigma}}\left(\bm{\vartheta},\bm{\delta}_{\hat{\bm{\Sigma}}}\right)\right|+\sup_{\bm{\vartheta}\in\Theta}\left|R_{\bm{\Sigma}}\left(\bm{\vartheta},\bm{\delta}_{\hat{\bm{\Sigma}}}\right)-R_{\bm{\Sigma}}\left(\bm{\vartheta},\bm{\delta}_{\bm{\Sigma}}\right)\right|.\label{eq:final_triangle_decomposition}
\end{eqnarray}
The first term of (\ref{eq:final_triangle_decomposition}) is bounded
by:
\begin{align*}
\sup_{\bm{\vartheta}\in\Theta}\left|R_{\left[n\right]}\left(\bm{\vartheta},\bm{\delta}_{\hat{\bm{\Sigma}}}\right)-R_{\bm{\Sigma}}\left(\bm{\vartheta},\bm{\delta}_{\hat{\bm{\Sigma}}}\right)\right| & \leq\sup_{\bm{\Gamma}\in\mathcal{S}}\sup_{\bm{\vartheta}\in\Theta}\left|R_{\left[n\right]}\left(\bm{\vartheta},\bm{\delta}_{\bm{\Gamma}}\right)-R_{\bm{\Sigma}}\left(\bm{\vartheta},\bm{\delta}_{\bm{\Gamma}}\right)\right|\\
 & =o_{p}\left(1\right)\qquad\text{as }n\rightarrow\infty
\end{align*}
by applying Lemma \ref{lem:41} to the random sequence $\hat{\bm{\Sigma}}\in\mathcal{S}$
w.p.a.1. The second term of (\ref{eq:final_triangle_decomposition})
is bounded by:
\begin{align*}
\sup_{\bm{\vartheta}\in\Theta}\left|R_{\bm{\Sigma}}\left(\bm{\vartheta},\bm{\delta}_{\hat{\bm{\Sigma}}}\right)-R_{\bm{\Sigma}}\left(\bm{\vartheta},\bm{\delta}_{\bm{\Sigma}}\right)\right| & \leq\sup_{\bm{\vartheta}\in\Theta}\int\sum_{k=1}^{J}\left|\vartheta_{k}\right|\left|\delta_{\hat{\bm{\Sigma}},k}\left(\mathbf{y}\right)-\delta_{\bm{\Sigma},k}\left(\mathbf{y}\right)\right|dP_{\bm{\vartheta},\bm{\Sigma}}\left(\mathbf{y}\right)\\
 & \leq B\sup_{\bm{\vartheta}\in\Theta}\left\Vert \bm{\delta}_{\hat{\bm{\Sigma}}}-\bm{\delta}_{\bm{\Sigma}}\right\Vert _{L^{1}\left(P_{\bm{\vartheta},\bm{\Sigma}}\right)}\\
 & =o_{p}\left(1\right)\qquad\text{as }n\rightarrow\infty,
\end{align*}
where we invoked Proposition \ref{prop:Prop2} (iii) for a fixed $\bm{\Sigma}$.
The conclusion follows.

\newpage{}

\section{Other Results Used in the Proofs}\label{subsec:Theorems-Used-in}

For completeness, we state the two theorems used in the proof of the
main results. The following Theorem \ref{thm:(Farkas's-Alternative-Linear}
is taken from \citet{Border2013}. 
\begin{thm}[Farkas' Alternative Linear Inequality; Theorem 12 of \citet{Border2013}]
\label{thm:(Farkas's-Alternative-Linear}Let $\mathbf{A}$ be an
$m\times n$ real matrix, $\mathbf{B}$ be $l\times n$ matrix, $\mathbf{b}\in\mathbb{R}^{m}$
and $\mathbf{c}\in\mathbb{R}^{l}$. Then, exactly one of the following
alternatives hold:

(i) There exists $\mathbf{x}\in\mathbb{R}^{n}$ satisfying
\begin{align*}
\mathbf{A}\mathbf{x} & =\mathbf{b},\quad\mathbf{B}\mathbf{x}\leq\mathbf{c},\quad\mathbf{x}\geq\mathbf{0}.
\end{align*}

(ii) There exists $\mathbf{p}\in\mathbb{R}^{m}$ and $\mathbf{q}\in\mathbb{R}^{l}$
satisfying 
\[
\mathbf{p}\mathbf{A}+\mathbf{q}\mathbf{B}\geq\mathbf{0},\quad\mathbf{q}\geq\mathbf{0},\quad\mathbf{p}^{\top}\mathbf{b}+\mathbf{q}^{\top}\mathbf{c}<0.
\]
\end{thm}
The following Strassen's coupling Theorem \ref{thm:(Strassen-Prokhorov/Strassen-Dud}
is taken from \citet[Sec 10.3.]{pollard2002user}, used in the proof
of Theorem \ref{thm:(Uniform-Convergence-of}.
\begin{thm}[{Strassen's Coupling, \citet[Sec 10.3., Theorem <8>]{pollard2002user}}]
\label{thm:(Strassen-Prokhorov/Strassen-Dud}Let $P$ and $Q$ be
tight probability measures on the Borel sigma field $\mathfrak{B}$
of a separable metric space $\mathcal{X}$. Let $\epsilon,\epsilon'$
be positive constants. There exists random elements $X$ and $Y$
of $\mathcal{X}$ with distributions $P$ and $Q$ such that $\Pr\left(d\left(X,Y\right)>\epsilon\right)\leq\epsilon'$
if and only if $P\left(B\right)\leq Q\left(B^{\epsilon}\right)+\epsilon'$
for all Borel sets $B$, where $B^{\epsilon}$ denotes the $\epsilon$-enlargement
of $B$.
\end{thm}
\begin{proof}
See, e.g., \citet[Sec 10.3.]{pollard2002user}.
\end{proof}
\newpage{}

\section{Details of Numerical Implementation}\label{sec:Details-of-Numerical}

\subsection{The Bayes Risk Under the Minimax/Maximin Rule}

The objective is to maximize the Bayes risk under the minimax/maximin
rule $\bm{\delta}$:
\begin{align}
 & \max_{\left(\bm{\theta}^{k},\pi^{k}\right)_{k\in\mathcal{J}}}\sum_{k=1}^{J}\pi^{k}R\left(\bm{\theta}^{k},\bm{\delta}\right)=\sum_{k=1}^{J}\pi^{k}\left(\theta_{k}^{k}-\sum_{i=1}^{J}\theta_{i}^{k}\mathbb{E}_{\bm{\theta}^{k}}\left[\delta_{i}\left(\mathbf{x}\right)\right]\right)\nonumber \\
= & \max_{\left(\bm{\theta}^{k},\pi^{k}\right)_{k\in\mathcal{J}}}\sum_{k=1}^{J}\pi^{k}\left[\theta_{k}^{k}-\sum_{i=1}^{J}\theta_{i}^{k}\int\mathbf{1}\left(\sum_{k=1}^{J}\theta_{i}^{k}\pi^{k}p_{\bm{\theta}^{k}}\left(\mathbf{x}\right)>\sum_{k=1}^{J}\theta_{j}^{k}\pi^{k}p_{\bm{\theta}^{k}}\left(\mathbf{x}\right)\:\forall j\neq i\right)p_{\bm{\theta}^{k}}\left(\mathbf{x}\right)d\mu\left(\mathbf{x}\right)\right].\label{eq:maximization_problem}
\end{align}
We also impose the constraints $\bm{\theta}^{k}\in\Theta^{k}$ during
the optimization.

Throughout, we approximate $\mathbb{E}_{\bm{\theta}^{k}}\left[\delta_{i}\left(\mathbf{x}\right)\right]$
with respect to the Normal density $p_{\bm{\theta}^{k}}\left(\mathbf{x}\right)$
using the quasi Monte Carlo (QMC) draws. Specifically, we draw $2^{19}$
$\mathbf{z}\sim\mathcal{N}\left(\mathbf{0},\mathbf{I}\right)$ based
on the Sobol sequence, and then take the location-scale transformation
$\mathbf{x}=\bm{\theta}^{k}+\sqrt{\bm{\Sigma}}\mathbf{z}$. This yields
the numerical optimization problem:
\[
\max_{\left(\bm{\theta}^{k},\pi^{k}\right)_{k\in\mathcal{J}}}\sum_{k=1}^{J}\pi^{k}\left[\theta_{k}^{k}-\sum_{i=1}^{J}\theta_{i}^{k}\left(\frac{1}{N_{s}}\sum_{s=1}^{N_{s}}\delta_{i}^{\tau}\left(\mathbf{x}^{s};\bm{\theta}^{k},\pi^{k}\right)\right)\right],
\]
where $s$ is the index for the QMC draw. 

\subsection{Softmax Approximation of the Objective Function With Analytic Gradients}

When the problem dimension is relatively large, derivative-free algorithms
struggle to find the optimum. Therefore, we employ smooth ``softmax''
approximation of the indicator function and provide analytic gradients
of the objective function to use derivative-based optimization algorithms
such as BFGS and its variants.

\subsubsection{Objective Function}

The objective function (\ref{eq:maximization_problem}) is non-smooth,
thereby preventing the use of gradient-based optimizers. Therefore,
we approximate the indicator with the softmax (logit) function as
follows:
\begin{align*}
h_{i}\left(\mathbf{x};\bm{\theta}^{k},\pi^{k}\right) & =\sum_{k=1}^{J}\theta_{i}^{k}\pi^{k}p_{\bm{\theta}^{k}}\left(\mathbf{x}\right)\\
\delta_{i}^{\tau}\left(\mathbf{x};\bm{\theta}^{k},\pi^{k}\right) & =\frac{\exp\left(\tau h_{i}\left(\mathbf{x};\bm{\theta}^{k},\pi^{k}\right)\right)}{\sum_{j\in\mathcal{J}}\exp\left(\tau h_{j}\left(\mathbf{x};\bm{\theta}^{k},\pi^{k}\right)\right)}.
\end{align*}
$\tau\uparrow\infty$ yields the original ``hard-max'' indicator
function. A large enough $\tau$ makes $\bm{\delta}^{\tau}\approx\bm{\delta}$
while it is still smooth. During the implementation, we take $\tau=100,000$.

\subsubsection{Gradients}

We invoke a series of chain rules to calculate the objective function
gradients. For the gradients, we introduce additional indices $m$
and $r$.

\paragraph{Component Gradients}

\[
\frac{\partial\delta_{i}^{\tau}}{\partial h_{m}}=\tau\delta_{i}^{\tau}\left(1_{\left\{ i=m\right\} }-\delta_{m}^{\tau}\right)
\]
\begin{align*}
\frac{\partial h_{i}}{\partial\theta_{u}^{r}} & =\pi^{r}p^{r}\left(\mathbf{x}\right)\left(1_{\left\{ i=u\right\} }+\theta_{i}^{r}\left[\bm{\Sigma}^{-1}\left(\mathbf{x}-\bm{\theta}^{r}\right)\right]_{u}\right)\\
\frac{\partial h_{i}}{\partial\pi^{r}} & =\theta_{i}^{r}p^{r}\left(\mathbf{x}\right)
\end{align*}

\paragraph*{Gradient of the Smoothened Pointwise Bayes Risk for a Fixed Support
Point $\bm{\theta}^{k}$}

Let 
\begin{align*}
G_{ik} & =\frac{1}{N_{s}}\sum_{s=1}^{N_{s}}\delta_{i}^{\tau}\left(\mathbf{x}^{s};\bm{\theta}^{k},\pi^{k}\right)\;\text{and}\;R_{k}=\theta_{k}^{k}-\sum_{i=1}^{J}\theta_{i}^{k}G_{ik}.
\end{align*}
Then, 
\begin{align*}
\frac{\partial G_{ik}}{\partial\theta_{u}^{r}} & =\frac{1}{N_{s}}\sum_{s=1}^{N_{s}}\left(\sum_{m=1}^{J}\frac{\partial\delta_{i}^{\tau}}{\partial h_{m}}\left(\mathbf{x}^{s}\right)\frac{\partial h_{m}}{\partial\theta_{u}^{r}}\left(\mathbf{x}^{s}\right)+\sum_{m=1}^{J}\frac{\partial\delta_{i}^{\tau}}{\partial x_{m}^{s}}\left(\mathbf{x}^{s}\right)\frac{\partial x_{m}^{s}}{\partial\theta_{u}^{r}}\right)\\
\frac{\partial R_{k}}{\partial\theta_{u}^{r}} & =1_{\left\{ r=k\right\} }1_{\left\{ u=k\right\} }-\sum_{i=1}^{J}\left(1_{\left\{ r=k\right\} }1_{\left\{ u=i\right\} }G_{ik}+\theta_{i}^{k}\frac{\partial G_{ik}}{\partial\theta_{u}^{r}}\right)\\
\frac{\partial G_{ik}}{\partial\pi^{r}} & =\frac{1}{N_{s}}\sum_{s=1}^{N_{s}}\left(\sum_{m=1}^{J}\frac{\partial\delta_{i}^{\tau}}{\partial h_{m}}\left(\mathbf{x}^{s}\right)\frac{\partial h_{m}}{\partial\pi^{r}}\left(\mathbf{x}^{s}\right)\right)\\
\frac{\partial R_{k}}{\partial\pi^{r}} & =-\sum_{i=1}^{J}\theta_{i}^{k}\frac{\partial G_{ik}}{\partial\pi^{r}},
\end{align*}
where
\begin{align*}
\frac{\partial x_{m}^{s}}{\partial\theta_{u}^{r}} & =1_{\left\{ r=k\right\} }1_{\left\{ u=m\right\} }\\
\frac{\partial\delta_{i}^{\tau}}{\partial x_{m}^{s}} & =\sum_{l=1}^{J}\frac{\partial\delta_{i}^{\tau}}{\partial h_{l}}\frac{\partial h_{l}}{\partial x_{m}^{s}}\\
 & =-\tau\delta_{i}^{\tau}\left(\left[\sum_{r=1}^{J}\theta_{i}^{r}\pi^{r}p^{r}\left(\mathbf{x}^{s}\right)\bm{\Sigma}^{-1}\left(\mathbf{x}^{s}-\bm{\theta}^{r}\right)\right]_{m}-\left[\sum_{l=1}^{J}\delta_{l}^{\tau}\left(\mathbf{x}^{s}\right)\left(\sum_{r=1}^{J}\theta_{l}^{r}\pi^{r}p^{r}\left(\mathbf{x}^{s}\right)\bm{\Sigma}^{-1}\left(\mathbf{x}^{s}-\bm{\theta}^{r}\right)\right)\right]_{m}\right).
\end{align*}

\paragraph{Objective Function Gradient}

Let $f=\sum_{k=1}^{J}\pi^{k}R\left(\bm{\theta}^{k},\bm{\delta}\right)$
(IPOPT minimizes). Then:
\begin{align*}
\frac{\partial f}{\partial\theta_{u}^{r}} & =\sum_{k=1}^{J}\pi^{k}\frac{\partial R_{k}}{\partial\theta_{u}^{r}}\\
\frac{\partial f}{\partial\pi^{r}} & =R_{r}+\sum_{k=1}^{J}\pi^{k}\frac{\partial R_{k}}{\partial\pi^{r}}.
\end{align*}

\end{document}